\documentclass[11pt]{article}

\usepackage[utf8]{inputenc}
\usepackage{amsmath}
\usepackage{amssymb,amsfonts,amsthm}
\usepackage{mathrsfs}
\usepackage{mathtools}
\usepackage{algorithm}
\usepackage{algorithmic}
\usepackage{xcolor}
\usepackage{tikz}
\usepackage{fullpage}
\usepackage{dsfont}
\usepackage{thmtools}
\usepackage{thm-restate}
\usepackage{enumitem}
\usepackage{graphicx}
\usepackage{sidecap}
\usepackage{bbm}
\usepackage{array}
\newcolumntype{M}[1]{>{\centering\arraybackslash}m{#1}}

\usepackage{accents}

\usepackage{microtype}      
\usepackage{libertine}
\usepackage{libertinust1math}

\usepackage{hyperref}
\usepackage{xcolor}
\definecolor{maroon}{rgb}{0.5, 0.0, 0.0}
\hypersetup{
  colorlinks   = true, 
  urlcolor     = blue, 
  linkcolor    = blue, 
  citecolor   = maroon 
}
\usepackage[capitalize, nameinlink]{cleveref}

\usepackage{tikz}
\usetikzlibrary{decorations.pathreplacing}
\usetikzlibrary{arrows,positioning}

\usepackage{comment}
\DeclareRobustCommand\shape{
 \lower5pt\hbox{
 \hskip-7pt
  \tikzset{circ/.style={circle, draw, fill=black, scale=.15}}
  \begin{tikzpicture}[semithick,scale=.3]
  \node (l1) at (0,.5) [circ]{};
  \node (l3) at (0.5,0.3) [circ]{};
  \draw[-] (l1) to node [auto] {} (l3);
    \end{tikzpicture}
  \hskip-8pt}
}

\newcommand{\lpn}{\textsf{LPN}}
\newcommand{\lpnspace}{\textsf{LPN }}
\newcommand{\lwe}{\textsf{LWE}}
\newcommand{\clwe}{\textsf{CLWE}}
\newcommand{\lwespace}{\textsf{LWE }}
\newcommand{\klpn}{$k$-\textsf{LPN}}
\newcommand{\klpnspace}{$k$-\textsf{LPN }}
\newcommand{\klwe}{$k$-\textsf{LWE}}
\newcommand{\kxor}{$k$-\textsf{XOR}}

\newcommand{\zqmult}{((\mathbb{Z}/q\mathbb{Z})^\times)}
\newcommand{\zqz}{\mathbb{Z}/q\mathbb{Z}}

\makeatletter
\newcommand{\mylabel}[2]{#2\def\@currentlabel{#2}\label{#1}}
\makeatother

\newcommand{\minsf}{\mathsf{min}}

\newcommand{\bern}{\mathsf{Bern}}
\newcommand{\bernoulli}{\bern}

\newcommand{\unif}{\mathsf{Unif}}
\newcommand{\randomsupport}{\mathsf{randomsupport}}
\newcommand{\support}{\mathsf{support}}
\newcommand{\errordistribution}{\mathsf{errordistribution}}
\newcommand{\sampledistribution}{\mathsf{sampledistribution}}
\newcommand{\modulus}{\mathsf{modulus}}
\newcommand{\samples}{\mathsf{samples}}

\newcommand{\hzero}{\mathsf{H}_0}
\newcommand{\hone}{\mathsf{H}_1}

\newtheorem{theorem}{Theorem}[section]

\newtheorem*{theorem*}{Theorem}

\newtheorem{lemma}{Lemma}[section]
\newtheorem*{lemma*}{Lemma}

\newtheorem{corollary}[lemma]{Corollary}
\newtheorem{observation}[lemma]{Observation}
\newtheorem{parameters}{Parameter Relation}[]
\newtheorem*{parameters*}{Parameter Relation}
\newtheorem{claim}[lemma]{Claim}

\newtheorem{proposition}[theorem]{Proposition}
\newtheorem*{proposition*}{Proposition}
\theoremstyle{definition}
\newtheorem{remark}{Remark}[]
\newtheorem{definition}{Definition}
\newtheorem{problem}{Problem}[]

\newcommand{\var}{\mathsf{Var}}

\DeclareMathOperator*{\E}{{\rm I}\kern-0.18em{\rm E}}

\newcommand{\prob}{\Pr}
\newcommand{\expect}{\E}
\renewcommand{\Pr}{\,{\rm I}\kern-0.18em{\rm P}}

\newcommand{\gb}[1]{{{\color{magenta}\texttt{[Guy: #1]}}}}

\newcommand{\Dcoeff}{\mathsf{D}^{\mathrm{coeff}}}
\newcommand{\Derr}{\mathsf{D}^{\mathrm{error}}}

\newcommand{\cX}{\mathcal{X}}

\newcommand{\cC}{\mathcal{C}}
\newcommand{\cD}{\mathcal{D}}

\newcommand{\cT}{\mathcal{T}}

\newcommand{\norm}[1]{\|#1\|}

\newcommand{\polylog}{\mathsf{polylog}}

\newcommand{\abs}[1]{\lvert #1 \rvert}

\newcommand{\iidsim}{\stackrel{\mathrm{i.i.d.}}\sim}

\newcommand{\Preimage}{\mathsf{Preimage}}

\newcommand{\indicator}{\mathds{1}}

\definecolor{todored}{HTML}{AD0029}
\newcommand{\stefanc}[1]{{\textcolor{todored}{\texttt{[Stefan:  #1]} }}}

\newcommand{\LWE}{\mathsf{LWE}}
\newcommand{\LPN}{\mathsf{LPN}}
\newcommand{\NLE}{\mathsf{NLE}}

\newcommand{\poly}{\mathsf{poly}}

\newcommand{\Bern}{\mathsf{Bernoulli}}

\newcommand{\dec}{\mathsf{Dec}}
\newcommand{\fail}{\mathsf{FAIL}}

\newcommand\MYcurrentlabel{xxx}
\newcommand{\MYstore}[2]{%
  \global\expandafter \def \csname MYMEMORY #1 \endcsname{#2}%
}
\newcommand{\MYload}[1]{%
  \csname MYMEMORY #1 \endcsname%
}
\newcommand{\MYnewlabel}[1]{%
  \renewcommand\MYcurrentlabel{#1}%
  \MYoldlabel{#1}%
}
\newcommand{\MYdummylabel}[1]{}
\newcommand{\torestate}[1]{%
  \let\MYoldlabel\label%
  \let\label\MYnewlabel%
  #1%
  \MYstore{\MYcurrentlabel}{#1}%
  \let\label\MYoldlabel%
}
\newcommand{\restatetheorem}[1]{%
  \let\MYoldlabel\label
  \let\label\MYdummylabel
  \begin{theorem*}[Restatement of \cref{#1}]
    \MYload{#1}
  \end{theorem*}
  \let\label\MYoldlabel
}
\newcommand{\restatelemma}[1]{%
  \let\MYoldlabel\label
  \let\label\MYdummylabel
  \begin{lemma*}[Restatement of \cref{#1}]
    \MYload{#1}
  \end{lemma*}
  \let\label\MYoldlabel
}
\newcommand{\restateprop}[1]{%
  \let\MYoldlabel\label
  \let\label\MYdummylabel
  \begin{proposition*}[Restatement of \cref{#1}]
    \MYload{#1}
  \end{proposition*}
  \let\label\MYoldlabel
}
\newcommand{\restatefact}[1]{%
  \let\MYoldlabel\label
  \let\label\MYdummylabel
  \begin{fact*}[Restatement of \cref{#1}]
    \MYload{#1}
  \end{fact*}
  \let\label\MYoldlabel
}
\newcommand{\restateparameters}[1]{%
  \let\MYoldlabel\label
  \let\label\MYdummylabel
  \begin{parameters*}[Restatement of \cref{#1}]
    \MYload{#1}
  \end{parameters*}
  \let\label\MYoldlabel
}
\newcommand{\restate}[1]{%
  \let\MYoldlabel\label
  \let\label\MYdummylabel
  \MYload{#1}
  \let\label\MYoldlabel
}

\usepackage[backend=biber,
backref=true,
style=alphabetic,
minalphanames=3,
maxalphanames=3,
maxnames=50
]{biblatex}
\newcommand{\set}[1]{\{#1\}}
\newcommand{\Set}[1]{\Big\{#1\Big\}}
\newcommand{\suchthat}{\;\Big\vert\;}

\newcommand{\paren}[1]{(#1)}
\newcommand{\Paren}[1]{\left(#1\right)}

\newcommand{\Brac}[1]{\left[#1\right]}

\newcommand{\iprod}[1]{\langle#1\rangle}

\newcommand{\OPT}{\mathsf{OPT}}
\newcommand{\iid}{i.i.d.\ }
\newcommand{\Tsample}{\cT_{\mathsf{sample}}}

\usepackage[titles]{tocloft}
\title{
Near-Optimal Time-Sparsity Trade-Offs for 
\\ Solving Noisy Linear Equations
}
\author{Kiril Bangachev\thanks{Supported by a Siebel Scholarship.}\\ MIT \and Guy Bresler\thanks{Supported by NSF Career Award CCF-1940205.}\\MIT \and  Stefan Tiegel\thanks{Supported by the European Union’s Horizon research and innovation programme (grant agreement no. 815464).}\\ETH Z\"urich \and Vinod Vaikuntanathan\thanks{Supported by NSF CNS-2154149 and a Simons Investigator Award.}\\MIT}

\begin{document}
\maketitle
\pagenumbering{gobble}
\begin{abstract}
    We present a polynomial-time reduction from solving noisy linear equations over $\zqz$ in dimension $\Theta(k\log n/\poly(\log k,\log q,\log\log n))$ with a uniformly random coefficient matrix to noisy linear equations over $\zqz$ in dimension $n$ where each row of the coefficient matrix has uniformly random support of size $k$. This allows us to deduce the hardness of sparse problems from their dense counterparts. In particular, we derive hardness results in the following canonical settings:
\begin{enumerate}
    \item Assuming the $\ell$-dimensional (dense) learning with errors (\lwe) problem over a polynomial-size field takes time $2^{\Omega(\ell)}$, $k$-sparse \lwespace in dimension $n$ takes time $$n^{\Omega({k}/{(\log k \cdot  (\log k + \log \log n))})}~.$$
    \item Assuming the $\ell$-dimensional (dense) learning parity with noise (\lpn) problem over $\mathbb{F}_2$ takes time $2^{\Omega(\ell/\log \ell)}$, $k$-sparse \lpnspace in dimension $n$ takes time $$n^{\Omega(k/(\log k \cdot (\log k + \log \log n)^2))}~.$$
\end{enumerate}
These running time lower bounds are nearly tight as both sparse problems can be solved in time $n^{O(k)},$ given sufficiently many samples.

Our reduction allows us to derive several consequences in cryptography and the computational complexity of statistical problems. 
In addition, as a new application, we give a reduction from $k$-sparse LWE to noisy tensor completion. Concretely, composing the two reductions implies that order-$k$ rank-$2^{k-1}$ noisy tensor completion in $\mathbb{R}^{n^{\otimes k}}$ takes time $n^{\Omega(k/ \log k \cdot (\log k + \log \log n))}$, assuming the exponential hardness of standard worst-case lattice problems.


Our reduction is (nearly) lossless and extremely versatile: it preserves the number of samples up to a $1-o(1)$ multiplicative factor with high probability and it preserves the noise distribution.
In particular, it also applies to the learning with rounding problem.
The same reduction works for a wide range of ring sizes from $q = 2$ to $q$ that is super-polynomial in the initial dimension $\ell$ and for varying support sizes between samples. Finally, our reduction is compatible with decision (testing), search, and strong refutation and, hence, reduces hardness in the sparse setting to that of the standard, dense, setting for all three tasks.
\end{abstract}

\newpage
\setcounter{tocdepth}{2}
\tableofcontents
\newpage
\pagenumbering{arabic}



\newpage

\section{Introduction}
\label{sec:intro_experiment}

\def\matA{{A}}
\def\vecy{{y}}
\def\vecs{{s}}
\def\vece{{e}}

Random noisy linear equations over $\zqz$ play a fundamental role in average-case complexity, cryptography, and learning theory.
Two canonical instances are the 
\emph{learning parity with noise} (\lpn) and the \emph{learning with errors} (\lwe) problems, that we will describe shortly.

Noisy linear equation problems are defined by a distribution over pairs $(\matA,\vecy) \in (\zqz)^{m \times n} \times (\zqz)^m$, where
$\matA$ is the \emph{public coefficient matrix} and $\vecy$ the \emph{labels}.
The rows of {the public matrix} $\matA$ are often taken to be independent and identically distributed, in which case $m$ is viewed as the number of independent examples in supervised learning. 
The basic goal is to determine whether there is an approximate linear relation between $\matA$ and $\vecy$, 
i.e., does there exist an $\vecs \in (\zqz)^n$ such that $\vecy \approx \matA\vecs \bmod{q}$. 
Three different tasks have received attention:
\begin{itemize}
    \item[] \textbf{(Decision)} Distinguish whether $\matA$ and $\vecy$ are independent or $\vecy \approx \matA\vecs$ for some $\vecs$.
    \item[] \textbf{(Search)} Under the promise that $\vecy \approx \matA\vecs$ for some $\vecs$, recover $\vecs$.
    \item[] \textbf{(Refutation)} Under the promise that $\matA$ and $\vecy$ are independent, certify (i.e., prove) that $\vecy$ is not close to $\matA\vecs$ for \emph{any} $\vecs$.
\end{itemize}
For the search and decision tasks, $\vecy \approx \matA\vecs$ means that $\vecy = \matA\vecs + \vece$ for some ``small'' error vector $\vece$.
For the refutation task, one is most often interested in certifying an upper bound on the maximum number of equations that are approximately satisfiable simultaneously. Importantly, the upper bound should be correct \emph{for all} inputs and 
``significantly better than trivial'' on the average. (We refer to~\cref{sec:tasks} for rigorous definitions of the three tasks.)
We are interested in how many equations (i.e., how large $m$ needs to be) and how much computation time is necessary to solve each of these tasks.
%

The setup described above is
nearly identical
to that of standard linear regression, studied extensively in statistics, for which a variety of efficient algorithms can achieve nearly optimal sample complexity under certain assumptions on $\matA$. 
The primary difference lies in the 
modular arithmetic, which in the presence of even small errors makes the above tasks algorithmically challenging. 

Depending on the specific choice of $q$ and distributions of $\matA,\vecs$, and $\vece$, researchers have put forth conjectures on how much time and equations are necessary to solve the problem. 
    Such \emph{hardness conjectures} 
   form the basis of numerous cryptographic primitives 
    and also have implications towards understanding the computational 
    complexities of other problems.
    One of the goals of average-case complexity theory is to simplify the landscape and
   relate the computational complexities of different problems. Can one show that the various hardness conjectures follow from just one of them?
   

In this paper, we relate two canonical settings of noisy linear equations.
We reduce the standard (dense) setting, in which each row of the public matrix $\matA$ is independent and uniform over $(\zqz)^L$, to the
$k$-sparse setting, in which each row  of the public matrix $\matA$ is $n$-dimensional and has an independent uniformly random support of size $k$.
In both cases, the secret $\vecs$ is uniform over $(\zqz)^L$, respectively $(\zqz)^n$. For concreteness, we state our theorem below in terms of \lpn{} and \lwe, which will be defined shortly.

\begin{theorem}[Informal, see \cref{thm:mainreduction} 
and \cref{cor:explwe,cor:nearexplpn,cor:subexphardness,cor:quasihardness}]
\label{thm:maininformal}
    There is a $\poly(n)$-time reduction from standard \lpn{} (respectively standard \lwe) in dimension $$L = \frac{k\log n}{\poly(\log q,\log k, \log\log n)}$$ to \klpn{} (respectively \klwe) in dimension $n$.
    The reduction preserves the number of samples up to a $1-o(1)$ multiplicative factor.
    Furthermore, for decision and search, the reduction preserves the noise distribution\footnote{In particular, it also to the learning with rounding problem.} and for refutation it preserves the strength of the refutation algorithm up to an additive $o(1)$ term.
\end{theorem}

All four settings of $\{\text{dense, $k$-sparse}\},$ $\{\text{\lpn, \lwe}\}$ 
are common in the literature and have found important applications to cryptography, statistics, and learning theory.
Yet, the relationship between the standard dense public matrix setting versus the sparse public matrix setting has remained elusive so far (although, see~\cref{prelim:sparselpn,prelim:sparselwe} for a discussion of prior work that made progress in this direction). 

This setting should not be confused with the setting in which the rows of the coefficient matrix are dense but the secret is sparse.
The relation of the sparse secrete versions to the standard (dense) settings is much better understood; we refer to~\cref{sec:otherworks} for a more in-depth discussion. 



\subsection{LPN and LWE}

In \lpn, the ring is binary, i.e. $q = 2$,
and the entries of the error vector are independently distributed as $\bern(\delta)$, where $\delta$ is a noise parameter (small $\delta$ yields a sparse error vector, with most entries equal to zero). 
Some works also consider \lpnspace over large rings but again with sparse noise~\cite{ishai2009secure,jain2021indistinguishability,jain2022indistinguishability,gibbs24somewhathomomorphic}. Unless otherwise stated, we will understand that $q = 2$ when we refer to \lpn.
In \lwe, the modulus $q$ is ``large'' (typically at least polynomial in $L$)
and the noise is dense (most commonly drawn from a discrete Gaussian distribution of standard deviation at least $2\sqrt{L},$ in which case a reduction from worst-case lattice problems is known \cite{regev05LWE}).

\subsubsection{The Standard Setting: Uniformly Random (Dense) Public Matrix}
In the standard variants 
of both \lpnspace and \lwe, the public matrix $A$ and secret $s$ are drawn independently and uniformly over $(\zqz)^L$ and, hence, both are dense. 
Under suitable choices of the noise parameters, the decision versions of both problems are widely believed to be nearly exponentially hard.\footnote{Note that search is as least as hard as decision, and so is the natural refutation task, as discussed in~\cref{sec:tasks}.}
In particular,\footnote{Oftentimes in statistical problems, there is a trade-off between the number of samples required and the running time of the algorithm. In this work, we mostly focus on a running-time lower bound. In that case, the lower bound holds regardless of the number of samples used. 
} 
\begin{quote}\begin{center}
       the fastest known algorithm for \lpnspace runs in time $2^{\Omega(L/\log L)}$~\cite{BKW}, 
\end{center}
\end{quote} and any improvement to this would be considered a breakthrough.


For \lwe, several algorithms improve on the most na\"{i}ve brute-force enumeration over the $q^L$ possible secrets, but for large enough $q$ and noise rate,
\begin{quote}\begin{center}
        the fastest known algorithm for \lwespace takes time $2^{\Omega(L)}$.
\end{center}
\end{quote} There is compelling evidence that this cannot be improved upon, as such an algorithm would yield a $2^{o(L)}$-time (quantum) algorithm for several \emph{worst-case} lattice problems (e.g. the approximate shortest vector problem with polynomial approximation factors) believed to require exponential time; see~\cite{regev05LWE,DBLP:conf/stoc/AggarwalDRS15,DBLP:journals/corr/abs-2211-11693} and the references therein.

The assumed near-exponential hardness in conjunction with the simple linear structure make both standard \lpnspace and \lwespace (and variants of \lpnspace and \lwespace known to inherit their hardness from the standard variants) extremely useful for a variety of applications, which we discuss in \cref{prelim:lpn,prelim:lwe}.


\subsubsection{The Setting of a Sparse Public Matrix}

In many applications, additional structure of the noisy linear equations is very useful.
Starting with the work of~\cite{alekhovich03averagecasevsapproximation}, researchers have considered 
noisy linear equations in which each row $a_i$ (in dimension $n$) of the public matrix
$A$ is $k$-sparse, i.e., has $k$ non-zero entries. Typically, the locations of the supports are uniform and independent between samples.
We denote the respective \lpn{} and \lwe{} variants by \klpn{} and \klwe{}\footnote{We note that ``\klwe'' does not refer to a unique model as different distributions for the values on the support are possible when 
$q>2.$ The results in the current work apply to \emph{any} distribution on the support. In~\cite{jain24sparseLWE}, the authors only consider the uniform distribution on the support. 
}.
One reason the sparse public matrix variant is useful is that the projections $a_i^\top s$ of $s$ on $a_i$ capture local interactions of the coordinates of the secret $s$ --- each $a_i^\top s$ depends only on $k$ of the coordinates of $s$. 
Hence, one can use \klpn{} and \klwe{} to argue about problems related to learning from ``local'' structure, which are ubiquitous in computer science. 

For both \klpn{} and \klwe{}, there is a simple brute-force algorithm that succeeds (for decision/search) in time $n^{O(k)}$ given $n^{1.01 k}$ equations. With access to that many random $k$-sparse equations, each of the possible $\binom{n}{k}$ supports will appear $n^{\Omega(k)}$ times. Restricting to a given support of size $k,$ one obtains a $k$-dimensional problem which can be solved in time $2^{O(k)}$ via the respective algorithms for dense $k$-dimensional \lpn/\lwe. Many other $n^{O(k)}$-time algorithms are known provided a sufficient number of samples, for example via gradient-descent dynamics in neural networks \cite{barak22hidden}.
Perhaps surprisingly, for small values of $k,$ 
no better running time than the naive $n^{\Theta(k)}$ is known.
Further, with access to fewer than $n^{k/2}$ samples, there exists some evidence for $\exp(\poly(n))$-hardness of refuting $k$-sparse \lpn{} \cite{kothari17soslbanycsp}. 

The study of the sparse variant \klwe{} was initiated much more recently \cite{jain24sparseLWE}, but again no improvement to the naive $n^{\Omega(k)}$ sample and time requirement is known (and in \cite{jain24sparseLWE}, the authors show that natural strategies fail). 
Altogether, 
\begin{quote}
    \begin{center}
         the fastest known algorithms for
        \klpn{} and \klwe{} take $n^{\Omega(k)}$ time.
    \end{center}
\end{quote}
 

\paragraph{Applications of \klpn{} and \klwe.}
Assuming that the $n^{\Omega(k)}$ runtime for \klpn{} is optimal has become a standard and very useful assumption in learning theory and cryptography.
In~\cite{feige02refutingkxor,alekhovich03averagecasevsapproximation}, a natural hardness assumption for 3-\lpn{} was shown to imply new hardness of approximation results for min bisection, densest subgraph, 
 and max bipartite clique. 
Hardness of \klpn{} has also been instrumental in deriving hardness results for statistical tasks: agnostically learning halfspaces
\cite{daniely2016complexity}, learning DNFs \cite{daniely2021local,bui2024structured},  and noisy tensor completion~\cite{barak16noisytensor} (see also our result described in~\cref{sec:applications_intro}).

Beyond applications to other problems, the refutation task on \klpn{}
(also known as refuting random \kxor{} formulas) is a canonical problem in average-case complexity that is of interest in its own right~\cite{feige02refutingkxor,alekhovich03averagecasevsapproximation,feige06witnesses,barak16noisytensor,Allen2015HowTR,raghavendra17stronglyrefuting,kothari17soslbanycsp,guruswami22algorithmscsp,guruswami23semirandom}.

In cryptography, sparse \lpn{} has found applications in private- and public-key encryption \cite{ishai08cryptoconstant,applebaum10pkcryptographyfromdifferentasssumptions,dottling12indcca}, signature schemes \cite{ishai08cryptoconstant},
construction of pseudorandom generators with linear stretch in $\mathbf{NC}^{\mathbf{0}}$ \cite{applebaum06prgnc0}, oblivious transfer \cite{couteau21silverlpnot}, multi-party homomorphic secret sharing \cite{dao23multipartyhomomorphic},
 indistinguishability obfuscation \cite{ragavan24io} and, concurrent to our work, somewhat homomorphic encryption~\cite{cryptoeprint:2024/1760}.

As noted above, the study of \klwe{} 
was only recently initiated in~\cite{jain24sparseLWE}. The authors use \klwe{} to construct state-of-the-art linear (and constant degree) homomorphic encryption schemes. In the current work, we give a further application of \klwe{} --- we derive a lower bound for tensor completion based on the worst-case hardness of lattice problems.

\subsubsection{An Unsatisfactory State of Affairs: Motivation For This Paper}

In light of the continuous use of the hardness of sparse noisy linear equations over the past more than 20 years
and recent developments (\cite{jain24sparseLWE} for \lwe), we anticipate 
further applications of the assumptions to fundamental problems in cryptography, learning theory, and statistics. Of central importance to this endeavour are the following questions: 
\begin{quote}
    \begin{center}
        \emph{What are reasonable hardness assumptions for problems on $k$-sparse linear equations?\\
        How can we be confident in these assumptions?}
    \end{center}
\end{quote}
These questions are the main motivation for our work.
We address them via a reduction: we exhibit an algorithm that transforms standard noisy linear equations into sparse noisy linear equations. 
Our reduction shows that significantly improving over the brute-force running times of the sparse version would improve the state-of-the-art algorithms for the dense version, which would be a major breakthrough.
We remark that our aim is to obtain \emph{tight} hardness for the sparse problems under standard assumptions on the dense problem, essentially converting between brute-force running times.
That is, obtaining $n^{\Omega(k)}$-type lower bounds for the sparse problem assuming $2^{\Omega(n/\poly(\log n))}$-hardness for the dense problem.


We note that the prior work of \cite{jain24sparseLWE} also gives a reduction from $k$-dimensional standard (dense) \lwe{} to $k$-sparse $n$-dimensional \lwe{}, showing that if the standard version takes exponential time, then the $k$-sparse version takes time at least $2^{\Omega(k)}$.
This lower bound, however, is suboptimal, as it does not depend on the dimension $n$ and does not show hardness when $k = O(\log n)$, a sparsity that is 
desirable for the applications in~\cite{jain24sparseLWE} and crucial for the applications of \klpnspace to hardness of learning. 
In particular, the level of sparsity often naturally corresponds to some measure of complexity in the learning problem (such as the number of terms in a DNF formula), and the goal is to rule out polynomial-time algorithms for as simple settings as possible (e.g., DNFs with as few terms as possible).
Further, it is far from the 
running time of the naive $n^{\Omega(k)}$ algorithm. We discuss this further in \cref{sec:priorreduction}.





\subsection{Main Result: Near-Optimal Time-Sparsity Tradeoffs}

Our main result,
\cref{thm:maininformal},
is 
a nearly tight reduction from the standard (dense) setting to the sparse setting for solving noisy linear equations. Our hardness results in the sparse problem nearly match known algorithmic guarantees under standard assumptions on the dense problem.

Our reduction is flexible and addresses a variety of scenarios. It applies to all three of the decision, search, and refutation problems (see~\cref{sec:tasks} for definitions).
Second, our reduction applies simultaneously to a wide range of ring sizes, capturing both the \lpn{} setting of $q = 2$ and polynomial $q$, and standard settings of \lwe{} such as $q$ being a large polynomial (or even larger) in the dimension of the dense problem. Third, the reduction is agnostic to the error distribution in the case of decision and search problems.
Finally, our reduction can accommodate other variations, such as 
an arbitrary distribution over $\zqmult^k$ on the support of the rows of $A$ (see \cref{sec:complete_arbitrary_rows}) and varying support sizes across samples (see \cref{sec:variablesupportsizes}).
\paragraph{Implied Hardness Results and Comparison to Known Algorithms.}
Our result implies that, assuming (near-)exponential hardness of the dense problems, we achieve hardness of the sparse problems that (nearly) matches algorithmic upper bounds.


\begin{corollary}[\klpn{} hardness]
\label{cor:introhardnessklpn}
    If there is no $2^{o(L/\log L)}$-time algorithm for the decision/search/ refutation problem for standard \lpn{} in dimension $L,$ then there is no $n^{o(\tfrac{k}{ \log k(\log k + \log \log n)^2})}$-time algorithm for the decision/search/refutation problem for \klpn{} in dimension $n.$
\end{corollary}

This result is nearly tight, since as noted earlier, given $O(n^{1.01k})$ samples there are $n^{O(k)}$-time algorithms. 

Our lower bounds do not restrict the number of samples accessible to the algorithm.
It is interesting to compare this with the more fine-grained picture for refutation when given fewer than $O(n^{k/2})$ samples.  
Given only $O(n^{k(1-\delta)/2+\delta})$ samples, for some constant $\delta$, the fastest known algorithm runs in time $2^{\tilde{O}(n^\delta)}$ \cite{raghavendra17stronglyrefuting}.
Furthermore, in the case of binary \lpn{} there is evidence, based on lower bounds against the powerful sum of squares hierarchy, that this might be necessary~\cite{grigoriev2001linear,schoenebeck2008linear,kothari17soslbanycsp}. In~\cref{sec:low_degree_hardness}, we provide similar evidence for any finite ring $\zqz$ and for arbitrary distribution on the support for the rows of the public matrix $A$ within the framework of \emph{low-degree polynomial tests}~\cite{barak19sos,hopkins2017efficient,hopkinsThesis}.
We leave open the problem of showing an analogous result using a reduction-based approach as in this paper.


In a similar vein to our \klpn{} result above, we obtain the following hardness in the case of \lwe:
\begin{corollary}[\klwe{} hardness]
    If there is no $2^{o(L)}$-time algorithm for the decision/search/refutation problem for standard \lwe{} in dimension $L$ with $q = \poly(L)$, there is no $n^{o(\tfrac{k}{\log k\cdot( \log k + \log \log n)})}$ time algorithm for the decision/search/refutation problem for \klwe{} in dimension $n.$
\end{corollary}

For all three of decision/search/refutation for both \lwe{} and \lpn, our results imply that under standard near-exponential hardness assumptions on the dense problem, there is no polynomial-time algorithm for the sparse problem when $k = \omega((\log \log n)^2(\log\log\log n))$ for \lpn{} and $k = \omega((\log \log n)(\log\log\log n))$ for \lwe{} (and $k = o(\sqrt{n})$ always).
For both \lwe{} and \lpn{}, we also derive similar hardness results in the sparse setting assuming only sub-exponential or quasi-polynomial hardness of the standard problems.  
\cref{sec:standardlpn,sec:standardlwe} describe these and other results.

\paragraph{Downstream Applications to Cryptography and Learning Theory.}
As a direct consequence of~\cref{cor:introhardnessklpn}, many (but not all) previous results based on \klpn{} can now be based on standard \lpn.
We refer to~\cref{sec:applications_learning,sec:applications_crypto} for an in-depth discussion.

\subsection{Application to Tensor Completion}
\label{sec:applications_intro}

Besides the downstream applcations obtained by combining~\cref{thm:maininformal} with previous reductions, we also show a new reduction from $k$-sparse \lwe{} to noisy tensor completion (see, e.g., \cite{barak16noisytensor}) that we discuss next.
Together with~\cref{thm:maininformal} this gives the first hardness for this problem based on a well-established worst-case assumption (when combined with the standard reduction from the gap-shortest vector problem to the dense version of \lwe~\cite{regev05LWE,peikert09publickey}).

We define the problem rigorously in \cref{sec:definingtensorcompletion} and only give an informal overview here.
In the tensor completion problem, there is an unknown  order $k$ tensor that is approximately rank $r$,
\begin{align}
\label{eq:align}
T = \sum_{i = 1}^r \sigma_i \cdot( u_i^{1}\otimes  u_i^{2}\otimes\cdots\otimes u_i^{k}) + \Delta  
\end{align} 
in $\mathbb{R}^{n \times \ldots \times n}$ and we observe $m$ uniformly random entries. 
Here, 
$\Delta$ is some appropriately constrained noise tensor.
For normalization, we assume that the components $u^j_i\in \mathbb{R}^n$ satisfy $\|u^j_{i}\|_2 = \sqrt{n}$. The
$\sigma_i$ are scalars (note that they can make the norms arbitrary large or small).
The goal is to construct an estimate $\hat{T}$ such that $\norm{T - \hat{T}}_1 \leq o(\norm{T}_1)$ using as few observed entries as possible.


There exist algorithms \cite{barak16noisytensor} (and \cite{montanari2016spectraltensor} for the noiseless version) for this problem running in time $n^{O(k)}$ with access to $m = \tilde{O}(rn^{k/2})$ observed entries under two assumptions which are \emph{necessary even for inefficient algorithms.}
First, \emph{the noise $\Delta$ has small norm relative to $T$}, for instance it holds that $\norm{\Delta}_1 = o(\norm{T}_1)$. 
Second, the components are incoherent, i.e. \emph{no component $u^j_i$ is concentrated on a small set of coordinates.}
We discuss these assumptions more in \cref{sec:definingtensorcompletion}.

Of course, if one aims to explicitly store $\widehat{T}\in \mathbb{R}^{n^{\otimes k}},$ there is no hope for an algorithm running in time $o(n^{k}).$ Yet, one may still hope that one can learn a data-structure $\widehat{T}$ which on input $(i_1,i_2,\ldots,i_k)$ returns $\widehat{T}_{i_1, i_2, \ldots, i_k}\in \mathbb{R}$ such that $\mathbb{E}_{(i_1, \ldots, i_k)} \lvert T_{i_1, \ldots, i_k} - \widehat{T}_{i_1, \ldots, i_k} \rvert = o\big(\mathbb{E}_{(i_1, \ldots, i_k)} \lvert {T}_{i_1, \ldots, i_k} \rvert\big).$ For example, if the algorithm finds scalars $(\widehat{\sigma_i})_{i = 1}^r$ and components $(\widehat{u_{i}}^j)_{1\le i \le k, 1 \le j \le r}$ such that 
$$
\sum_{i = 1}^r \sigma_i \cdot( u_i^{1}\otimes  u_i^{2}\otimes\cdots\otimes u_i^{k})\approx \sum_{i = 1}^r \widehat{\sigma_i} \cdot( \widehat{u_i}^{1}\otimes  \widehat{u_i}^{2}\otimes\cdots\otimes \widehat{u_i}^{k}),
$$ then the data structure can 
compute $\widehat{T}_{i_1, i_2, \ldots, i_k}$ as $\sum_{\ell = 1}^r\widehat{\sigma_\ell}(\widehat{u_\ell}^1)_{i_1}(\widehat{u_\ell}^2)_{i_2}\cdots(\widehat{u_\ell}^k)_{i_k}$ in time $\tilde{O}(kr)$, which is much faster than $n^k$ when the rank $r$ is small.

Such an algorithm running in time $n^{o(k)},$ however, is not known. We give the first, to the best of our knowledge, hardness evidence via a reduction from the decisional binary $k$-sparse \lwe{} problem (in which all values on the support are equal to 1).
Combining this with~\cref{thm:maininformal} and the worst-to-average-case reduction from the gap shortest vector problem to standard (dense) LWE~\cite{regev05LWE} gives hardness evidence based on worst-case lattice problems.
\begin{theorem}[Lattice-Based Lower-Bound For Tensor Completion, Informal]
\label{thm:tensorcompletioninformal}
Assume that the gap-shortest vector problem in dimension $L$ takes time at least $2^{\Omega(L)}$ in the worst-case (for quantum algorithms).
Then, order-$k$ noisy tensor completion with rank $r\ge 2^{k-1}$ cannot be solved in time $n^{o(\frac{k\log n}{\log k\cdot (\log k + \log \log n)})}$ even under the {small noise} and {incoherence} assumptions.
\end{theorem}

We remark that the small noise and incoherence assumption hold with exponentially high probability.
Our lower bound will follow by showing that we cannot solve the following tensor distinguishing problem: Either we observe i.i.d. sampled entries of a noisy low-rank tensor, or we observe i.i.d. sampled entries of a tensor with independent, mean $o(1)$, and constant variance entries.




\section{Overview of Our Reduction}
\label{sec:technical_overview}
We next give an overview of the proof of our main theorem, namely \cref{thm:maininformal}.

\subsection{The High-Level Approach: Linear Mapping from Dense to Sparse Equations}

In what follows, we denote the input to the dense problem by $(B,y)$ and the input to the sparse problem by $(A,y')$.
Suppose we are given $(B,y)$ such that $B$ is uniform over $(\zqz)^{m \times L}$ and we want to solve one of the decision/search/refutation problems given access to an oracle for the same problem in the sparse setting.
The high-level idea is to transform each dense equation to a sparse one in a way that preserves the approximate linear relation $b_i ^\top s\approx y$ (or absence thereof) between ``coefficients'' $b_i$ in row $i$ of $B$ and the label $y_i$.

A common approach in the literature towards this end (e.g. \cite{boneh13bit} for binary LWE, and \cite{jain24sparseLWE} for sparse LWE in certain regimes, discussed in \cref{sec:priorreduction}) is to map the coefficient vector of each equation.
If we can find a linear mapping, this will preserve the linear relation between the coefficient vector and the label.
More specifically, we aim to find a matrix $G \in (\zqz)^{n \times L}$ (we sometime refer to this as a \emph{gadget} or \emph{decoding matrix}) such that given a sample $(b_i, y_i) \in (\zqz)^L \times \zqz$, we can efficiently find a $k$-sparse $a_i\in (\zqz)^{n}$ such that $a_i^\top G = b_i^\top$.
Note that given this, it holds that 
\begin{align}
\label{eq:gadgeteq1}
y_i \approx b_i^\top s \text{ for some } s 
\quad \Leftrightarrow \quad y_i \approx a_i^\top (Gs) \text{ for some } s. 
\end{align}
Further, it should hold that the sparse equation has the correct distribution, i.e., that $a_i$ is uniform over $k$-sparse vectors in $(\zqz)^n$.
That is, we require the following property:
\begin{align}   
\label{eq:gadgeteq2}
    \text{If }b_i\sim\unif((\zqz)^L), \text{ then }a_i \text{ is a uniformly random }k\text{-sparse vector in }(\zqz)^{n}.
\end{align}

We refer to the process of computing $a$ given $b$ as the \emph{decoding algorithm}. (Indeed, a watchful reader would already have seen the analogy with syndrome decoding in coding theory).
We claim that such a primitive can imply a reduction between all of the decision/search/refutation versions.
We illustrate this here for the search task, and refer to~\cref{sec:reduction_for_all_three_tasks} for the other two tasks.\footnote{Crucially, for the refutation task we need that \emph{almost all} of the equations are mapped successfully with high probability, while this is not necessary for decision and search. See~\cref{sec:reduction_for_all_three_tasks} for details.}
Assume that $y_i = b_i^\top s + e_i$ for some error term $e_i$.
Then, we can compute the mapped equation $y_i' = y_i + a_i^\top z$, where $z \in (\zqz)^n$ is drawn uniformly at random.
It holds that
\[
    y_i' = y_i + a_i^\top z = b_i^\top s +e_i + a_i^\top z= a_i^\top (Gs + z) + e_i \,,
\]
where we use the same $z$ for each sample. The only role of $z$ is that the final secret $Gs + z$ is uniform over $(\zqz)^n$ and independent of everything else.
Thus, given a search oracle for the sparse version we can recover $Gs + z$ and hence also $Gs$.
If $G$ is full-rank we can then recover $s$. 

\subsection{Prior Gadget Matrix Constructions and Challenges}
\label{sec:priorreduction}
The only prior work that we are aware of which reduces standard (dense) noisy linear equations to sparse ones is 
\cite{jain24sparseLWE} in the case of  \lwespace (and large-modulus LPN) with $q$ prime.
The authors follow the gadget matrix approach outlined above with $L = k$ and $G \in (\zqz)^{n \times k}$ the Vandermonde matrix, i.e., whose $i$-th row is equal to $(i^0, i^1, i^2, \ldots, i^{k-1})$.

To solve $a^\top G = b^\top$ on input $b\sim\unif((\zqz)^k)$ such that $a$ is a uniformly random $k$-sparse vector, 
they first draw a uniformly random support $S\in \binom{[n]}{k}$ \emph{independently} of $b$.
Then, they restrict $a$ and $G$ to coordinates $S$ and solve the equation $a^{\top}_SG_{S\times [k]} = b^\top$.
Whenever $q\ge n,$ \emph{every $k\times k$ minor of $G$ is of full rank} and, hence, a unique solution for $a_S$ exists.
Since $b$ is uniform over $(\zqz)^k$, it is easy to check that $a_S$ is also uniform over $(\zqz)^k$, so $a$ has the correct distribution.
Overall, the reduction of \cite{jain24sparseLWE} shows that $n$-dimensional $k$-sparse \lwe{} is at least as hard as $k$-dimensional \lwe; assuming exponential hardness of \lwespace implies that the $k$-sparse version takes time at least $2^{\Omega(k)}$.

We outline two ways in which one might hope to improve over this reduction:

\paragraph{Desideratum 1: Small Values of $k.$}
As noted just above, the hardness implied by \cite{jain24sparseLWE} for $n$-dimensional \klwe{} is at best only $2^{\Omega(k)}$, which does not yield any hardness against polynomial-time (in $n$) attacks when $k = O(\log n)$. 
It is instructive to analyze where their reduction is loose.
As an intermediate step of transforming a dense $b \in (\zqz)^L$ to a $k$-sparse $a \in (\zqz)^n$, they compute a dense vector $\tilde{a} \in (\zqz)^k$ such that $\tilde{a}^\top \tilde{G} = b$, for some deterministic matrix $\tilde{G}$.
Note that the entropy of $b$ is equal to $L \log q$ and this must be smaller than the entropy of $\tilde{a}$.
Since the entropy of $\tilde{a}$ is at most $k\log q$, we must have $k \log q \geq L \log q$ and hence $k \geq L$. This approach does not exploit the entropy of the support location of $a\in (\zqz)^n.$ There are $\binom{n}{k}$ choices for the support of $a,$ which yields $\log \binom{n}{k}\approx k\log n$ additional bits of entropy. Thus, if one simply requires $a^\top G = b^\top$ (instead of $\tilde{a}^\top \tilde{G} = b$), the corresponding entropy calculation is $k \log n + k\log q\ge L\log q,$ allowing for larger values of $L.$
We record this as a guiding principle:
\begin{equation}
\label{eq:P1}
\tag{P1}
\textit{
        For small $k,$ when solving $a^\top G = b^\top$, one must choose the support of $a$ based on $b$. 
    }
\end{equation}

Note that a similar entropy does give \emph{some} limitations for general decoding matrices when $L = \omega(k).$
Since the entropy of $a$ is at most $k \log n$ it must hold that $L \log q \leq k \log n$, i.e., $L \leq k \log n / \log q$.\footnote{In this argument we implicitly assumed that $G$ is deterministic. Since we use the same $G$ for all equations, even for a random $G$ the average entropy added on any single equation is at most $\tfrac{n L \log q}{m}$ which is negligible for all parameter settings we consider.}
We will show that this upper bound can nearly be achieved.
In particular, we will be able to choose $L = k\log n/\poly(\log k,\log\log n, \log q)$.

\paragraph{Desideratum 2: Small Moduli.}
We would like our reduction to work also for small moduli, and in particular for $q = 2$ as in the case of \lpn.
Note that a key step in the reduction of \cite{jain24sparseLWE} is the inversion of an arbitrary $k \times k$ sub-matrix of $G$, which requires every $k\times k$ minor of $G$ to be full rank.
Such matrices $G$ are known as \emph{maximum distance separable} and in the regime $n\ge k,$
conjecturally, cannot exist if $q< n-2$ (e.g. \cite{huffman03fundamentals}). Hence, one cannot hope for a reduction in the important case of binary \lpn{} using the approach of \cite{jain24sparseLWE}.

We would also like to be more versatile, e.g. obtain sparse vectors $a_i$ with distribution on the support different from uniform over $(\zqz)^k$. As one example, our application to the hardness of tensor completion from \lwespace{} requires \emph{binary} $k$-sparse coefficients $a_i$. This does not seem possible using the approach in \cite{jain24sparseLWE}.

\subsection{A Toy Reduction for Block-wise Sparsity}

A gadget mapping dense $L$-dimensional vectors to {\em blockwise} $k$-sparse $n$-dimensional vectors is not hard to construct. A blockwise sparse vector is one whose coordinates are divided into $k$ blocks, each of which has a single non-zero coordinate. Indeed, for a parameter $t=L/k$, let $H \in (\zqz)^{t \times q^{t}}$ denote a matrix whose columns are all the vectors in $(\zqz)^{t}$. Then, we will let the decoding gadget $G = I_k \otimes H \in (\zqz)^{t k \times q^{t}k}$. First of all, since $L=t k$ and $n= q^{t} k$, notice that $$ L = \frac{k\log (n/k)}{\log q} \,.$$
Secondly, given an $L$-dimensional vector $b$, one can decode each block of $t$ coordinates of $b$ into a unit vector in $q^{t}$ dimensions. Concatenating the unit vectors gives us a block-wise sparse $n$-dimensional vector $a$ such that $a^\top G = b^\top$. 

Permuting the coordinates might seem to suffice to turn blockwise sparsity into true sparsity. However, since the same permutation has to be used for all $m$ samples, the blockwise structure starts becoming apparent from even a small number of samples. For example, given $\omega(n^2)$ samples, every two coordinates $i,j$ appear in some row of the public matrix $A$ when the supports are uniformly random $k$-tuples. This property is not shared by the blockwise model for any number of samples under any permutation.

To decode to random sparse vectors, we need a cleverer idea which we describe next. We begin with an informal non-mathematical overview.

\subsection{Motivation for Our Reduction: Translation of Natural Languages}
\label{sec:naturallanguages}
The problem at hand is reminiscent of an age-old task: translating from one language (dense coefficient vectors) to another (sparse coefficient vectors). 
Translation of natural languages motivates our approach.
First, we split each sentence (coefficient vector) into words and translate word-by-word using a dictionary. Second, we use the grammar of the two languages to combine the respective words into meaningful sentences (coefficient vectors).
The intuition why such an approach is efficient is that the small number of possible words (as opposed to possible sentences or even paragraphs) allows for efficient translation. Imagine if one had to print a dictionary in the real world containing translations of all possible sentences!

For us, the grammar of the ``dense language'' is simple. Each sentence is simply a concatenation of words. Our dictionary allows us to translate each of these words into a sequence of $k$ indices over some small domain. These new sequences of length $k$ are the words in the ``sparse language''.      
One of our contributions is to specify a grammar that allows us to combine $k$ indices over a small domain into a $k$-sparse equation. It should already be clear that the approach taken is more sophisticated than simple concatenation, which would blow-up the sparsity $k.$

\subsection{Overview of Our Reduction to Sparse Equations} 
\label{sec:nontechnicaloverview}
We begin by explaining the reduction in the case that the distribution on the support is deterministically equal to 1, which captures both sparse problems \klpn{} and \emph{binary} \klwe.
At the end of this section, we explain how a simple modification addresses arbitrary distributions over $\zqmult^k,$ the integers coprime with $q,$ on the support.

Recall that as in~\cite{jain24sparseLWE}, we want to construct a gadget matrix $G\in (\zqz)^{n\times L}$ and, on input a uniform random $b \in (\zqz)^L$, solve $a^\top G = b^\top$ for a $k$-sparse $a\in \{0,1\}^n$. 
However, as observed in \eqref{eq:P1}, in doing so, we need to use the randomness of $b$ in choosing the location of the support of $a.$
In particular, we need to sample from (some large enough subset of) the possible supports. This task can be interpreted as a sampling version of 
the popular $k$-SUM problem, in which given $n$ vectors  in $(\zqz)^L$ and a target $b$ one needs to find $k$ of them that sum to $b.$
Evidence for hardness of this problem exists both in the average-case setting (when 
the $n$ vectors are random)
based on lattice problems \cite{brakerski2020hardness} and in the setting of worst-case vectors based on the exponential-time hypothesis \cite{patrascu10fastersat,abboud22sethksum}. Essentially, these results state that the time required is $\min(n^{\Omega(k)}, 2^{\Omega(L/\log k)})$, both of which are too slow for our purposes. Instead, we need a reduction that runs in time $n^{o(k)}$ (to conclude $n^{o(k)}$ hardness of $k$-sparse \lwe) and $2^{L/\log L} = 2^{L/\log \omega(k)}$ in the relevant regime $L = \omega(k)$
in the case of \lpn{} (as our reduction needs to be faster than the strongest plausible $2^{L/\log L}$ hardness of \lpn). This step of solving $k$-sum poses a serious obstacle to our approach. Unfortunately,
our task is actually even harder -- we furthermore need to \emph{sample} from the $k$-SUM solutions. 

We circumvent the hardness of $k$-sum by endowing $G$ with a convenient structure that makes finding solutions and sampling from them easy.

\paragraph{Word Decoding: Splitting into Small Sub-Problems That Can Be Solved Efficiently.}
We will decompose the problem into $t$ sub-problems by splitting the string $b$ into $t$ equal parts (similar to the celebrated BKW algorithm \cite{BKW}) of size $h\coloneqq L/t.$
Namely, represent each $b\in (\zqz)^L$ as a concatenation of $t$ uniform independent strings $b = (b_{1}, b_{2}, \ldots, b_{t})$ in $(\zqz)^h$.
We will decode each $b_{j}$ into a $k$-tuple of indices over $[hq]^k$ via a smaller gadget matrix $V \in (\zqz)^{hq \times h}$ (this is the dictionary mentioned above). As we will need to solve $k$-SUM over this $V,$ it is natural that we want $V$ to have a very simple explicit structure. We take $V$ to contain all multiples of the identity:
\begin{align*}
    V= 
    \begin{pmatrix}
      0 \times I_h\\
      1\times I_h\\
      \quad \vdots\\
      (q-1)\times I_h
    \end{pmatrix}.
\end{align*}
We call the $hq$ rows $V_{1,:}, V_{2,:},\ldots, V_{hq,:} \in \mathbb{Z}^h_q$ words.
We want to find $\xi_j = (\xi_j(1),  \ldots, \xi_j(k)) \in [hq]^k$ such that $$\sum_{\psi = 1}^k V_{\xi_j(\psi),:} = b_{j}\,.$$ 
Furthermore, if $b_j$ is uniform over $(\zqz)^h$, we will require the marginal distribution of $\xi_j$ to be uniform over $[hq]^k$.
We achieve this by outputting a uniformly random $\xi_j$ that satisfies $\sum_{\psi = 1}^k w_{\xi_j(\psi)} = b_{j}$.

Here, we take the opportunity to introduce some useful notation:
\begin{equation}
\label{eq:preimagedef}
\begin{split}
\Preimage_{V}(b) \coloneqq \Set{(\xi(1), \xi(2), \ldots,\xi(k))\in [hq]^k\suchthat \sum_{\psi = 1}^k V_{\xi(\psi).:} = b}.
\end{split}
\end{equation}
The simple structure of $V$ allows for extremely efficient sampling from the preimage of $b.$

\begin{lemma}[Efficient preimage sampling]
\torestate{
\label{lem:samplingintro}
On input $k>h+2$ and $b \in (\zqz)^h,$ \cref{alg:samplingfrompreimage} outputs a uniformly random sample from $\Preimage_V(b).$ The algorithm runs in time $\poly(k,h,\log q).$
}
\end{lemma}

On a high level, this can (roughly) be achieved by sampling indices uniformly at random (from $[hq]$) until the number of indices left matches the number of coordinates that differ between the current state, i.e., the sum of rows of $V$ indexed by the current indices, and $b$.
The choice of the remaining coordinates is then fixed.
Permuting them randomly will result in a uniformly random element from the preimage. We give full detail in \cref{sec:proofofsampling}.

When combining the different words into sentences, it will be crucial that $\xi_j(1),  \ldots, \xi_j(k)$ are independent and uniform over $[hq].$ To guarantee this, a key property will be that each $\Preimage_V(b)$ is of nearly equal size. Otherwise, $k$-tuples $\xi_j(1),  \ldots, \xi_j(k)$ corresponding to $b$ with large (respectively, small)
$\Preimage_V(b)$ will be under-represented (respectively, over-represented). Again, the simple structure of $V$ can be utilized to show the following fact.

\begin{lemma}[Uniformity of preimage size]
\torestate{
\label{lem:nearuniformityintro}
    Suppose that $k \ge 4h(\log h +\log q +  \log (1/\psi))$ where $\psi\in (0,1)$ is any real number. Then, for each $u \in (\zqz)^h$,
    $$
    \frac{(1-\psi)(hq)^k}{q^h}\le 
    |\Preimage_V(u)|
    \le 
    \frac{(1+\psi)(hq)^k}{q^h}.
    $$
}
\end{lemma}
For reasons that will become clear, we will choose $\psi = o(1/t),$ say $\psi = (ht)^{-2}.$

We give the proof of 
\cref{lem:nearuniformityintro}
in \cref{sec:randomwalkproof}, which follows by analyzing the point-wise convergence of the following random walk on $(\zqz)^h.$ At each step $j,$ the random walk chooses independently uniformly random $i_j \in [h]$ and $m_j\in \{0,1,2,\ldots, q-1\}$ and moves $m_j$ units in the direction $e_{i_j}$. Recalling the definition of $\Preimage_V(b)$ in \eqref{eq:preimagedef}, we interpret this as $\xi(j) = h m_j +i_j$ so that $V_{\xi(j), :} =m_j (e_{i_j})^\top.$

In the case $q = 2,$ this is just the standard lazy random walk over $\{0,1\}^h.$ It is well-known that this walk converges to an $\epsilon$- total variation distance from the uniform distribution after roughly $h(\log h + \log(1/\epsilon))$ steps (e.g. \cite{LevinPeresWilmer2006}). For our applications, a stronger point-wise convergence guarantee is crucial.
Trying to deduce a point-wise convergence bound from the total variation bound would require $k \gtrsim h^2$, which would result in a much looser lower bound. Instead, we directly verify the needed point-wise convergence. 

In order to make the output tuples exactly equal, we will use a simple rejection-based method that rejects each $b_j$ with probability $1 - \frac{|\Preimage_V(b_j)|}{{(1+\psi)(hq)^k}/{q^h}},$ which is a well-defined probability and is at most 
$2\psi$ under the guarantee of \cref{lem:nearuniformityintro}.
To make this algorithmically efficient, we need to be able to compute $|\Preimage_V(b_j)|.$ Again, the simple structure of $V$ allows such a primitive:

\begin{lemma}[Computing preimage sizes]
\torestate{
\label{lem:sizeestimationintro}
There exists an algorithm, \cref{alg:sizeestimation}, which on input $k$ and $b \in (\zqz)^h,$ outputs
$|\Preimage_V(b)|.$
The algorithm runs in time $\poly(k,h,\log q).$
}
\end{lemma}

The key observation for the algorithm is that due to symmetry, the size of the set $\Preimage_V(b)$ only depends on the number of non-zero entries of $b$.
This size can be efficiently computed via dynamic programming. We give full detail in \cref{sec:estimationproof}.

\begin{remark} \cref{lem:nearuniformityintro} explains why at least when using the gadget matrix $V,$ we need to split the original $L$-dimensional problem into $t$ smaller problems. For the desired mixing to occur (even in total variation), we need $k \ge h\log(h)$.  Hence, if we directly use $V$ to decode the entire dense vector, we will reduce the hardness of $k$-sparse \lpn/\lwe{} in dimension $n$ from standard \lpn/\lwe{} in dimension $h = O(k/\log k).$ Thus, the obtained hardness will be at most $2^{O(k/\log k)},$ which is
significantly worse than the desired $2^{k\log n/\poly(\log k, \log \log n)}$ for small values of $k.$  We overcome this issue via our novel ``grammar'' of sparse sentences.
\end{remark}

\paragraph{Sentence Decoding: Combining the Solutions of Small Sub-Problems.}
 We will think of concatenations of words as sentences.
 Each sentence will have exactly $t$ words.
 The rows of our final gadget matrix will be all possible sentences, of which there are $ (hq)^t$.

Let $n = (hq)^t$.
We will use the decodings $\{\xi_j\}_{j=1}^t$ of $\{b_{j}\}_{j = 1}^t$ obtained via the matrix $V$ to produce a $k$-tuple $\xi \in [n]^k = ([hq]^t)^k$ decoding of $b.$ The key insight is that $\{b_{j}\}_{j = 1}^t$ are $t$ i.i.d. vectors in $(\zqz)^h$ and, hence, $\{\xi_j\}_{j=1}^t$ are $t$ i.i.d. tuples in $[hq]^k.$ In particular, $(\xi_j(\psi))_{1\le j \le t, 1\le \psi\le k}$ are $kt$ independent uniform integers in $[hq].$

By definition, $\xi$ should be uniformly random. We define $\xi(\psi)\coloneqq (\xi_1(\psi),\xi_2(\psi),\ldots, \xi_t(\psi))$ for $1 \leq \psi \leq k$ and $\xi= (\xi(1), \xi(2), \ldots, \xi(k)).$ Visually, we represent $\xi$ as the following matrix in $[hq]^{t\times k}.$
$$\xi = 
\begin{pmatrix}
    \xi_1(1) \; \xi_1(2)\; \ldots \; \xi_1(k)\\
    \xi_2(1) \; \xi_2(2) \; \ldots\; \xi_2(k)\\
    \vdots\\
    \xi_t(1)\; \xi_t(2)\; \ldots  \; \xi_t(k)\\
\end{pmatrix}.
$$
In it, each row corresponds to a decoded word (that is, $\xi_j(1), \xi_j(2), \ldots, \xi_j(k)$ is the $k$-tuple corresponding to $b_j$). Each column $\xi(\psi) = (\xi_1(\psi), \xi_2(\psi), \ldots, \xi_t(\psi))^\top$ corresponds to a $t$-tuple in $[hq]^t$ and, hence, can be treated as an index in $[(hq)^t] = [n].$ This is the $k$-sparse grammar we referred to.
Crucially, each of the $k$ indices (columns) are independent uniform. Thus, we have managed to translate $b\in (\zqz)^{ht}$ into $k$ uniform independent indices in $[(hq)^t] = [n],$ as desired. 

To finish, we put this in the language of gadget matrices and demonstrate a linear relationship between $b$ and $\xi.$ Let $G\in (\zqz)^{n \times L}$ be such that its rows are indexed by $ (\phi_1, \ldots, \phi_t) \in [hq]^t$ and each row is defined by
\begin{equation}
\label{eq:gdefintro}
G_{(\phi_1, \phi_2,\ldots, \phi_t),:} = 
(V_{\phi_1,:},V_{\phi_2,:},\ldots, V_{\phi_t,:})\in (\zqz)^{ht} = (\zqz)^{L} \,.
\end{equation} 
In other words, $G$ contains all possible concatenations of $t$ words from $V.$

Since $\sum_{\psi = 1}^k V_{\xi_j(\psi),:} = b_{j}^\top$ holds for each $j$, it must be the case that 
\begin{equation}
\begin{split}
\label{eq:introsumingadget}
\sum_{\psi = 1}^k G_{\xi(\psi)} = \sum_{\psi = 1}^k
 \Big( V_{\xi_1(\psi),:},V_{\xi_2(\psi),:},\ldots, V_{\xi_t(\psi),:}\Big) &= \bigg( \sum_{\psi = 1}^k V_{\xi_1(\psi),:},  \sum_{\psi = 1}^k V_{\xi_2(\psi),:},\ldots,  \sum_{\psi = 1}^k V_{\xi_t(\psi),:}\bigg) \\
 &= \big( b_1^\top, b_2^\top, \ldots, b_t^\top\big) = b^\top\,.
 \end{split}
\end{equation}
Under a mild condition $k = o(\sqrt{n}),$ the birthday paradox shows that with high probability, all of the $\xi(\psi)$ are distinct. When this event holds we can represent $\xi$ as a (binary) $k$-sparse vector $a$ that is uniformly random (the non-zero coordinates of $a$ are indexed by $\xi(1), \xi(2), \ldots, \xi(k)$).
\eqref{eq:introsumingadget} is equivalent to $a^\top G = b^\top$ as desired.

\paragraph{Example Choice of Parameters for LPN.}
We formally choose the parameters in \cref{sec:parameters} but here we give a brief intuition about how large they are for concreteness.
For simplicity, we focus on the decision version of \lpn, i.e., $q = 2$.
We assume $2^{\Omega(L/\log L)}$ hardness of standard  \lpn{} and will derive a $n^{\Omega(\frac{k}{\log k (\log k + \log \log n)^2})}$ lower bound for the $k$-sparse decision version.\footnote{For technical reasons we require the condition that $k = o(\sqrt{n})$.}

Note that $t = \frac{\log n}{\log h + \log q} = \frac{\log n}{1 + \log h}.$ Hence, $L = ht = \frac{h\log n}{1 + \log h}.$ This value is clearly increasing in $h.$ Thus, we want to choose the largest $h$ possible in order to get the strongest runtime lower bound.

However, to apply Lemma~\ref{lem:nearuniformityintro} we require 
$k \ge 4h(\log h + \log q + \log (ht)^2)$ and $4h(\log h + \log q + \log (ht)^2) = \Theta(h(\log h + \log\log n)).$ Hence, we can choose $h$ at most $h = c\frac{k}{\log k + \log\log n}$ for some small absolute constant $c>0.$ Thus, $t = \frac{\log n}{\log h + \log q} =\Theta(\frac{\log n}{\log k}).$ Altogether, $L = ht = \Theta(\frac{k\log n}{\log k(\log k + \log\log n)}).$ We obtain  $$\exp\Big(\Omega\big(\frac{L}{\log L}\big)\Big) = \exp\Big(\Omega\big(\frac{k\log n}{\log k(\log k + \log\log n)^2}\big)\Big)\quad \text{ hardness.}$$

\paragraph{Other Distributions on the Support.} 
Given a vector $b\in (\zqz)^L,$ we first solve the binary problem. That is, find a $k$-tuple $(i_1, \ldots, i_k) \in [n]^k$ such that the corresponding rows of $G$ sum to $b$.
If $b$ is uniform, then so is $(i_1, \ldots, i_k),$ as already mentioned.
Next, sample $\rho = (\rho_1, \rho_2, \ldots, \rho_k)\in \zqmult^k$ from $\mathcal{D}$.
Suppose that the row of $G$ indexed by $i_j$ corresponds to
\[
(\kappa_1(j)e^\top_{\Delta_1(j)},\kappa_2(j)e^\top_{\Delta_2(j)},\ldots, \kappa_t(j)e^\top_{\Delta_t(j)}) \,,
\]
where $\kappa_i(j) \in 
\zqz$ and $\Delta_i(j) \in [h].$
We use the convention that the first $h$ (zero) rows of $V$ are $0e_1^\top,0e_2^\top,\ldots,0e_s^\top.$
Now, replace this index with the index $i_j'$ corresponding to
\[
(\kappa_1(j)\rho_j^{-1}e^\top_{\Delta_1(j)},\kappa_2(j)\rho_j^{-1}e^\top_{\Delta_2(j)},\ldots, \kappa_t(j)\rho_j^{-1}e^\top_{\Delta_t(j)})\,.
\]
Since $\rho_j $ and $q$ are coprime, $i_j'$ is still uniform over $[n]$.
Suppose again that all $i_j'$ are distinct, we can show that this happens with high probability.
Let $a'$ be the vector with support $(i'_1, i'_2, \ldots, i'_k)$ such that $a'_{i'_j} = \rho_j$.
Then, $(a')^\top G = a^\top G = b^\top$ and clearly the support of $a'$ is uniform and the distribution on the support equal to $\cD$. We give full detail in \cref{sec:complete_arbitrary_rows}.

\paragraph{Variable Support Sizes.} Finally, we can also accommodate for varying support sizes. The reason is that all of our primitives - \cref{lem:samplingintro,lem:sizeestimationintro,lem:nearuniformityintro} only rely on $k$ being sufficiently large and, thus, it is enough to accommodate for the effectively smallest possible support size $k.$ 

More concretely,
suppose that the distribution of support sizes is $\mathcal{K}$ where with probability $1 - \eta$ with $\eta = o(1/m)$ the support size is at least $k_\minsf.$ Then, for each sample $b_i,$ we first draw $k_i\sim \mathcal{K}$ and then use the word-decoding and sentence-decoding procedures for $k_i.$
We can show that the $k$-tuples are still uniform for all $k_i \ge k_\minsf$ if the input $b$ is random.
With high probability, we never see $k_i< k_\minsf.$ More detail is given in \cref{sec:variablesupportsizes}.

\subsection{Reduction for Decision, Search, and Refutation}
\label{sec:reduction_for_all_three_tasks}

As before, \cref{eq:gadgeteq1,eq:gadgeteq2} are sufficient to obtain a reduction for any of the decision, search, and refutation tasks.
Thus, given a search oracle we can recover $Gs + z$ and hence also $Gs$.
We already discussed this for the search task.
Recall that given $(b_i, y_i)$, our reduction outputs $(a_i, y_i' \coloneqq y_i + a_i^\top z)$, where $z$ is uniform over $(\zqz)^n$.
Note that this implies that the ``planted'' distribution of the decision task will be mapped correctly.
Similarly, if in the original input $y_i$ and $b_i$ are independent then so will $y_i'$ and $a_i$ be since decoding $b_i$ into $a_i$ does not depend on $y_i$ and $y_i' = a_i^\top z + y_i$.
Note that for both decision and search, the reduction preserves the error of each equation.

Lastly, we discuss refutation. For simplicity of exposition, we focus on the \lpn{} setting here, but the same arguments apply to \lwe{} (see \cref{sec:reductionandtasks}).
Given $(A,y),$ we want two guarantees. For all approximately-satisfiable systems, i.e. $As\approx y$ for some $s,$ we want to output \textit{approximately satisfiable} with probability 1.
For the case when $A,y$ are drawn independently (and hence the system is far from satisfiable with high probability), we want to output \textit{strongly unsatisfiable}, with high probability.

In the discussion on search and testing, we considered the slightly idealized setting in which every equation is decoded successfully. Yet, many reductions (including ours) lose samples.
Note that for decision and search even only preserving a small 1/100 (or less) fraction of the samples would be sufficient.
For refutation, this is not sufficient since our algorithm has to map approximately satisfiable inputs to approximately satisfiable outputs. Yet, consider the following extreme case. A very large $99/100$ fraction of the clauses in the input are satisfiable due to the fact that they are all the same $(a_1,y_1).$ The rest $1/100$ of instances are highly unsatisfiable. 
A decoding process which always fails on $a_1$ and succeeds on all the other samples
would lose the $99/100$ fraction of satisfiable equations and
could produce a highly unsatisfiable instance. 

Luckily, we are able to construct gadget matrices and a decoding algorithm that only lose a $o(1)$-fraction of the samples with high probability over worst-case inputs. Indeed, recall that each $b_j\in (\zqz)^h$ is rejected with probability at most $O(\psi) = O((ht)^{-2})$ and, hence, $b$ is rejected with probability at most $O(t(ht)^{-2}) = o(1).$
This resolves the above-described issue. Hence, our reduction maps approximately satisfiable instances to approximately satisfiable instances and instances with $B,y$ independent to instances with $A,y$ independent. 
Thus, it also works for refutation.

\subsection{Organization of the Rest of the Paper}
In \cref{sec:background_short} we give formal definitions and summarize relevant prior work.
In \cref{sec:ourcontributions} we formally state our results.
The following section \cref{sec:reductiontosparsity} contains the full reduction.
Our reduction as well as its run-time are parametrized. In \cref{sec:parameters} we describe several specific parameter instantiations which lead to hardness results for (sparse) \klpn{} and \klwe{} based on popular assumptions for the hardness of standard \lpn{} and  \lwe{.
We end with our application to noisy tensor completion in \cref{sec:tensorcompletion}.

\section{Background on Noisy Linear Equations and Tensor Completion}
\label{sec:background_short}
\subsection{Noisy Linear Equations over \texorpdfstring{$\zqz$}{ZqZ}}
 In this current paper, we focus on the following family of noisy linear equations over $\zqz$ unless otherwise stated. Recall that $\zqz$ is the ring of integers modulo $q,$ which is furthermore a field when $q$ is a prime.
Let $q$ be a given modulus and $\Dcoeff$ be a distribution over $(\zqz)^n$ and $\Derr= (\Derr(y))_{y \in \zqz}$ be a family of distributions.
We generate samples $\{(a_i, y_i)\}_{i=1}^m$ as follows:
\begin{enumerate}
    \item Sample $a_1, a_2, \ldots, a_m\iidsim \Dcoeff$.
    \item Sample $s\sim \unif((\zqz)^n)$. 
    \item Sample $e_j\sim \Derr(a_j^\top s)$ independently for each $j \in [m].$
    \item Output $\{(a_i, y_i \coloneqq a_i^\top s + e_i)\}$.
\end{enumerate}
We will denote an input consisting of $m$ samples from this distribution by 
\begin{align}
\label{eq:defnle}
    \NLE(\samples = m,\modulus = q,  \Dcoeff, \Derr).
\end{align}
For example, $\NLE(\samples = m,\modulus = 2,  \unif(\mathbb{Z}_2^n), \bernoulli(\frac{1}{2}-\gamma))$ captures standard \lpnspace  and 
$\NLE(\samples = m,\modulus = q,  \unif((\zqz)^n), D_\sigma)$ captures standard \lwe. 
Most frequently, the noise $\Derr(y)$ is independent of $y\in \zqz,$ e.g. Discrete Gaussian. We state in this more general case to include problems like learning with rounding, which are also captured by our reduction \cite{banerjee12pseudorandomlattices}.

The following notation will be useful for sparse linear equations.
Suppose that $\mathcal{D}$ is a distribution over $\zqmult^k,$ where $(\mathbb{Z}/q\mathbb{Z})^\times$ denotes the set of residues relative prime with $q.$
Let $\randomsupport(n,k,\mathcal{D})$ be the following distribution over $k$-sparse vectors in $(\zqz)^n$:
First, draw a uniformly random set $\{i_1, \ldots, i_k\}$ of size $k$.
Then, draw $\rho = (\rho_1, \rho_2, \ldots, \rho_k)\sim \mathcal{D}$ and set $a_{i_j} = \rho_j$ for all $1\le j \le k$ and $a_i = 0$ otherwise. 

We are interested in instances of $\NLE(m, q,  \randomsupport(n,k,\mathcal{D}), \Derr).$
Concrete examples of interest include:
\begin{enumerate}
    \item \emph{\klpn}:  Let $q = 2$ and $\mathcal{D}$ output the vector $1^k$ deterministically. 
    \item \emph{Binary \klwe}: Let $\mathcal{D}$ output the vector $1^k$ deterministically.
    This is the basis for our reduction to tensor completion, see \cref{sec:tensorcompletion}.
    \item \emph{\klwe{} with (Almost) Uniform Values on The Support}: Let $\mathcal{D}$ be $\unif\Big(\zqmult^k\Big)$.
\end{enumerate}
The last example almost recovers the sparse \lwespace instance from \cite{jain24sparseLWE}.
The only difference is that in \cite{jain24sparseLWE}, the distribution $\mathcal{D}$ on the support is $\unif\Big((\zqz)^k\Big)$ instead. 
This seems to make little difference in the regime of a large prime $q$ and $k$ as in \cite{jain24sparseLWE}.
However, our work focuses on the case when both $k$ and $q$ are very small with respect to $n$.
In particular, this means that under $\mathcal{D}=\unif\Big((\zqz)^k\Big),$ there would be many samples $a$ with support of size 1.
This makes both testing and search easy.
We refer the reader to \cref{sec:variablesupportsizes} for a further discussion.  

\subsection{Statistical Tasks for Noisy Linear Equations over \texorpdfstring{$\zqz$}{ZqZ}}}
\label{sec:tasks}
Associated to noisy linear equations are several different tasks. The simplest one is \emph{decision}. 

\begin{definition}[Decision]
\label{def:detection}
On input $m$ pairs $\{(a_i, y_i)\}_{i = 1}^m$ as well as the parameters $m,q,$ \linebreak
$\Dcoeff,\Derr,$ one needs to solve the following hypothesis testing problem.
\begin{equation}
\begin{split}
    & \mathsf{H}_0: \quad a_1,\ldots, a_m\iidsim\Dcoeff \text{ and independently } y_1, \ldots, y_m\iidsim\unif(\zqz) \label{eq:tetsing} \quad \text{ versus} \\
    & \mathsf{H}_1: \quad \NLE(\samples = m,\modulus = q,  \Dcoeff, \Derr).
\end{split}
\end{equation}
\end{definition}
In words, the goal is to decide whether there is an approximate linear dependence between the samples $a_1, a_2, \ldots, a_m$ and labels $y_1, y_2, \ldots, y_m$.
Another task is to recover the (approximate) linear relation under the promise that there exists one:
\begin{definition}[Search]
\label{def:search}
On input $m$ pairs $\{(a_i, y_i)\}_{i = 1}^m$ from 
$$\NLE(\samples = m,\modulus = q,  \Dcoeff, \Derr)$$ 
as well as the parameters $m,q,$
$\Dcoeff,\Derr$, recover $s$.
\end{definition}

Finally, of interest is the refutation task. This task arises naturally in certain frameworks such as the sum-of-squares hierarchy when one tries to solve the search or detection task (see, e.g., \cite{barak16noisytensor,raghavendra17stronglyrefuting}). In words, the task is to construct a proof that the uncorrelated $\mathsf{H}_1$ distribution in \eqref{eq:tetsing} does not correspond to \emph{any} secret vector $s$.
One way to formalized this would be to output an upper bound on the number of simultaneously satisfiable equations.
We will use the following more general formalization:
Suppose that $\mu:\zqz\longrightarrow [0,1]$ is some weight function such that $\mu(x) = 0$ if and only if $x = 0$ and $\mu(y) = \mu(-y).$
For example, $\mu(x) $ could be the indicator $\indicator[x \neq 0]$ (which will correspond to the number of unsatisfiable equations). Or, the normalized distance to the circle defined by $\mu(x) = \big(\frac{\min(|x|, q-|x|)}{q}\big)^p$ for some $p\in [1,+\infty)$(which will correspond to an $L_p$ weight of the noise vector in \lwe).

\begin{definition}[Strong Refutation]
\label{def:refutation}
Suppose that with high probability, for some $1\ge \Delta>\delta> 0$ and\footnote{The setting of $\delta = 0$ corresponds to \emph{weak refutation} and its analogue is also studied in the literature for other natural problems such as the planted clique \cite{barak19sos}. In the case of linear equations over finite fields, however, weak refutation is trivially solved by Gaussian elimination as also noted in \cite{raghavendra17stronglyrefuting}.}
weight function $\mu,$ it is the case that 
\begin{align*}
&
\frac{1}{m}\sum_{i = 1}^m \mu(y_i-\langle a_i, s\rangle) \ge\Delta \text{ for all }s\\
& \quad\quad\quad\quad\text{ if } a_1,a_2,\ldots, a_m\iidsim\Dcoeff, y_1, y_2, \ldots, y_m\iidsim\unif(\zqz).
\end{align*}
An algorithm strongly
$(\mu, \delta, \Delta, p)$-refutes the approximate satisfiability of a random instance on input $m$ pairs $\{(a_i, y_i)\}_{i = 1}^m$ if it has the following guarantee.
\begin{enumerate}
    \item It outputs \textsf{STRONGLY UNSATISFIABLE} with probability at least $p$ if $a_1,a_2,\ldots, a_m\iidsim\Dcoeff$ and $ y_1, y_2, \ldots, y_m\iidsim\unif(\zqz)$.
    \item It outputs \textsf{APPROXIMATELY SATISFIABLE} with probability $1$ for any $\{(a_i, y_i)\}_{i = 1}^m$ for which there exists some $s\in (\zqz)^n$ such that $\frac{1}{m}\sum_{i = 1}^m \mu(y_i-\langle a_1, s\rangle)  \le \delta$.
\end{enumerate}  We call an $m$-tuple 
$\{(a_i, y_i)\}_{i = 1}^m$ with the last property $\delta$-\textsf{APPROXIMATELY SATISFIABLE}.
\end{definition}

\begin{remark}[On the Relative Hardness of Testing, Search, and Refutation]
An algorithm solving search also solves the testing task. Similarly, under the additional guarantee for refutation that 
\begin{align*}
    & \frac{1}{m}\sum_{i = 1}^m \mu(y_i-\langle a_1, s\rangle)  \le \delta\text{ for some $s$}\\
& \quad\quad\quad\quad\text{ if } \{(a_i, y_i)\}_{i = 1}^m\sim 
\NLE(\samples = m,\modulus = q,  \Dcoeff, \Derr)
\end{align*}
with high probability,
a refutation algorithm also solves the testing task.
Thus, a lower-bound for testing also implies a lower bound for the other two problems. 
\end{remark}

We now give concrete instantiations of the above for the case of \lpnspace  and \lwe.

\subsection{Learning Parity With Noise}
\label{prelim:lpn}
Based on the hardness of \lpn, one can construct one-way functions, secret key identification schemes \cite{krzysztof12cryptoLPN}, bit commitment schemes \cite{abhishek12comitments}, 
private key encryption schemes \cite{gilbert08howtoencrypt,applebaum09fastcrypto},
public key encryption schemes \cite{applebaum10pkcryptographyfromdifferentasssumptions,kiltz14cca,yu16auxilliary}, signatures schemes with constant overhead \cite{ishai08cryptoconstant}, 
two-round oblivious transfer \cite{dottling20tworound},
and others. In learning theory, \lpnspace  is perhaps the most canonical problem with a statistical-computational gap and has 
found applications to lower bound constructions for learning probabilistic finite automata with an evaluator \cite{kearns94learnability}, certain classes of graphical models \cite{mossel06phylogeny,bresler19bolzmann}, 
agnostically learning halfspaces \cite{kalai2008agnostically,klivans2014embedding},
learning ReLU functions \cite{goel2017reliably,goel2019time}, and others
.


\subsubsection{Standard \lpn}
\label{sec:standardlpn}

As discussed before, we set $q= 2,\Dcoeff = \unif(\{0,1\}^n)$ and $\Derr = \bernoulli(\frac{1}{2}-\gamma)$ for some $\gamma>0$.
When $\gamma>0$ is a constant, $n$ samples are sufficient and necessary to solve the search problem \emph{information theoretically} with high probability. 
The celebrated Blum-Kalai-Wasserman algorithm \cite{BKW} solves the search problem in time $2^{O(n/\log n)}$ with access to $2^{O(n/\log n)}$ samples.
In \cite{lyubashesky05lpn}, Lyubashevsky builds on the BKW algorithm and gives a $2^{O(n/\log \log n)}$-time algorithm with access to only $n^{1+\epsilon}$ samples for any constant $\epsilon>0.$
The algorithms of  Blum-Kalai-Wasserman and Lyubashevsky are sate-of-the-art for standard learning parity with noise and the respective running-times are oftentimes assumed as optimal.
Formal evidence for the hardness of learning parity with noise is a $\exp(\Omega(n))$ lower bound in the statistical-query model \cite{kearns94learnability}.
While the reader might be at first surprised that this suggests stronger hardness, it is no mistake.
The statistical-query model is a restricted model of computation and does not capture the combinatorial algorithms of \cite{BKW,lyubashesky05lpn}.
Finally, \cite{feldman2006new} provides hardness evidence for \lpn{} via reductions from learning DNF formulas and juntas. 

\paragraph{Our Hardness Assumptions.}

In this current work, we will analyse the implications of the following standard hardness assumptions of standard \lpn:
\begin{enumerate}[leftmargin=4\parindent]
    \item [\mylabel{itm:nearexplpn}{\textbf{(\lpnspace A1)}}] \emph{Near-Exponential Hardness of \lpn: } Detection/Search/Refutation for standard \lpnspace  is hard  in time  $2^{o(n/\log n)}.$
    \item [\mylabel{itm:subexplpn}{\textbf{(\lpnspace A2)}}] \emph{Sub-Exponential Hardness of \lpn: } Detection/Search/ Refutation for standard \lpnspace  is hard in time  $2^{n^{1-\alpha}}$ for some absolute constant $\alpha \in (0,1).$
    \item [\mylabel{itm:quasilpn}{\textbf{(\lpnspace A3)}}] \emph{Quasi-Polynomial Hardness of \lpn: }  Detection/Search/ Refutation for standard \lpnspace  is hard in time  $2^{(\log n)^{1+c}}$ for some absolute constant $c>0.$
\end{enumerate}


\subsubsection{Sparse \lpn}
\label{prelim:sparselpn}
The problem of sparse \lpnspace  was introduced in \cite{alekhovich03averagecasevsapproximation,feige02refutingkxor}.
In the $k$-sparse \lpnspace  setting, we set $q= 2,\Dcoeff = \randomsupport(n,k,1^k)$ and $\Derr = \bernoulli(\frac{1}{2}-\gamma)$ for some $\gamma>0.$
Note that such equations can equivalently be interpreted as random noisy $k$-XOR formulas. Concretely, viewing each $k$-sparse sample $a_i$ as the indicator of $(t_i^1, t_i^2, \ldots, t_i^k)\in [n]^k,$ one is provided with $m$ uniformly drawn samples 
\begin{align}
\label{eq:noisykxor}
\Big\{\big((t_i^1, t_i^2, \ldots, t_i^k), (\bigotimes_{j = 1}^k s_{t_i^j})\otimes e_{i}\big)\Big\}_{i= 1}^{m}.
\end{align}

The refutation task is at least as hard as the decision task.
Further, there are search-to-decision reductions~\cite{applebaum12pseudorandom,bogdanov19XORcodes}.
In particular, \cite{bogdanov19XORcodes} shows that under the assumption that one can solve decisional $k$-sparse \lpnspace in time $\mathcal{T}$ with $m$ samples, one can also solve the search version in time $\poly(\mathcal{T}, m, n)$ using $O(m^{1 + 2/k})$ samples.

The task of refuting noisy $k$-XOR has received significant attention.
In this task, we take $\mu(x) = \indicator[x\neq 0]$ in \cref{def:refutation} and, hence, one aims to refute the fact that a non-trivial fraction of clauses can be satisfied. 
Spectral and sum-of-squares refutation algorithms achieve the following guarantees. Using $m = \tilde{O}(n^{k/2})$ samples, there exists an algorithm running in time $n^{O(k)}$ which refutes the noisy $k$-XOR instance~\cite{barak16noisytensor,Allen2015HowTR}. 
Further, there are sub-exponential time algorithms that work even when given less samples.
For any $\delta\in (0,1),$ there exists an algorithm running in time $\exp(\tilde{O}(n^\delta))$ which refutes the noisy $k$-XOR instance using $m = \tilde{O}(n^{(k/2 - 1)(1-\delta) +1})$ samples~\cite{raghavendra17stronglyrefuting}. 

\paragraph{Known Hardness Results}

The following hardness results are known for $k$-sparse \lpn.
Within the sum-of-squares hierarchy, the algorithms of~\cite{barak16noisytensor,montanari2016spectraltensor,raghavendra17stronglyrefuting} are tight up to logarithmic factors (by setting $D = n^{\delta}$ in the following).
In particular, the canonical degree-$D$ SoS relaxation, fails to refute an instance of the noisy $k$-XOR with $m = n\times \tilde{O}(n/D)^{\frac{k-1}{2}}$ samples.  
Further, in certain regimes, there exists further evidence of hardness against other restricted models of computation such as $\mathbf{AC^0}$ circuits and myopic algorithms, see~\cite{applebaum12pseudorandom}. 

The only hardness evidence (conditionally) ruling out any polynomial time algorithm for refuting noisy $k$-XOR that we are aware of is the following win-win result.
In \cite{liu17hardnesssparselpn}, the authors show that one of the following two statements is true.
Either there is a polynomial-time reduction from standard (dense) \lpn{} with noise rate $\mu$ to a version of $k$-sparse \lpn{} in which each entry of the coefficent vectors are distributed as $\bernoulli(\mu)$.
Or, one can build a public key encryption scheme from standard \lpn. The authors of \cite{liu17hardnesssparselpn} view this as evidence for the hardness of sparse \lpn{} as public-key encryption from standard \lpn{} with low noise has remained elusive over the years and would be a major breakthrough.

\subsection{Learning With Errors and Lattice Problems}
\label{prelim:lwe}
\lwespace  is a cornerstone of modern cryptography that enables fully homomorphic encryption \cite{brakerski11fhe} (building on Gentry's original work \cite{gentry09fhe,gentry10fhe} and improvements \cite{smart10fhe,disjk10fheintegers,brakerski11fheringlwe} which all rely on variants of \lwe), non-interactive zero-knowledge proofs \cite{canetti19fiatshamir}, public-key encryption \cite{regev05LWE}, succinct non-interactive argument systems \cite{jin24snarg,mathialagan24snark} and many others (see \cite{peikert16decade}).
More recently, the \lwespace assumption 
has also been used to establish hardness of several fundamental statistical problems such as learning Gaussian mixture models \cite{bruna21clwe,gupte2022continuous}, learning (intersections of) halfspaces \cite{klivans06cryptographic,diakonikolas2022cryptographic,tiegel23agnostic,diakonikolas2023near,tiegel24intersection}, learning periodic neurons \cite{zadik24neurons}, and others. 
\subsubsection{Standard \lwe}

In the standard learning with errors problem, one takes $\Dcoeff = \unif((\zqz)^n)$ and $\Derr = \mathcal{D}_{\sigma}$, where $\mathcal{D}_{\sigma}$ is a discrete Gaussian over the integer lattice centered at $0$ (see \cite{regev05LWE} for a formal definition).

There are several known algorithms for the search variant, achieving different guarantees in different regimes. First, one can adapt the Blum-Kalai-Wasserman algorithm \cite{BKW,wagner02generalizedbirthday} and obtain a running time of $q^{\Omega(n/\log(q/B))}$ with $q^{\Omega(n/\log(q/B))}$ samples where $B$ is the support-size of the noise\footnote{That is, there exists some set $\mathcal{B}\subseteq \zqz$ of size $B$ such that $e_1, e_2, \ldots, e_m\in B$ with high probability.}. Next, an algebraic algorithm of Arora and Ge \cite{arora11newalgorithms} works in time $O(n^B)$ with $n^{O(B)}$ samples. Finally, the Lenstra-Lenstra-Lovasz algorithm \cite{Lenstra1982LLL} and improvements \cite{SCHNORR1987201} based on a reduction to lattice problems succeed whenever $q/B = \exp(\Omega(n))$ or $m = \exp(\Omega(n)).$

In particular, for the setting of large polynomial modulus $q = \Theta(n^C)$ (where $C$ is a large constant) and large enough  Gaussian noise $D_\sigma$ (with $\sigma = n^{c}$ for some $c\ge 1, c\ll C$) the fastest algorithms run in time $2^{\Omega(n)}$.
This is no coincidence, and is known to be tight under standard assumptions on \emph{worst-case} lattice probblems.
We do not define these problems as we will not need them in the current paper and reference \cite{peikert16decade}.
In particular~\cite{regev05LWE}, let $q \le \exp(\poly(n))$ and 
$\beta \in (0,1)$ be such that $\beta q\ge 2\sqrt{n}.$
Then, the decisional problem for  
$$
\LWE(\samples = m, \modulus = q, \Dcoeff = \unif((\zqz)^n),
\Derr = D_{\beta q}
)
$$
is at least as hard as the gap shortest vector problem and the shortest independent vectors problem on $n$-dimensional lattices within approximation ratio $\tilde{O}(n/\beta)$ for quantum algorithms.
In \cite{peikert09publickey}, this result is dequantized for a variation of the gap shortest vector problem and under the additional assumption $q \ge 2^{n/2}$ for the standard  gap shortest vector problem.
This assumption was removed in~\cite{brakerski2013classical}, however, they only showed hardness under the assumption that $\sqrt{n}$-dimensional lattice problems are hard.

\paragraph{Our Hardness Assumptions.}

Still, a near- or sub-exponential time algorithm for $n$-dimensional \lwespace would be considered a breakthrough.
Together with the known best algorithms for \lwespace and lattice problems, this gives rise to the following common hardness assumptions: 

\begin{enumerate}[leftmargin=4\parindent]
    \item [\mylabel{itm:nearexplwe}{\textbf{(LWE A1)}}] \emph{Exponential Hardness of \lwe: } 
    Detection/Search/Refutation for standard \lwe
    $$
\LWE(\samples = m, \modulus = q, \Dcoeff = \unif((\zqz)^n),
\Derr = D_{\beta q}
)
$$
    is hard in time  $2^{o(n)}$ when $\beta q>2\sqrt{n}.$
    \item [\mylabel{itm:subexplwe}{\textbf{(LWE A2)}}] \emph{Sub-Exponential Hardness of \lwe: } For some constant $\beta \in (0,1),$ 
    Detection/Search/ Refutation for standard LWE
    $$
\LWE(\samples = m, \modulus = q, \Dcoeff = \unif((\zqz)^n),
\Derr = D_{\beta q}
)
$$
    is hard in time  $2^{n^{1-\alpha}}$ when $\beta q>2 \sqrt{n}$ for some absolute constant $\alpha\in (0,1).$
    \item [\mylabel{itm:quasilwe}{\textbf{(LWE A3)}}] \emph{Quasi-Polynomial Hardness of \lwe: } For some constant $\beta \in (0,1),$ 
    Detection/Search/ Refutation for standard LWE
    $$
\LWE(\samples = m, \modulus = q, \Dcoeff = \unif((\zqz)^n),
\Derr = D_{\beta q}
)
$$
    is hard in time  $2^{(\log n)^{1+c}}$ when $\beta q>2 \sqrt{n}$ for some absolute constant $c>0.$
\end{enumerate}


\label{sec:standardlwe}
\subsubsection{Sparse \lwe}
\label{prelim:sparselwe}
The study of \lwespace with a sparse public matrix is much more recent and was initiated in \cite{jain24sparseLWE}. Specifically, they consider the case of 
$$
\LWE\Big(\samples = m,\modulus = q,  \randomsupport\big(n,k,\unif((\zqz)^k)\big), D_\sigma\Big).
$$
They show that when $q \geq n$, $k$-sparse \lwespace in dimension $n$ is at least as hard as 
standard $k$-dimensional \lwe.
Beyond that, the authors of \cite{jain24sparseLWE} show that natural strategies based on lattice reductions, finding short vectors in the kernel of the public matrix, and variants of the BKW algorithm do not seem to break $k$-sparse \lwe.

\subsection{The Tensor Completion Problem}
\label{sec:definingtensorcompletion}

A noisy rank-$r$ order-$k$ tensor is defined by the equation 
\begin{equation}
\label{eq:defnoisytensor}
T = \sum_{i = 1}^r \sigma_i \times u_i^{1}\otimes  u_i^{2}\otimes\cdots\otimes u_i^{k} + \Delta \in \mathbb{R}^{n\times n\times \cdots \times n}.
\end{equation}
where $\sigma_i$ are scalars (note that they can make the norms arbitrary large), $\Delta$ is some noise tensor and $u_{i}^{j} \in \mathbb{R}^{n}$ for all $i \in [r], j\in [k]$ with $\|u^j_{i}\|_2 = \sqrt{n}$ for normalization. 

\paragraph{Tensor Completion.}
In the tensor completion problem, one observes a set of randomly chosen entries of an unknown tensor.
That is, for a uniformly random subset $S$ of $[n]^k,$ one observes $\{T_{i_1, i_2, \ldots, i_k}\}_{(i_1, i_2, \ldots, i_k)\in S}.$
The goal is to output a tensor $\widehat{T}\in (\mathbb{R}^{n})^{\otimes k}$ such that 
$T$ and $\widehat{T}$ are close in some metric.
This problem has numerous applications throughout the sciences including 
image impainting \cite{liu09tensorcompletionimpainting}, hyperspectral imaging \cite{li10hyperspectral,signoretto19tensorcompletionhyperspectral}, geophysical imaging \cite{kreimer13tensorcompletiongeophysics}.
A standard choice is the (normalized) $\ell_1$ norm~\cite{barak16noisytensor}, i.e. 
$$\frac{1}{n^k}\sum_{i_1, i_2, \ldots, i_k\in [n]^k}|T_{i_1, i_2, \ldots, i_k} - \widehat{T}_{i_1, i_2, \ldots, i_k}| =\frac{1}{n^k} \|T - \widehat{T}\|_1.$$

Of course, insofar as $\widehat{T}\in (\mathbb{R}^{n})^{\otimes k},$ explicitly outputting a tensor $\widehat{T}$ requires time and space $\Omega(n^k).$ Yet, in the low-rank regime this information is compressible. A noiseless rank-$r$ order-$k$ tensor is uniquely specified by $rnk$ numbers, the entries of $\{u_i^j\}_{1\le i \le r, 1\le j \le k}.$ Given these numbers, one can evaluate an entry in $\tilde{O}(rk)$ time. 

Thus, we formalize the tensor completion task via the following 
evaluation
oracle.

\begin{problem}[Tensor Completion --- Search Problem]
\label{problem:tensorcompletion}
In the noisy tensor completion problem, there is a hidden noisy order-$k,$ rank-$r$ tensor $T$ as in \eqref{eq:defnoisytensor} where $\|\Delta\|_1 =o(\|T\|_1).$ 
On input a uniformly random subset $S$ of $[n]^k$ of size $m,$ one observes $\{T_{i_1, i_2, \ldots, i_k}\}_{(i_1, i_2, \ldots, i_k)\in S}.$  The goal is to compute in time $\mathcal{T}$ a data structure $\widehat{T}$ with the following guarantees. First, $\widehat{T}$ should produce a real number on input $(i_1, i_2, \ldots,i_k)\in [n]^k$ in time $\mathcal{T}.$ Second, 
$$
\expect_{(i_1, i_2, \ldots,i_k)\sim \unif([n]^k)}
\left[\left|\widehat{T}[i_1, i_2, \ldots,i_k] - 
T_{i_1, i_2, \ldots,i_k}
\right|\right] \leq o\Big(\frac{1}{n^k}\|T\|_1\Big) \,.
$$
\end{problem}

One can also ask for the analogous refutation (with,say, $\mu(x, y) = \|x-y\|_1$) and testing tasks for tensor completion introduced in \cref{sec:tasks}. 

\paragraph{Typical Assumptions.} For the problem to be solvable even by inefficient algorithms, several assumptions are necessary: 
\begin{enumerate}
    \item \emph{Incoherence.} One needs to impose some condition that the vectors are spread-out.
    Otherwise, if the mass concentrates on a small set, given only random samples one may not be able to observe any non-trivial entries.
    The following condition is sufficient: A vector $a\in \mathbb{R}^n$ is $C$-incoherent if 
    $\|a\|_\infty\le \frac{C}{\sqrt{n}}\|a\|_2.$ Then, assume that all components are incoherent for some small $C$ \cite{barak16noisytensor,montanari2016spectraltensor}.
    \item \emph{Signal-to-Noise Ratio.} Finally, one needs to impose a condition on the noise $\Delta.$ As a minimal assumption, one needs $\|\Delta\|_1 = o(\|T\|_1).$ 
\end{enumerate}

\paragraph{Algorithms.} Perhaps surprisingly, without any further distributional assumptions on $\Delta$ and $\{u_i^j\}_{1\le i \le r, 1\le j \le k},$ the above three conditions are sufficient for efficient algorithms \cite{barak16noisytensor} (the related \cite{montanari2016spectraltensor}  is on noiseless tensor completion). Here we give the respective guarantees.

\begin{theorem}[Informal, \cite{barak16noisytensor,montanari2016spectraltensor}] Consider an instance of the tensor completion problem \cref{problem:tensorcompletion} with $\polylog(n)$-incoherent components and $r\le n^{k/2-1}.$ Then, there exists an algorithm solving the search tensor completion problem with access to $\polylog(n)\times (rn^{k/2})$ samples in $\poly(n^{k/2})$ time.
\end{theorem}

\paragraph{Lower Bounds.} One may recognize that the decisional version of noisy $k$-XOR (sparse \lpn) problem is directly related to a tensor distinguishing problem that can be solved via a tensor completion algorithm. Indeed, in the planted case of noise $k$-XOR, let the secret be $s\in \{\pm 1\}^n.$ Then, one observes entries of the tensor $T = s^{\otimes k} + \Delta$ where each entry of $\Delta$ is in $\{-2,0,2\}.$ Namely, $\Delta_{i_1,i_2,\ldots,i_k} = 0$ in the case of noise equal to 0 and $\Delta_{i_1,i_2,\ldots,i_k} = -2 s_{i_1}s_{i_2}\cdots s_{i_k}$ in the case of noise equal to 1. Furthermore, the tensor is $1$-incoherent. 
In the null case, one just observes completely independent entries that are uniform over $\{-1,1\}$.
Hence, one can translate the lower bounds for noisy $k$-XOR into lower-bounds for tensor completion. 

In the current work, we will develop a similar bound but instead based on sparse \lwespace (and, via our \cref{thm:maininformal}) based on standard \lwe. The advantage of an \lwespace lower-bound (over an \lpnspace  lower bound) is that standard \lwespace inherits the worst-case hardness of lattice problems as in \cite{regev05LWE,peikert09publickey}. Additionally, we obtain a lower bound based on \emph{standard} \lpnspace  as we know it implies the hardness of \emph{sparse} \lpn.

\subsection{Other Works}
\label{sec:otherworks}
\subsubsection{Other ``Non-Standard'' Variants of \lpn{} and  \lwe} 
We remark that besides the sparse variants of \lpn{} and \lwe{} in this paper, other settings beyond the standard ones have also been studied. Again, the motivation is similar --- additional structure in the noisy linear equations is useful towards constructing cryptographic protocols and computational lower bounds.
Different variants are useful towards different goals.
In all cases, a worthwhile pursuit is to relate to the hardness of the standard variants via reductions. Such reductions are well-understood in some cases and entirely open in others.

\paragraph{Structured Public Matrices.} The main object of study in the current work is noisy linear equations with sparse public matrices. Other structures of the public matrix have also proven useful in applications. In \cite{applebaum23sampling}, the authors study \lpn{} with a public matrix with $k$-sparse rows but correlated support locations which ensure a certain structure of the kernel. In \cite{dao24sparsedense}, the authors study a dense-sparse \lpn{} in which the public matrix is a product of a sparse and dense matrix. For these two variants, to the best of our knowledge, no reduction is known from standard \lpn. 

The authors of
 \cite{boneh13bit} show a reduction from standard \lwe{} to \lwe{} when each entry of the public matrix $A$ is an independent uniform bit. We will say more about the high-level idea of this reduction briefly.

\paragraph{Structured Secrets.}

Perhaps the most canonical instance of a structured secret is the case when the secret is $k$-sparse.
But the rows of the coefficient matrix are dense (and uniform over $(\zqz)^L$.
Works using a sparse secret in \lpn{} abound
\cite{kearns98sq,feldman2006new,klivans2014embedding,valiant15findingcorrelations,goel2017reliably,goel2019time,bresler19bolzmann,YAN202176improved,barak22hidden,abbe24staircase,bshouty2024approximatingnumberrelevantvariables}. Similarly, common is \lwe{} with a secret drawn from a $k$-sparse or otherwise low-entropy distribution \cite{goldwasser10entropiclwe,miccancio19binarysecret,brakerski20geenralentropic,gupte2022continuous}. Reductions from standard \lwe{} to \lwe{} with a different secret distribution are somewhat well understood. It turns out that the key 
quantity determining the hardness is 
the min-entropy of the distribution
\cite{goldwasser10entropiclwe,miccancio19binarysecret,brakerski20geenralentropic,gupte2022continuous}. For example, \cite{gupte2022continuous} show that modulus-$q$ \lwe{} with $k$-sparse binary secret in dimension $n$ is as hard as standard modulus-$q$ \lwe{} in dimension $\Theta(k\log (n/k)/\log q).$ 
In addition, for \lpn{}, the reduction of Applebaum, Cash, Peikert and Sahai~\cite{applebaum09fastcrypto} can be repurposed to show that choosing the \lpn{} secret from the error distribution has the same hardness as choosing it at random, modulo a slight difference in the number of samples.

We note that the reduction techniques for structured secrets do not seem applicable to our case of structured public matrices -- as the secret is unknown in both the starting (standard) problem and target problem, no step of explicit decoding a dense secret vector into a sparse secret vector is needed and existential arguments based on the leftover hash lemma suffice. In our case, we need to explicitly decode dense row vectors of the public matrix into sparse vectors, so we need constructive rather than existential arguments. 

\paragraph{Structured Noise.} 
The key difference between \lpn{} and \lwe{} is the sparsity of the noise distribution. While \lwe{} is typically studied with dense (discrete Gaussian) noise, \lpn{} is often of interest with sparse noise. Yet, in certain applications it turns out that \lwe{} with a sparse noise distribution is useful, e.g. \cite{miccancio13sparsenoise,sun20revisitingbinarynoiselwe,jain2021indistinguishability,jain2022indistinguishability,ragavan24io}. It turns out that the noise distribution has a profound effect on the hardness of \lwe. For example, if the noise is binary, the celebrated Arora-Ge algorithm breaks \lwe{} in polynomial time with access to quadratically many samples \cite{arora11newalgorithms}. In \cite{miccancio13sparsenoise,sun20revisitingbinarynoiselwe}, the authors give a direct reduction from worst-case lattice problems to \lwe{} with binary (or otherwise small) noise, but those only apply to a regime of $n(1+ \Omega(1/\log n))$ samples.  

Another setting is that of learning with rounding (LWR)~\cite{banerjee12pseudorandomlattices} where the samples are of the form $(a, \lfloor \frac{p\cdot a^T s}{q} \rfloor \bmod{p})$.
One could think of LWR as LWE with a deterministic, sample-dependent, noise. Our reductions are agnostic of the noise distributions and apply to the LWR setting as well: that is, they reduce (dense) LWR to sparse LWR in the same way as they reduce LWE to sparse LWE.



\paragraph{Continuous \lwe{} (\clwe).} \clwe{} was recently introduced in \cite{bruna21clwe}. In it, the public matrix $A$ has i.i.d. Gaussian entries, the secret $s$ is a random unit vector, and the noise is Gaussian.
The modular arithmetic is enforced in the labels, explcitly  $y_i \equiv a_i^\top s + e_i\pmod{1}.$
In \cite{gupte2022continuous}, the authors show that \clwe{} and \lwe{} are equivalent. \clwe{} is the starting point of many recent hardness results of statistical tasks \cite{bruna21clwe,gupte2022continuous,diakonikolas2022cryptographic,diakonikolas2023near,tiegel23agnostic,tiegel24intersection,zadik24neurons}.

\subsubsection{Memory Lower Bounds in the Streaming Model}
Researchers have also considered other computational resources than running time and number of samples. For example, in \cite{raz17timespace,moshkovitz17mixing,garg17extractor,garg21memory,} the authors study the streaming model 
in which the pairs $(a_i, y_i)$ arrive sequentially. As usual in streaming algorithms, space is the main resource of interest. The state-of-the-art results show that the search problem for standard \lpn{} takes either memory $\Omega(n^2/\epsilon)$ or time 
$\exp(\Omega(n))$ \cite{garg21memory}. In the case of \klpn{} for $k\le n^{.9},$ one needs either 
memory $nk^{.99}/\epsilon$ or $\exp(\Omega(k\log k))$ samples \cite{garg21memory}. Results for larger $k$ and other variants of \lpn{} such as the case of a sparse secret are also covered in the aforementioned works.


\section{Our Contributions}
\label{sec:ourcontributions}

\subsection{Main Reduction}


Here, we state our main reduction guarantee. It holds under the following assumption on the set of parameters.

\begin{parameters}
\torestate{
\label{parameterdependenciess}
Let $L,q,n,k,m, h,t \in \mathbb{N}$ are such that
\begin{enumerate}
    \item $n  = (hq)^t,L = ht.$
    \item $k \ge 4h(\log h + \log q + 2\log(ht))$ and $k = o(\sqrt{n}).$
\end{enumerate}}
\end{parameters}

The reduction is slightly different depending on the specific task --- testing, search, or refutation. Nevertheless, we  can summarize as follows.
\begin{theorem}
    [Full statements in \cref{cor:testingreduction,cor:searchreduction,cor:refutationreduction}] 
    \label{thm:mainreduction}
    Let $L,q,n,k,m, h,t \in \mathbb{N}$ satisfy the conditions in \cref{parameterdependenciess}.
    Let $\cD$ be any distribution over $\zqmult^k$ from which we can produce a sample in time $\Tsample$. 
    Assume that there exists an algorithm that, given $m_1\coloneqq \lfloor m(1- 3/L - k^2/n)\rfloor$ samples, runs in time $\mathcal{T}$ and solves $\textsf{problem}\in \{\textsf{testing},\textsf{search}, \textsf{refutation}\}$ with probability $p$ in the $k$-sparse setting.
    I.e., where the relevant distributions are
\begin{align*}
    \NLE(\samples = m_1, \modulus = q, \Dcoeff = \randomsupport(n, k,\mathcal{D}),\Derr)
\end{align*}
and/or $(\randomsupport(n, k,\mathcal{D}) \times \unif(\zqz))^{\otimes m_1}$.
Then, there exists an 
algorithm which solves in time\linebreak $\mathcal{T} + O(\poly(n,m,\Tsample))$ $\textsf{problem}$ with probability $p- \exp(- 2mL^{-2})$ in the standard setting.
I.e., where the relevant distributions are
\begin{align*}
    \NLE(\samples = m, \modulus = q, \Dcoeff = \unif((\zqz)^L),\Derr)
\end{align*}
and/or $(\unif((\zqz)^L) \times \unif(\zqz))^{\otimes m}$.
\end{theorem}

We note that our reduction is polynomial-time given that we can sample from $\mathcal{D}$ in polynomial time. For example, if $\mathcal{D}$ is just deterministically $1^k,$ we have $\Tsample = O(1)$ deterministically.



We also discuss the probability of success. We always take $ L = n^{o(1)}$ and $m = \Omega(n),$ sometimes as large as $m = n^{o(k/\poly(\log k, \log \log n))}.$
Thus, $\exp(-2mL^{-2})$ is only a small inverse-exponential additive factor. The number of samples is preserved up to $1 - o(1)$ multiplicative factor. 

The key primitive in this reduction is decoding individual equation from the standard (dense) setting to the sparse setting. 
\begin{theorem}
    \torestate{
    \label{thm:decoding_matrix_and_alg}
    Suppose that 
    \cref{parameterdependenciess} holds. Then, there exists an algorithm
    (\cref{alg:sentencedecoding}) that
    performs the following operation in time  $\poly(n,\Tsample)).$
    \begin{enumerate}
        \item Given $b \in (\zqz)^L$, \cref{alg:sentencedecoding} 
        either outputs $\fail$ or outputs a $k$-sparse $a \in (\zqz)^n$ such that $a^\top G = b^\top$. Furthermore, the algorithm outputs $\fail$ with probability at most $2/L + k^2/n.$
        \item If $b \sim \unif((\zqz)^L)$, then, conditioned on not failing, the output $a$ of \cref{alg:sentencedecoding} is distributed according to $\randomsupport(n,k,\cD)$. 
    \end{enumerate}}
\end{theorem}

\subsection{Implied Hardness}  We summarize with tables the implied hardness of sparse \lwespace and sparse \lpnspace based on multiple standard assumptions of the hardness of standard variants.
The full statements and proofs (based on concrete choices of the relevant parameters) are in~\cref{cor:testingreduction,cor:searchreduction,cor:refutationreduction}.

\begin{table}[h]
\centering
\begin{tabular}{ |M{5cm}|M{5cm}|M{5cm}|  }
 \hline
 \multicolumn{3}{|c|}{Hardness of Sparse \lpnspace for Decision and Search} \\
 \hline
 \hline
 Assumed time complexity of standard \lpnspace in dimension $L$ with error distribution $\Bern(\gamma).$ & Implied time complexity of $k$-sparse \lpnspace in dimension $n$ with error distribution $\Bern(\gamma).$&
Restriction on sparsity\\
 \hline
 \hline

   \cref{cor:nearexplpn}: & & \\

 $\exp\paren{\Omega(L/\log L)}$ &
 $\exp\paren{\Omega\big(\frac{k\log n}{\log k(\log k + \log \log n)^2}\big)}$ &
 $k = \omega((\log\log n)^2 \log\log\log n)$, $k = o(\sqrt{n})$
 \\
 
 \hline

 \cref{cor:subexphardness}: & & \\
 $\exp\paren{\Omega(L^{1-\alpha})},\alpha\in(0,1)$&
 $\exp\paren{\Omega\big(\big(\frac{k(\log n)}{\log k(\log k + \log \log n)}\big)^{1-\alpha}\big)}$&
 $k = \omega( (\log n)^{\alpha/(1-\alpha)}\log\log n),$ $ k = o(\sqrt{n})$\\
 \hline

 \cref{cor:quasihardness}: & & \\
 $\exp\paren{\Omega((\log L)^{1+c})},c>0$&
 $\exp\paren{\Omega\big((\log k)^{1+c}\big)}$&
 $k = \exp(\omega((\log n)^{1/(1+c)})),$ $ k = o(\sqrt{n})$\\
 \hline
\end{tabular}
\end{table}
The results for \lwespace are similar looking.

\begin{table}[H]
\centering
\begin{tabular}{ |M{5cm}|M{5cm}|M{5cm}|  }
 \hline
 \multicolumn{3}{|c|}{Hardness of Sparse \lwespace} \\
 \hline
 \hline
 Assumed time complexity of standard \lwespace in dimension $L$ with modulus $q = L^\kappa$ (and error distribution $\Derr$ in case of decision and search)& Implied time complexity of $k$-sparse \lwespace in dimension $n$ with modulus $q,$ distribution on the support $\mathcal{D}$ (and error distribution $\Derr$ in case of decision and search)&
 Restriction on sparsity\\
 \hline
 \hline
 \cref{cor:explwe}: & & \\
 $\exp\paren{\Omega(L)}$&
 $\exp\paren{\Omega\big(\frac{k\log n}{\log k(\log k + \log \log n)}\big)}$&
 $k = \omega((\log\log n)(\log\log\log n)),$
 $ 
 k = o(\sqrt{n})
 $\\
 \hline
 \cref{cor:subexphardness}: & & \\
 $\exp\paren{\Omega(L^{1-\alpha})},\alpha\in(0,1)$&
 $\exp\paren{\Omega\big(\big(\frac{k(\log n)}{\log k(\log k + \log \log n)}\big)^{1-\alpha}\big)}$&
 $k = \omega( (\log n)^{\alpha/(1-\alpha)}\log\log n),$ $ k = o(\sqrt{n})$\\
 \hline

 \cref{cor:quasihardness}: & & \\
 $\exp\paren{\Omega((\log L)^{1+c})},c>0$&
 $\exp\paren{\Omega\big((\log k)^{1+c}\big)}$&
 $k = \exp(\omega((\log n)^{1/(1+c)})),$ $ k = o(\sqrt{n})$\\
 \hline
\end{tabular}
\end{table}

\subsection{Tightness of Results}
\label{sec:tightness}
Of course, one may wonder whether our hardness results for sparse \lpnspace and sparse \lwespace are tight even under the most aggressive (near)-exponential assumptions, in which case we get a $n^{\Omega(\tfrac{k}{\poly(\log k, \log \log n)})}$ time complexity. 
To the best of our knowledge, even for the decision problem (which is the easiest among the three) the best known algorithm requires $O(\sqrt{q} \cdot (qn)^{k/2})$ time.

\paragraph{Folklore Algorithm for the Decision Problem.}

Observe that there are $\binom{n}{k}(q-1)^{k}$ distinct $k$-sparse vectors over $(\zqz)^n$.
Recall that under both null and planted, $a_i\sim \randomsupport(n,k,\mathcal{D})$.
Split them into $O(\sqrt{q})$ groups of size $O((qn)^{k/2})$ each. 
Then, with large constant probability, in at least a (say) 1/10 fraction of the groups, one can find two indices $i_t\neq j_t$ such that $a_{i_t} =  a_{j_t}$.
Without loss of generality, assume this happens for groups $1, \ldots, s'$ for some ($s' = O(\sqrt{q})$).
Next, in each of these groups, compute $\psi_t\coloneqq y_{i_t} - y_{j_t}$.
Note that under null, we have that $\psi_t \iidsim \unif(\zqz)$, while under planted $\mathcal{E}^{*}\coloneqq \mathcal{L}(e_1 - e_{2})$ where $e_1, e_2\iidsim\mathcal{E}.$\footnote{For a random variable $X$, we denote by $\mathcal{L}(X)$ the distribution of $X$.}
Under the mild condition that $\mathcal{E}^{*}$ is at least constantly far in total variation to $\unif(\zqz)$, we can now successfully test between the two cases with constant success probability since they are supported on $q$ elements and we have $t = O(\sqrt{q})$ samples (see, e.g.,~\cite{canonne2020survey}).
We can amplify the success probability to $1-\eta$ by standard amplification, if we increase the number of equations to $O(\sqrt{q} (qn)^{k/2} \cdot \log (1/\eta))$ (and runs in the same time).

\paragraph{Limitation of Sample-to-Sample Decoding.} Of course, one should observe that we get lower bounds of the form $n^{\frac{k}{\poly(\log k, \log \log n)}}$ instead of 
$n^k.$ While this gap does not seem significant, one may still hope to close it. It turns out that at least in the case of \lpn, the gap is unavoidable using the gadget matrix approach of decoding each equation individually.
Indeed, suppose that we start from $L$-dimensional standard \lpnspace and decode each $b \in \{0,1\}^L$ individually (for simplicity, say with probability 1) into some a $k$-sparse $a\in \{0,1\}^n$ such that $a^\top G = b^\top$ for some gadget matrix $G.$ Since $b$ can be recovered from $a,$ we know that the entropy of $a$ is at least the entropy of $b.$ That is, $\log\binom{n}{k}\ge L,$ hence  
$k\log n \ge L.$ However, the hardness of the starting $L$-dimensional problem is at most $\exp(L/\log L)$ due to \cite{BKW}. Thus, the hardness of the final problem is at most $\exp(L/\log L)\le \exp(k\log n /\log (k\log n)) = n^{\frac{k}{\log k + \log \log n}}.$ In particular, this also means that one cannot hope to get hardness using this approach when $k = O(\log \log n).$ Thus, our reduction is close to optimal within the class of reductions using gadget matrices and decoding equations individually.


\paragraph{Low-Degree Polynomial Lower Bounds.}

We give evidence for the tightness of the $n^{k/2}$ behaviour in the low-degree polynomial testing framework. The proof is standard for the low-degree polynomial method, so we defer it to \cref{appendix:proofofldp}. It recovers the refutation lower-bound \cite{schoenebeck2008linear,kothari17soslbanycsp} when $k = \log(n)^{O(1)}.$
We remark that this result suggests hardness even when the error distribution is constantly 0.
This is not a mistake, as the low-degree framework is known not to capture algorithms for the noiseless case, such as Gaussian elimination or the LLL algorithm~\cite{zadik22latticebasedsurpassSoS,diakonikolas2022non}.

\begin{proposition}
\torestate{
\label{thm:genericLDP}
The detection problem with any error distribution $\mathcal{E}$ and distribution on the support $\mathcal{D}$
for 
    $
\lwespace(\textsf{modulus} = q, \mathsf{samples} \sim \randomsupport(n,k,\mathcal{D}), \mathsf{secret} \sim \unif((\zqz)^n), \mathsf{error} \sim \mathcal{E})
$
is hard for degree-$D$ polynomial tests when $m = o\big(\frac{n}{k} \times (ne^3/kD)^{k/2 -1}\big).$
}
\end{proposition}
Interpreting degree-$O(k\log n)$ polynomial tests as a proxy for algorithms running in time $n^{O(k)},$ we get $n^{O(k)}$ time hardness with $(n/k^2\polylog(n))^{k/2}$ samples. 
One must note that this result suggests that when the number of samples is significantly lower than $n^{k/2},$ say $m = \tilde{O}(n)^{(k/2-1)(1-\delta) + 1},$ the distinguishing problem is hard even for degree $D = \tilde{O}(n^{\delta})$ polynomials. We pose the following question as an open problem:
\begin{problem} Find a reduction from standard \lpnspace to $k$-sparse \lpnspace suggesting 
$\exp(\tilde{O}(n^{\delta}))$ hardness with access to only 
$m=\tilde{O}(n)^{(k/2-1)(1-\delta) + 1}$ samples.
\end{problem}
In light of the recent work \cite{chen2024algorithmssparselpnlspn}, such a reduction should be noise-aware. In the low noise-regime of $\Bern(\gamma)$ where $\gamma\le n^{-\frac{1+\delta}{2}}$ and $m \gg n^{\frac{1+\delta}{2} + k\frac{1-\delta}{2}},$ the work \cite{chen2024algorithmssparselpnlspn} exhibits a polynomial-time algorithm.

Recalling our discussion of the $n^{\frac{k}{\poly(\log k, \log \log n)}}$ versus 
$n^k$ gap, we again note that a reduction based on separately decoding individual equations cannot give the desired result.

\subsection{Application to Noisy Tensor Completion}
\label{sec:resultsontensor}

Our lower bound for tensor completion will follow from the following theorem.
It says that we can map samples from the decisional $k$-sparse \lwespace problem to samples from the following tensor distinguishing problem:
Samples from the planted distributions will be transformed into i.i.d.\ sampled entries of a noisy low-rank tensor, and samples from the planted distributions will be transformed into i.i.d.\ sampled entries of a tensor with independent entries with mean $o(1)$ and constant variance.
Note that a completion algorithm as outlined in~\cref{sec:tensorcompletion} can be used to solve this distinguishing problem, by splitting the set of input samples into two, running the algorithm on the first half, and computing the empirical error on the second half.

\begin{theorem}
\torestate{
\label{thm:tensorcompletion}
There exist an algorithm running in time $O(m\log q)$ which on input
$$
\NLE(\samples = m,\modulus = q,\Dcoeff =\randomsupport(n,k,1^k),\Derr = \mathcal{E})
$$
produces $m$ uniformly and independently sampled entries from an order-$k$ rank-$2^{k-1}$ real noisy tensor $T$ with noise $\Delta.$ Furthermore, $T,\Delta$ with high probability satisfy the following assumptions.
\begin{enumerate}
    \item \emph{Incoherence:} The noiseless tensor is $O(1)$-incoherent.
    \item \emph{Noise:} The noise $\Delta$ satisfies 
    $\|\Delta\|_1/\|T\|_1 = O(\alpha),$ where 
    $\alpha\coloneqq \expect_{e\sim \mathcal{E}}[\min(|e|, |q-e|)/q].$
\end{enumerate}
Further, the same algorithm maps $m$ i.i.d.\ samples of the form $(b_i,y_i) \sim \randomsupport(n,k,1^k) \times \unif(\zqz)$ to $m$ uniformly and independently sampled entries from a tensor in which each entry is independent, has mean $o(1)$, and constant variance (assuming that $q = \omega(1)$).
}
\end{theorem}

We summarize the implied hardness based on standard \lwespace by composing \cref{thm:tensorcompletion} with our reductions to sparse \lwe. 
In the case of tensors over $\mathbb{C}$, we can reduce the rank above to 1, see \cref{thm:tensorcompletioncomplex}.

\begin{table}[H]
\centering
\begin{tabular}{ |M{5cm}|M{5cm}|M{5cm}|  }
 \hline
 \multicolumn{3}{|c|}{Hardness of Noisy Tensor Completion based on Standard \lwespace} \\
 \hline
 \hline
 Assumed time complexity of standard \lwespace in dimension $L$ with modulus $q =L^\kappa = \poly(L)$ and error distribution $D_{\beta q},\beta q>2\sqrt{L}.$& Implied time complexity of  rank-$2^{k-1}$ order-$k$ noisy tensor completion with $O(1)$ incoherent entries&
 Restriction on the order of the tensor\\
 \hline
 \hline
  &&\\
 \ref{itm:nearexplwe}: $\exp\paren{\Omega(L)}$ & $\exp\paren{\Omega(\frac{k\log n}{\log k(\log k + \log \log n)})}$, with error $\frac{\|\Delta\|_1}{\|T\|_1} = O\big(\big(\frac{\log k(\log k + \log \log n)}{k\log n}\big)^{\kappa-1/2}\big) $ & $k = \omega((\log\log n)(\log\log\log n)),$
 $ 
 k = o(\sqrt{n})
 $
 \\
 & & \\
 
 \hline
 &&\\
  &
 $\exp\paren{\Omega((\frac{k\log n}{\log k(\log k + \log \log n)})^{1-\alpha})}$, &
 \\
\ref{itm:subexplwe}:
 $\exp\paren{\Omega(L^{1-\alpha})}$& with error $\frac{\|\Delta\|_1}{\|T\|_1} = O\big(\big(\frac{\log k(\log k + \log \log n)}{k\log n}\big)^{\kappa-1/2}\big)$ & $k = \omega( (\log n)^{\alpha/(1-\alpha)}(\log \log n)^2),$ $ k =o(\sqrt{n})$ \\
 && \\
 \hline
  &&\\
 \ref{itm:quasilwe}:
 $\exp\paren{\Omega((\log L)^{1+c})}$&
 $\exp\paren{\Omega((\log k)^{1+c})}$  with error $\frac{\|\Delta\|_1}{\|T\|_1} = O\big(\big(\frac{\log k(\log k + \log \log n)}{k\log n}\big)^{\kappa-1/2}\big).$&
 $k = \exp(\omega( (\log n)^{1/(1+c)}),$ $ k =o(\sqrt{n}).$\\
 && \\
 \hline
\end{tabular}
\end{table}

Perhaps surprisingly, we derive the essentially same time and sample-complexity (albeit with different noise and rank) from standard \lpn.
These follow directly from our reduction from standard \lpnspace to $k$-sparse \lpnspace and the discussion at the end of~\cref{sec:definingtensorcompletion}.

\begin{table}[H]
\centering
\begin{tabular}{ |M{5cm}|M{5cm}|M{5cm}|  }
 \hline
 \multicolumn{3}{|c|}{Hardness of Noisy Tensor Completion based on Standard \lpnspace} \\
 \hline
 \hline
 Assumed time complexity of standard \lpnspace in dimension $L$ with error distribution $\Bern(1/2-\gamma).$& Implied time complexity of  rank-$1$ order-$k$ noisy tensor completion with $O(1)$ incoherent entries.
 &
 Restriction on the order of the tensor\\
 \hline
 \hline
  &&\\
 \ref{itm:nearexplpn}:
 $\exp\paren{\Omega(L/\log L)}$ &
 $\exp\paren{\Omega(\frac{k\log n}{\log k(\log k + \log \log n)^2})}$, with error $\frac{\|\Delta\|_1}{\|T\|_1} = 1 -2\gamma + \tilde{O}(\frac{1}{\sqrt{m}})$&
 $k = \omega((\log\log n)^2\log\log\log n)$
 $ 
 k = o(\sqrt{n})
 $\\
  &&\\
 \hline
  &&\\
 \ref{itm:subexplpn}:
 $\exp\paren{\Omega(L^{1-\alpha})}$&
 $\exp\paren{\Omega((\frac{k\log n}{\log k(\log k + \log \log n)})^{1-\alpha})}$,  with error $\frac{\|\Delta\|_1}{\|T\|_1} = 1 -2\gamma + \tilde{O}(\frac{1}{\sqrt{m}})$&
 $k = \omega( (\log n)^{\alpha/(1-\alpha)}(\log \log n)^2),$ $ k =o(\sqrt{n}).$\\
  &&\\
 \hline
  &&\\
 \ref{itm:quasilpn}:
 $\exp\paren{\Omega((\log L)^{1+c})}$&
 $\exp\paren{\Omega((\log k)^{1+c})}$,  with error $\frac{\|\Delta\|_1}{\|T\|_1} = 1 -2\gamma + \tilde{O}(\frac{1}{\sqrt{m}})$&
 $k = \exp(\omega( (\log n)^{1/(1+c)}),$ $ k =o(\sqrt{n}).$\\
  &&\\
 \hline
\end{tabular}
\end{table}

\subsection{Downstream Applications to Learning Theory}
\label{sec:applications_learning}


The decision version of sparse \lpn, or equivalently, distinguishing noisy planted $k$-XOR instances from fully random ones, is a popular assumption to show hardness result for PAC learning and related models.
It is particularly appealing for the following reason:
Since sparse \lpnspace is a parameterized assumption, often the sparsity naturally maps to some other parameter in the learning problem, so you can understand the dependence on this parameter via the dependence of sparse \lpnspace on the sparsity.
Ruling out algorithms for smaller sparsities then naturally leads to hardness results with better dependence on the relevant parameters in the learning problem.

We briefly state the relevant definitions.
In PAC learning, or the related agnostic learning model, we have the following set-up:
Let $\cC$ be a (known) class of Boolean functions over some domain $\cX$.
In PAC learning, there is an unknown function $f^* \in \cC$ and we observe \iid samples $(x_i, f^*(x_i))$, where $x_i \sim \cD_x$ for some arbitrary (unknown) distribution $\cD_x$ over $\cX$.
The task is to output an arbitrary function $\hat{f}$ (not necessarily in $\cC$) such that $\prob_{x \sim \cD_x}(\hat{f}(x) \neq f^*(x)) \leq \varepsilon$ for some target accuracy $\varepsilon$.
The fact that we do not require that $\hat{f} \in \cC$, is referred to as \emph{improper} learning.
A weaker notion of learning, called weak learning, is to output $\hat{f}$ that achieves error better than 1/2 (or random guessing in the case of biased labels).

The agnostic model aims to incorporate noise in the PAC model.
Instead of receiving clean \iid samples $(x_i,f^*(x_i))$, we observe \iid samples $(x_i,y_i) \sim \cD$ drawn from some arbitrary \emph{joint} distribution over $\cX \times \{0,1\}$, and it is not necessarily the case that for all $y_i = f^*(x_i)$ for a function $f^* \in \cC$.
The goal is to output $\hat{f}$ (again, not necessarily in $\cC$) that competes with the best classifier in $\cC$.
Formally, let $\OPT = \inf_{f \in \cC} \prob_{(x,y) \sim \cD} (f(x) \neq y)$.
We ask to output $\hat{f}$ such that $\prob_{(x,y) \sim \cD} (\hat{f}(x) \neq y) \leq \OPT + \varepsilon$.
Sometimes (in both settings) we also consider the case where the marginal distribution of $x$ is known, e.g., standard Gaussian or uniform over the hypercube.

\paragraph{PAC Learning DNFs.}

A DNF formula over $n$ (boolean) variables is a Boolean function that is a disjunction of conjunctions.
We say a DNF formula has $s$ \emph{terms}, if there are at most $s$ disjunctions (there can be arbitrarily many conjunctions in each term).
DNF formulas are one of the simplest yet most fundamental function classes in learning theory.
In particular, learning DNFs reduces to many other natural learning problems such as learning intersections of halfspaces, agnostically learning halfspaces, and learning sparse polynomial threshold functions.
DNFs with only a constant number of terms can be learned in polynomial time \cite{valiant1984theory}.
Without restricting the number of terms, the fastest known algorithm runs in time $n^{O(n^{1/3} \log s)}$~\cite{klivans_dnfs}.

To-date, the strongest known lower bound, ruling out efficient algorithms to learn $\omega(1)$-term DNFs, was shown in~\cite{daniely2021local} under the assumption that there exist pseudo-random generators (short PRGs) with constant locality and arbitrarily large polynomial stretch.
In two very recent works~\cite{bui2024structured,ragavan24io}, it has been observed that in some applications, this assumption can be replaced by the sparse \lpnspace assumption (by introducting the so-called notion of a \emph{structured-seed local PRG}).
The work of \cite{bui2024structured} shows that this is also the case of the application to learning DNFs in~\cite{daniely2021local}.
In particular, letting $k$ be large enough, combining the construction of~\cite{daniely2021local} with their structured-seed local PRG, \cite[Section 6.4]{bui2024structured} show that, if there is no polynomial-time algorithm for $k$-sparse \lpn, there is no polynomial-time algorithm to learn DNFs with $\exp(k^{\omega(1)})$ terms.
Assuming that the assumption holds with $k = \omega(1)$, implies that there is no polynomial time algorithm to learn DNFs with $\omega(1)$ terms.
Using~\cref{cor:nearexplpn}, we can conclude the following, setting $k = \omega((\log \log n)^2 \log \log n)$:
\begin{theorem}
    \label{thm:application_dnfs}
    Let $L \in \mathbb{N}$ and $\alpha \in (0,1/2)$ be an absolute constant. Assuming the decision version of standard (dense) \lpnspace in dimension $L$ with constant noise rate requires time $2^{\Omega(L/\log L)}$, there is no polynomial-time algorithm for learning $\exp((\log \log n)^{\omega(1)})$-term DNFs over $n$ variables.
\end{theorem}


\paragraph{Agnostically Learning Halfspaces over Arbitrary Distributions.}

Daniely~\cite{daniely2016complexity} shows that under the assumption that $k$-sparse \lpnspace with constant noise rates for $k = \log^s n$ for some large enough absolute constant $s$, takes time at least $n^{\Omega(k)}$, we cannot agnostically learn halfspaces up to any error better than 1/2 in polynomial time (in dimension $N$), even if $\OPT = \exp(-\log^{1-\nu}(N))$ for some arbitrarily small constant $\nu$.
Let $\alpha > 0$ be a small enough constant (depending on $\nu$). Upon careful examination of their proof, one can see that their reduction still goes through if we only assume there is no $\exp(\Omega((k \log n)^{1-\alpha}))$-time algorithm for $k$-sparse \lpnspace when $k = \log^s n$.\footnote{Specifically, set $r = \tfrac{k^{1-\alpha}}{\log^\alpha(n)}$ and $d = \tfrac{r}{\log \log k}$ in section "Connecting the dots" in Appendix~3 of~\cite{daniely2016complexity}.}
Together with~\cref{cor:subexphardness} we obtain the following theorem:\footnote{We do not make the dependence of $\alpha$ on $\nu$ explicit, but offer the following explanation: It stems from the fact that for smaller $\nu$ we need to choose $s$ larger. Since our reduction from dense to sparse \lpnspace restricts how large $k$ can be, we need to choose $\alpha$ small enough, so that we can allow $k = \log^s n$. The exact dependence can in principle be extracted from the proofs.}
\begin{theorem}
    \label{thm:application_daniely_agnostic_LTF}
    Let $\nu > 0$ be an absolute constant.
    Then, there exists $\alpha = \alpha(\nu) < 1$ such that assuming the decision version of standard (dense) \lpnspace with constant noise rate in dimension $L$ takes time at least $2^{\Omega(L^{1-\alpha})}$, there is no polynomial-time algorithm to agnostically learn halfspaces, even if $\OPT = \exp(-\log^{1-\nu}(N))$.
\end{theorem}

\subsection{Downstream Applications to Cryptography}
\label{sec:applications_crypto}

\paragraph{PRGs with Linear Stretch and small locality.} Following the work of \cite{applebaum06prgnc0}, one can create a pseudorandom generator with linear stretch and locality $\max(k,5)$ from the hardness of \klpn{}  in polynomial time. Thus, assuming near-exponential hardness \ref{itm:nearexplpn} of standard \lpn{}, we obtain a PRG with linear stretch and locality $k$ for any $k =\omega((\log \log n)^2\log\log\log n)$, i.e. slightly super constant locality. Under the weaker subexponential hardness \ref{itm:subexplpn}, we obtain a PRG with linear stretch and locality $k$ for any $k =\omega((\log n)^{(1-\alpha)/\alpha}(\log \log n)^2)$.

\paragraph{Multi-party Homomorphic Secret Sharing~\cite{dao23multipartyhomomorphic} and Somewhat Homomorphic Encryption~\cite{cryptoeprint:2024/1760}.} The multiparty homomorphic secret sharing protocol in \cite{dao23multipartyhomomorphic} is based on the hardness of the decision problem for
\klpn{} with $k = \polylog(n)$ against poly-time distinguishers. Our results show that such hardness follows from near-exponential hardness \ref{itm:nearexplpn} or the weaker sub-exponential hardness \ref{itm:subexplpn} of standard \lpn.  The same holds for the recent work of \cite{cryptoeprint:2024/1760} who construct a somewhat homomorphic encryption scheme from sparse LPN and the decisional Diffie-Hellman assumption.

\medskip\noindent
In contrast, some other applications, most notably indistinguishability obfuscation~\cite{ragavan24io} requires the hardness of $k$-sparse LPN for constant $k$. As of now, our reductions do not have any implications to this work.

\section{Full Reduction from Dense to Sparse Equations}
\label{sec:reductiontosparsity}
Here, we make the outline in \cref{sec:nontechnicaloverview} rigorous.
Recall that the key ingredient is decoding each dense noisy linear equation individually into a sparse noisy linear equation. Formally, our procedure achieves the following guarantee.

\restatetheorem{thm:decoding_matrix_and_alg}
By performing this operation for all input samples, we obtain our full reduction from dense to sparse noisy linear equations:

\begin{algorithm}
    \caption{Full Reduction to Sparse Noisy Linear Equations}
    \label{alg:fullreduction}
     \hspace*{\algorithmicindent} 
     \textbf{Parameters:} 
     $k,L,n$ and
     \cref{alg:sentencedecoding}.\\ 
     \hspace*{\algorithmicindent} \textbf{Input:} $m$ samples $(b_i, y_i)$ from either $\NLE( q, \unif((\zqz)^L), \unif((\zqz)^L),  \mathcal{E})$ or $\unif((\zqz)^L) \times \unif(\zqz)$.\\
     \hspace*{\algorithmicindent} \textbf{ }\\
     \hspace*{\algorithmicindent} \textbf{Procedure:}\\
     \hspace*{\algorithmicindent} 
     \hspace*{\algorithmicindent} 
     \textbf{1.} Sample $z\sim \unif((\zqz)^n).$\\
     \hspace*{\algorithmicindent}\hspace*{\algorithmicindent} \textbf{2.} For $i = 1,2 \ldots, m:$\\
     \hspace*{\algorithmicindent} \hspace*{\algorithmicindent}\hspace*{\algorithmicindent} \textbf{a.} Run \cref{alg:sentencedecoding} on input $b_i$ and parameters $k,L,n.$\\
     \hspace*{\algorithmicindent} \hspace*{\algorithmicindent}\hspace*{\algorithmicindent} \textbf{b.} If \cref{alg:sentencedecoding} outputs $\fail$, discard the sample.\\
     \hspace*{\algorithmicindent}\hspace*{\algorithmicindent} \hspace*{\algorithmicindent} \textbf{c.} Else, let $a_i \in (\zqz)^n$ be its output.
     Output the sample $(a_i, \langle a_i,z\rangle + y_i)$.  
\end{algorithm}

We analyze \cref{alg:fullreduction} in \cref{sec:reductionandtasks} separately for the testing, search, and refutation tasks. 
We remark that for the refutation task it would be sufficient to output the sample $(a_i, y_i)$ in Step 1.c) above (but it still is correct as written).

Before that, we give the proof of \cref{thm:decoding_matrix_and_alg}. 
The proof follows the steps in \cref{sec:nontechnicaloverview} and is organized as follows. Just like in \cref{sec:nontechnicaloverview}, we begin with the case of binary $k$-sparse vectors.
\begin{enumerate}
    \item In \cref{sec:worddecodingprimitives}, we define the properties of the word decoding $V$ and sentence decoding $G$ matrices needed.
    \item In \cref{sec:worddecodingalgorithmfull} we define and analyze the word-decoding procedure for binary $k$-sparse vectors.
    \item In \cref{sec:sentencedecodingalgorithm} we define and analyze the sentence-decoding procedure for binary $k$-sparse vectors.
    \item In \cref{sec:complete_arbitrary_rows} we show how to use the sentence decoding for binary $k$-sparse vectors for arbitrary distributions $\mathcal{D}$ on the support.
    \item In \cref{sec:variablesupportsizes} we show how to handle distributions with variable support sizes.
\end{enumerate}

\subsection{Word Decoding Primitives and Properties}
\label{sec:worddecodingprimitives}
Here we prove the three main claims about word decoding -- \cref{lem:samplingintro,lem:nearuniformityintro,lem:sizeestimationintro} as well as the fact that our sentence-decoding matrix is invertible.

\subsubsection{Efficient Preimage Sampling: Proof of \cref{lem:samplingintro}}
\label{sec:proofofsampling}
For a vector $b \in (\zqz)^h$, recall that we defines its preimage with respect to a matrix $V$ as
\begin{equation}
\begin{split}
\Preimage_{V}(b) \coloneqq \Set{(\xi(1), \xi(2), \ldots,\xi(k))\in [hq]^k\suchthat \sum_{\psi = 1}^k V_{\xi(\psi).:} = b}.
\end{split}
\end{equation}
Further, recall that our word-decoding matrix was defined as
\begin{align*}
    V= 
    \begin{pmatrix}
      0 \times I_h\\
      1\times I_h\\
      \quad \vdots\\
      (q-1)\times I_h
    \end{pmatrix} \,.
\end{align*}
We now restate and prove~\cref{lem:samplingintro}.
\restatelemma{lem:samplingintro}

\begin{algorithm}
    \caption{Sampling From $\Preimage_V(b)$}
    \label{alg:samplingfrompreimage}
     \hspace*{\algorithmicindent} \textbf{Parameters:} $k,h$ such that $k> h+2.$ 
     \hspace*{\algorithmicindent} \textbf{ }\\
     \hspace*{\algorithmicindent} \textbf{Input:} $b\in (\zqz)^h.$ \\
     \hspace*{\algorithmicindent} \textbf{ }\\
     \hspace*{\algorithmicindent} \textbf{Initialize:} $\beta_0 = (\underbrace{0,0,\ldots, 0}_h) \in (\zqz)^h, A_0 = [h].$ 
     \hspace*{\algorithmicindent} \textbf{ }\\
     \hspace*{\algorithmicindent} \textbf{Procedure:} For $r = 0$ to $k:$\\
     \hspace*{\algorithmicindent} 
     \hspace*{\algorithmicindent}
     \textbf{1. }If $|A_r|= k-r,$ compute the $k-r$ vectors $\xi'({r+1}), \xi'({r+2}), \ldots, \xi'({k})$ corresponding\\ 
     \hspace*{\algorithmicindent} 
     \hspace*{\algorithmicindent}
     to the indices
     of
     $\{(b - \beta_r)_{j}\times e_j^\top\}_{j \in A_r}$ in $V.$ Let $\pi$ be a random permutation of $r+1,r+2, \ldots, k.$ 
     \\
     \hspace*{\algorithmicindent} 
     \hspace*{\algorithmicindent}
     Set $\xi({j})= \xi'({\pi(j)}).$
     Return 
     $
     \xi(1),\xi(2), \ldots, \xi(k).
     $
     \\
     \hspace*{\algorithmicindent} 
     \hspace*{\algorithmicindent}
     \textbf{2. }Else if $|A_r| = k-r-1.$ Take $\xi({r+1})$ to be a uniformly random element indexing one of\\
     \hspace*{\algorithmicindent} 
     \hspace*{\algorithmicindent}
     the
     zero
     rows in $V$ and  rows corresponding to coordinates where $\beta_r$ and $b$ differ: 
     $$\{1,2,\ldots, h\}\quad\bigcup
     \bigcup_{m \in A_r, j \in \{1,2,\ldots,q-1\}}\{jh + m\}.$$\\      
     \hspace*{\algorithmicindent}
     \hspace*{\algorithmicindent} 
     \textbf{3. }Else if $|A_r| \le k-r-2,$ sample $\xi({r+1})\sim \unif([hq])$.\\
     \hspace*{\algorithmicindent} 
     \hspace*{\algorithmicindent} 
     \textbf{4. }Compute $\beta_{r+1} = \beta_{r}+ V_{\xi({r+1}), :}$ and 
     $A_{r+1}= \{m \in [h]\; : \; (\beta_{r+1})_m\neq b_m\}.
     $
     \\
\end{algorithm}
The running time of~\cref{alg:samplingfrompreimage} is clearly $\poly(k,h,\log q)$ We proceed to correctness, which follows from the following two claims.

\begin{claim} The algorithm always returns. When it returns $\xi(1),\xi(2), \ldots, \xi(k),$ it holds that 
$$
\sum_{\ell= 1}^k V_{\xi(\ell),:} = b.
$$
\end{claim}
\begin{proof} 
\textbf{If the algorithm returns, it does so correctly.}
Suppose that case 1 occurs. That is, the algorithm reaches a point where $|A_r| = k-r.$ As $|A_r|\le h <k,$ it follows that $r\ge 1.$

By the definition of $\xi({r+1}), \xi({r+2}), \ldots, \xi({k}),$ it holds that 
$$
\sum_{\ell = r + 1}^k 
V_{\xi(\ell),:} = 
\sum_{\ell = r + 1}^k 
V_{\xi'(\ell),:}
=
\sum_{m\in A_r}  
(b - \beta_r)_{j}\times e_j^\top = 
b - \beta_r,
$$
as $A_{r}$ is defined as $\{m \in [h]\; : \; (\beta_{r})_m\neq b_m\}$ for $r\ge 1.$ Now, observe that case 1 can only occur once as the algorithm returns when it occurs. Hence, by the definition of $\beta_r,$ it holds that
$
\beta_r = 
\sum_{\ell = 1}^r 
V_{\xi(\ell),:}.
$
Altogether, 
$$
\sum_{\ell = 1}^k 
V_{\xi(\ell),:} = 
\sum_{\ell = 1}^r 
V_{\xi(\ell),:} + 
\sum_{\ell = r+1}^k 
V_{\xi(\ell),:} = 
\beta_r + (b - \beta_r) = b.
$$

\noindent
\textbf{The algorithm always returns.}
Now, we need to show that the algorithm always reaches case 1. Track the quantity $k-r- |A_r|.$ When $r = 0,$ this quantity is $k - h>0.$ When $r = k$ (if this point is ever reached), 
$k - r - |A_r| = - |A_k|\le 0.$

We make the following simple observation that $k-r- |A_r|\ge 0$ at any point of time during the algorithm. This follows by induction and case analysis of which of 2 and 3 is invoked:
\begin{enumerate}
    \item If case 2 is invoked, that is, 
    $k-r- |A_r| = 1,$ then 
    $k-(r+1)- |A_{r+1}| \in\{0, 1\}.$
    The vector $\beta_{r+1}$ coincides with $\beta_r$ outside of 
    coordinates indexed by $A_r = \{m \in [h]\; : \; (\beta_{r})_m\neq b_m\},$ so $A_{r+1}\subseteq A_r$ and 
    $k-(r+1)- |A_{r+1}|\ge k - (r+1) - |A_r| = 0.$ Note also that $\beta_r,\beta_{r+1}$ differ on at most one coordinate, so 
    $|A_{r+1}|\ge |A_r|-1.$ Thus, 
    $k-(r+1)- |A_{r+1}|\le k - (r+1) - (|A_r|-1)=k - r-|A_r| = 1.$
    \item If case 3 is invoked, that is, 
    $k-r- |A_r| \ge 2,$ then $k-(r+1)- |A_{r+1}| \ge 0.$ Indeed, note that $\beta_r$ and $\beta_{r+1}$ differ on at most one coordinate. Thus, $|A_{r+1}|\le |A_r| + 1$ and, so, 
    $k-(r+1)- |A_{r+1}| \ge k-(r+1)- (|A_{r}|+1)\ge  0.$
\end{enumerate}
In particular, this means that if the algorithm has not returned by step $k$ yet, $0 \le k - k - |A_k| = - |A_k|\le 0.$ Hence, $|A_k| = 0.$ Thus, it will return at step $k.$\end{proof}

So far, we know that the algorithm always returns correctly. We need to show that its output is uniform over $\Preimage_V(b).$ This follows immediately from the following claim.
\begin{claim} Consider any $(\xi''(1), \xi''(2), \ldots,\xi''(k))\in \Preimage_V(b).$ There exists a unique realization of the randomness used in \cref{alg:samplingfrompreimage} such that the algorithm outputs $\xi''(1), \ldots, \xi''(k).$    
\end{claim}
\begin{proof} We define the following quantities. Let $\beta''_r = \sum_{j = 1}^rV_{\xi''(j),:}$ and $A''_r= \{m\in [h]\; : \; (\beta_r)_m\neq b_m\}.$ Note that $\beta''_k = b,A''_k = \emptyset.$
We make the following observation:
\begin{observation}
\label{obs:differonfew}
For any $r,$ the vectors $\beta''_r$ and $b = \beta''_k$ differ on at most $k- r$ coordinates.
\end{observation}
\begin{proof}
    $
    \beta''_k - \beta''_r = 
    \sum_{j = r+1}^kV_{\xi''(j),:},
    $
    which is a sum of at most $k-r$ unit vectors.
\end{proof}
Let also
$$
Z = \min\Set{ z \in [k]\; \suchthat
k - r - |A''_{r}| = 0
}.
$$
This is a well-defined quantity since clearly $A''_k = \emptyset$ by the definition of $\Preimage_V(b)$ and, thus, for $k = r,$
$k - r - |A''_r| = 0.$
Define also 
$$
T = \min\Set{ t \in [k]\; \suchthat
k - t- |A''_{t}| \le 1
}.
$$
Again, $T$ is well-defined and $T\le Z.$
We now claim that the unique randomness that leads to $(\xi(1), \xi(2), \ldots,\xi(k)) = (\xi''(1), \xi''(2), \ldots,\xi''(k))$ in \cref{alg:samplingfrompreimage} is the following.
\begin{enumerate}
    \item On steps $r = 0,1,2,\ldots, T-1,$ we are in case 3 and 
    $\xi(r+1) = \xi''(r+1).$
    \item On (the potentially empty set of) steps $r = T, T+1, \ldots, Z-1,$ we are in case 2 and $\xi(r+1) = \xi''(r+1).$
    \item On steps $r = Z,$ we are in case 1. Then, 
    the sampled $\xi'(r+1), \ldots, \xi'(k)$ are all distinct and are a permutation of $\{\xi''(r+1), \ldots, \xi''(k)\}.$ The sampled permutation $\pi$ is the unique permutation of $\{r+1,\ldots, k\}$ such that $\xi''(j) = \xi'(\pi(j)).$ 
\end{enumerate}
We prove these claims by analyzing the algorithm step by step.
\begin{enumerate}
    \item \textbf{Steps $r = 0,1,2,\ldots, T-1.$} Note that initially, $k - 0 - |A_0| =k - s\ge 2.$ Hence, case 1 is invoked. It must be the case that $\xi(1) = \xi'(1).$ As 
    $k - t- |A''_{t}| \ge 2 $ for $t \in \{1,2,\ldots, T-1\}$ by the definition of $T,$ case 1 must also be invoked at steps $0,1,2,\ldots, T-1.$
    \item 
    \textbf{Steps $r = T, T+1, \ldots, Z-1.$}
    Suppose that $T <Z$ (that is, $\{T, T+1, \ldots, Z-1\}$ is a non-empty set of steps).
    Consider step $T.$ At this point of time, $k - T - |A''_T|= 1.$ Now, we claim that $\xi''(T+1)$ must be either indexing a zero row in $V$ or a row in which $\beta_T, b$ differ, i.e.
    $$\{1,2,\ldots, h\}\quad\bigcup
     \bigcup_{m \in A_r, j \in \{1,2,\ldots,q-1\}}\{jh + m\}.$$
    Indeed, if indexes some other row $m',$ then $\beta''_{T+1}$ and $b$ would differ in $A''_T\cup \{m'\}$ so $|A''_{T+1}|= |A''_T| + 1.$ Thus, $k - |A''_{T+1}| - (T+1) = k - |A''_T| - T - 2 = -1.$ This is impossible since $\beta''_{T+1}$ and $\beta''_k = b$ differ on at most $k - (T+1)$  coordinates by \cref{obs:differonfew}, meaning that 
    $|A''_{T+1}|\le k - (T+1).$  

    Altogether, this means that indeed 
    $$
    \xi''(T+1)\in
    \{1,2,\ldots, h\}\quad\bigcup
     \bigcup_{m \in A_r, j \in \{1,2,\ldots,q-1\}}\{jh + m\},
    $$
    which is exactly the admissible set for $\xi(T+1)$ by \cref{alg:samplingfrompreimage}. Thus, there is a unique choice of $\xi(T+1).$

    Finally, by the definitions of $Z,T,$ it must be the case that 
    $$
    k - r - |A_r| = 1\text{ for }T\le r \le Z-1,
    $$
    so there is similarly a unique choice at each step $ T+1, T+2,\ldots, Z-1.$

    \item
    \textbf{Step $Z.$}
    Finally, suppose that $r = Z.$ Thus, 
    $0 = k- r- |A''_r|.$
    Note that at this point of time, $\beta''_r$ and $b = \beta''_k$ differ on exactly $k - r$ coordinates. However,
    $$
    \beta''_k - \beta''_r =  \sum_{j = r+1}^kV_{\xi''(j),:}
    $$
    is a sum of at most $k-r$ unit vectors. This means that  $\beta''_k - \beta''_r$ is a sum of \emph{exactly} $k-r$ unit vectors, which are in some arbitrary order 
    $$
    \{(b - \beta''_r)_j \times e_j^\top\}_{j \in A''_r}.
    $$
    As in step 1 of \cref{alg:samplingfrompreimage}, the algorithm correctly computes indices $\xi'(r+1), \ldots, \xi'(k)$ corresponding to these vectors.
    In particular, $\xi''(r+1), \ldots, \xi''(k)$ is some permutation of $\xi'(r+1), \ldots, \xi'(k).$
    As $\xi'(r+1), \ldots, \xi'(k)$ are distinct (since $A_r$ indexes $k-r$ distinct unit vectors $e_j^\top$), there is a unique permutation $\pi$ that sends $\xi'(r+1), \ldots, \xi'(k)$ to $\xi''(r+1), \ldots, \xi''(k).$
\end{enumerate}

Altogether, we have shown that indeed there is a unique choice of the random bits used by the algorithm that yield $\xi''(1), \xi''(2), \ldots, \xi''(k).$ This is sufficient.
\end{proof}

We also derive the following corollary from the proof. 

\begin{corollary}
\label{cor:marginaldistributionofwalk}
Suppose that $k > h+2$ and let $u \in (\zqz)^h.$ Let $K\subseteq [k]$ be such that $|K|\le k - h- 2.$ Then, if $(\xi(1), \xi(2), \ldots, \xi(k))\sim \unif(\Preimage_V(u)),$ it holds that
$$
(\xi(j))_{j \in K}\text{ is uniform over }[hq]^{|K|}.
$$
\end{corollary}
\begin{proof}
    Due to symmetry, it is enough to consider $K = \{1,2,\ldots, k - h- 2\}.$ Note, however, that in \cref{alg:samplingfrompreimage}, during the first $k-h-2,$ only case 3 is invoked. This is true since $|A_r|\le h \le k - r-2$ whenever $r \le k - h- 2.$ Hence, $\xi(1), \xi(2), \ldots, \xi(k-h-2)$ are uniform over $[hq]^{k-h-2}.$ 
\end{proof}

\subsubsection{Computing Preimage Sizes: Proof of \cref{lem:sizeestimationintro}}
\label{sec:estimationproof}
We next proof~\cref{lem:sizeestimationintro}, restated below.
$V$ and $\Preimage_{V}(b)$ are as in the previous section.
\restatelemma{lem:sizeestimationintro}

\begin{algorithm}
    \caption{Computing $|\Preimage_V(b)|.$}
    \label{alg:sizeestimation}
     \hspace*{\algorithmicindent} \textbf{Parameters:} $k,h$ such that $k>s.$ 
     \hspace*{\algorithmicindent} \textbf{ }\\
     \hspace*{\algorithmicindent} \textbf{Input:} $b\in (\zqz)^h.$ 
     \hspace*{\algorithmicindent} \textbf{ }\\
     \hspace*{\algorithmicindent} \textbf{Preprocessing: }
     Compute $\Delta = |\{ m \in [h]\; : \; b_m\neq 0\}|.$
     \\
     \hspace*{\algorithmicindent} \textbf{ }\\
     \hspace*{\algorithmicindent} \textbf{Initialize: } Array 
     $f$ indexed by $\{-1,0,1,2,\ldots, h,h+1\}\times \{0,1,2,\ldots, k\}$ where $f[0,0] = 1$ and\\
     \hspace*{\algorithmicindent}
     all other values are equal to $0.$
     \\
     \hspace*{\algorithmicindent} \textbf{Procedure:}\\
     \hspace*{\algorithmicindent}
     For $r = 1$ to $k:$\\
     \hspace*{\algorithmicindent} \hspace*{\algorithmicindent} \hspace*{\algorithmicindent} For $t = 0$ to $h:$\\
     \hspace*{\algorithmicindent} \hspace*{\algorithmicindent} \hspace*{\algorithmicindent}
     \hspace*{\algorithmicindent} 
     Compute $f(t,r) = t \times f(t-1,r-1) + 
   (h + t(q-2))\times f(t,r-1) + 
   (h-t)(q-1) \times f(t+1,r-1).$\\
     \hspace*{\algorithmicindent} \textbf{Return} $f(\Delta, k).$
\end{algorithm}
Again, the algorithm clearly works in time $\poly(k,h,\log q).$ We proceed to correctness. We claim that $f(t,r)$ computes exactly 
$$
\Big|\Set{
\xi(1), \xi(2), \ldots, \xi(r)\in [hq]^r\suchthat
\sum_{j = 1}^r V_{\xi(j),:} = \beta}
\Big|
$$
where $\beta\in (\zqz)^h$ is an arbitrary vector with $t$ non-zero coordinates. 

We first need to show that this quantity is indeed well-defined. 

\begin{lemma} Suppose that $\beta_1,\beta_2\in (\zqz)^h$ are two vectors both of which have exactly $t$ non-zero coordinates. Then, 
$$
\Big|\Set{
\xi(1), \xi(2), \ldots, \xi(r)\in [hq]^r\suchthat
\sum_{j = 1}^r V_{\xi(j),:} = \beta_1}
\Big| = 
\Big|\Set{
\xi(1), \xi(2), \ldots, \xi(r)\in [hq]^r\suchthat
\sum_{j = 1}^r V_{\xi(j),:} = \beta_2}
\Big|.
$$
\end{lemma}
\begin{proof} Let 
\begin{align*}
    & \beta_1 = \sum_{\ell = 1}^t j^1_\ell \times e_{\kappa(1,\ell)}^\top,\quad\quad \beta_2 = \sum_{\ell = 1}^t j^2_\ell \times e_{\kappa(2,\ell)}^\top,\\
\end{align*}
where $j^u_\ell \in (\zqz)\backslash\{0\}$ are the non-zero coefficients and $\kappa(u,\ell)\in [h]$ are the corresponding coordinates. 

We construct the following bijection between 
$
\Preimage_V(\beta_1)$ and
$\Preimage_V(\beta_2).$

Take some $(\xi(1), \xi(2), \ldots, \xi(k))\in \Preimage_V(\beta_1).$ Suppose that $\xi(i)$ corresponds to the vector $\gamma_i e^\top_{m_i}.$ Now, define $\xi'(i)$ as follows:
\begin{enumerate}
    \item If $m_i\not \in \{\kappa(1,1), \kappa(1,2),\ldots, \kappa(1,t)\},$ then 
    $\xi'(i) \coloneqq \xi(i).$
    \item If $m_i = \kappa(1,\ell)$ for some $\ell$ and $m_j\neq m_i$ for any $j<i,$ then define $\xi'(i)$ as the index corresponding to 
    $(\gamma_i + j_\ell^2 - j_\ell^1)\times e_{\kappa(2,\ell)}^\top.$
    \item If $m_i = \kappa(1,\ell)$ for some $\ell$ and there exists some $j<i$ such that $m_j = m_i,$ then define $\xi'(i)$ as the index corresponding to 
    $\gamma_i\times e_{\kappa(2,\ell)}^\top.$
\end{enumerate}
One can trivially check that 
$(\xi'(1), \xi'(2), \ldots, \xi'(k))\in \Preimage_V(\beta_2)$ and the constructed map $\xi\longrightarrow\xi'$ is a bijection between 
$
\Preimage_V(\beta_1)$ and
$\Preimage_V(\beta_2).$
\end{proof}

With this in mind, $\Big|\Set{
\xi(1), \xi(2), \ldots, \xi(r)\in [hq]^r\suchthat
\sum_{j = 1}^r V_{\xi(j),:} = \beta}
\Big|
$
where $\beta\in (\zqz)^h$ is an arbitrary vector with $t$ non-zero coordinates is a well-defined quantity. Now, we can show by induction on $r$ that $f(t,r)$ computes this quantity. The base case $r = 0$ is clear. Now, we proceed to induction. Take any $\beta\in (\zqz)^h$ with $t$ non-zero coordinates and  
$$
\xi(1), \xi(2), \ldots, \xi(r)\in [hq]^r\text{ such that }
\sum_{j = 1}^r V_{\xi(j),:} = \beta.
$$
Consider $\xi(r).$ There are 4 different possibilities:
\begin{enumerate}
    \item If $V_{\xi(r),:} = \underbrace{0,0,\ldots, 0}_r.$ There are $h$ such choices for $\xi(r).$ For each of them, note that 
    $\sum_{j = 1}^{r-1} V_{\xi(j),:} = \beta.$ By induction, there are $f(t,r-1)$ for $\xi(1), \xi(2), \ldots, \xi(r-1).$
    \item $V_{\xi(r),:}$ is a unit vector with support disjoint from $\beta.$ Note that since $\beta$ has support of size $t,$ there are $h-t$ choices for the location of the non-zero coordinate of $V_{\xi(r),:}$ and $q-1$ choices for its value. For each of these $(q-1)(h-t)$ choices, $\sum_{j = 1}^{r-1} V_{\xi(j),:}$ has support of size $t+1.$ Hence, there are 
    $f(t+1, r-1)$ ways to choose $\xi(1), \xi(2), \ldots, \xi(r-1).$
    \item $V_{\xi(r),:} = me_j^\top$ is a unit vector such that 
    $j$ is in the support of $\beta,$ but $\beta_j \neq m.$ There are 
    $t$ choices for $j$ and, for each of them,
    $q-2$ choices for $m.$ Furthermore, $\sum_{j = 1}^{r-1} V_{\xi(j),:}$ has support of size $t.$ Hence, there are 
    $f(t, r-1)$ ways to choose $\xi(1), \xi(2), \ldots, \xi(r-1).$
    \item $V_{\xi(r),:} = me_j^\top$ is a unit vector such that 
    $j$ is in the support of $\beta$ and $\beta_j =  m.$ There are $t$ choices for $j$  and for each, a unique choice for $m.$ Furthermore, $\sum_{j = 1}^{r-1} V_{\xi(j),:}$ has support of size $t-1.$ Hence, there are 
    $f(t-1, r-1)$ ways to choose $\xi(1), \xi(2), \ldots, \xi(r-1).$   
\end{enumerate}
    Combining all four cases, we conclude that indeed
    \begin{align*}
    f(t,r) &=  
    h\times f(t,r-1)
   + (h-t)(q-1)\times f(t+1, r-1)
   + t(q-2)\times f(t, r-1) 
   + t\times f(t-1,r-1)\\
   & = 
   t \times f(t-1,r-1) + 
   (h + t(q-2))\times f(t,r-1) + 
   (h-t)(q-1) \times f(t+1,r-1).
    \end{align*}

\subsubsection{Uniformity of Preimage Size: Proof of \texorpdfstring{\cref{lem:nearuniformityintro}}{near uniformity}}
\label{sec:randomwalkproof}
We next proof~\cref{lem:nearuniformityintro}, restated below.
\restatelemma{lem:nearuniformityintro}
We will rephrase \cref{lem:nearuniformityintro} in terms of a random walk over $(\zqz)^h.$ At each time step $j,$ the walk chooses a uniformly random direction $e_i^\top$ and  uniformly random $m \in \zqz$ and takes a step of length $m$ in direction $e_i^\top.$ This corresponds to $\xi(j) = V_{\xi(j),:} = me_i^\top = V_{mh + i,:}.$ Our goal is to show the following point-wise convergence of the random walk.

\begin{lemma} Consider the following random walk over $(\zqz)^h$ starting at $x_0= (0,0,\ldots, 0).$ At each step $j,$ a uniformly random $i \in [h]$ and $m\in \{0,1,2,\ldots,q-1\}$ are drawn. Then,
$x_j = x_{j-1} + me_i.$

Suppose that $k \ge 4h(\log h+\log q +  \log (1/\psi))$ where $\psi\in (0,1)$ is any number. Then, for each $u \in (\zqz)^h,$
$$
\frac{1-\psi}{q^h}\le
\prob[x_k = u]\le \frac{1+\psi }{q^h}.
$$
\end{lemma}
\begin{proof}
Denote by $i_1,i_2, \ldots, i_k$ the respective draws of $i$ in the random walk.

\noindent
\textbf{Lower Bound.} Notice that conditioned on the fact that $\{i_1,i_2, \ldots, i_k\} = [h],$ $x_k$ is uniform over $(\zqz)^h.$ This is true since to each coordinate, we have added (at least once) a uniform element from $(\zqz)^h.$ Hence, 
$$
\prob[x_k = u]\ge 
\frac{1}{q^h}
\prob\big[\{i_1,i_2, \ldots, i_k\} = [h]\big].
$$
Now, we bound the right-hand side as follows. 
\begin{align*}
& \prob\big[\{i_1,i_2, \ldots, i_k\} = [h]\big]\ge 
1- \sum_{i \in [h]}\prob\big[i \not \in \{i_1,i_2, \ldots, i_k\}\big]\\
& =
1 - h\times
\prob\big[1 \not \in \{i_1,i_2, \ldots, i_k\}\big] = 1 - h\times \prod^k_{j =1}\prob[i_j \neq 1]\\
& = 1 - h \times (1-1/h)^k \ge 1 - h \times (1-1/h)^{2h(\log h + \log \psi^{-1})}\\
& \ge 1 - h\times e^{-(\log h + \log \psi^{-1})} =
1- \psi.
\end{align*}

\noindent
\textbf{Upper Bound.} Suppose that $u$ has support $A = \{i \in [h]\;: \; u_i\neq 0\}.$ Hence, $A\subseteq \{i_1,i_2, \ldots, i_k\}.$
It follows that 
\begin{align*}
    & \prob[x_k = u] = \sum_{B \subseteq [h]\; : \; A\subseteq B}
    \prob[x_k = u|\{i_1,i_2, \ldots, i_k\} = B]\times 
    \prob[\{i_1,i_2, \ldots, i_k\} = B].
\end{align*}
Now, observe that if 
$A\subseteq B = \{i_1,i_2, \ldots, i_k\}, $ then
$x_k|_{B}\sim \unif((\zqz)^{|B|}).$ Furthermore, over the complement, we have $x_k|_{[h]\backslash B} = 0.$ Thus, $\prob[x_k = u|\{i_1,i_2, \ldots, i_k\} = B] = 1/q^{|B|}$ as the values on $\{i_1,i_2, \ldots, i_k\}$ should match. 

On the other hand, 
$$
\prob[\{i_1,i_2, \ldots, i_k\} = B]\le \big(\frac{|B|}{h}\big)^k
$$
since each $i_j$ needs to be in the set $B.$
Altogether, 
\begin{align*}
    & \prob[x_k = u]
    = 
    \sum_{B\subseteq [h]\; : \; A\subseteq B}
    \frac{1}{q^{|B|}}\big(\frac{|B|}{h}\big)^k
    \le 
    \sum_{B\subseteq [h]}
    \frac{1}{q^{|B|}}\big(\frac{|B|}{h}\big)^k
     = \sum_{t = 0}^h 
    \binom{h}{t}
    \frac{1}{q^t}
    \big(\frac{t}{h}\big)^k =
    \frac{1}{q^h} +
    \sum_{t = 1}^{h-1}
    \binom{h}{t}
    \frac{1}{q^t}
    \big(\frac{t}{h}\big)^k.
\end{align*}

To bound the last expression, we will first show that $t\longrightarrow \binom{h}{t}
    \frac{1}{q^t}
    \big(\frac{t}{h}\big)^k$ is increasing on $\{1,\ldots,h-1\}.$ Indeed, this is equivalent to 
\begin{align*}
    & \binom{h}{t}
    \frac{1}{q^t}
    \big(\frac{t}{h}\big)^k \le 
    \binom{h}{t+1}
    \frac{1}{q^{t+1}}
    \big(\frac{t+1}{h}\big)^k \Longleftrightarrow\\
    & q\le \big(\frac{t+1}{t}\big)^k\times \frac{\binom{h}{t+1}}{\binom{h}{t}}\Longleftrightarrow\\
    & 
    q\times (t+1)/(h-t)\le \big(\frac{t+1}{t}\big)^k
\end{align*}
As $t\in \{1,2,\ldots,h-1\},$ it is enough to show that  
$
qh\le (1 + 1/h)^{k}.
$ This follows immediately as $k\ge 4h(\log h+ \log q).$ Altogether, 
\begin{align*}
    & \prob[x_k = u]\le \frac{1}{q^h} +
    \sum_{t = 1}^{h-1}
    \binom{h}{t}
    \frac{1}{q^t}
    \big(\frac{t}{h}\big)^k\\
    & \le 
    \frac{1}{q^h} + (h-1)
    \binom{h}{h-1}
    \frac{1}{q^{h-1}}
    \big(\frac{h-1}{h}\big)^k\\
    &  = 
    \frac{1}{q^h}\Bigg(1 + qh(h-1)\big(\frac{h-1}{h}\big)^k\Bigg)\\
    & \le 
    \frac{1}{q^h}\Bigg(1 + qh(h-1)\big(\frac{h-1}{h}\big)^{4h(\log h + \log q + \log(1/\psi))}\Bigg)\\
    &\le
    \frac{1}{q^h}\big(1 + qh(h-1) e^{-2\log hq/\psi}\big)\le 
    \frac{1 + \psi}{q^h}.\qedhere
\end{align*}
\end{proof}

\subsubsection{Full Rank of Sentence-Decoding Matrix}
Our last step will be to show that our sentence-decoding matrix $G$ is full rank, which is needed for the search reduction. 
Recall (cf.~\eqref{eq:gdefintro}) that $G\in (\zqz)^{n \times L}$ was defined by indexing rows by $ (\phi_1, \ldots, \phi_t) \in [hq]^t$ and each row is defined as $G_{(\phi_1, \phi_2,\ldots, \phi_t),:} = (V_{\phi_1,:},V_{\phi_2,:},\ldots, V_{\phi_t,:})$.
\begin{lemma}
\label{lem:Ginvertibility}
The matrix $G\in (\zqz)^{n\times L}$ defined as above is of rank $L.$
\end{lemma}
\begin{proof} We simply show that $I_L$ is a submatrix of $G.$ Indeed, note that for any $u \in [L],$ where $u = u_1h + u_2$ for $0\le u_1\le t-1, 1\le u_2\le h,$ 
$$
G_{\underbrace{1,1,\ldots, 1}_{u_1},u_2,\underbrace{1,1,\ldots, 1}_{t - u_1 -1},:} = 
\underbrace{0,0,\ldots ,0}_{u_1h + u_2 -1},1,\underbrace{0,0\ldots ,0}_{th - u_1h - u_2}.
$$    
\end{proof}

\subsection{Word Decoding}
\label{sec:worddecodingalgorithmfull}
Now, we show how to decode a single $b\sim\unif((\zqz)^h)$ using a good word-decoding matrix.

\begin{algorithm}
    \caption{Word Decoding}
    \label{alg:worddecoding}
     \hspace*{\algorithmicindent} \textbf{Parameters:} $k,h.$\\ 
     \hspace*{\algorithmicindent} \textbf{ }\\
     \hspace*{\algorithmicindent} \textbf{Input:} $b\in (\zqz)^h.$\\
     \hspace*{\algorithmicindent} \textbf{ }\\
    \hspace*{\algorithmicindent} \textbf{Procedure:}\\
     \hspace*{\algorithmicindent} \textbf{1.} Compute $|\Preimage_V(b)|$ using \cref{alg:sizeestimation}.\\
      \hspace*{\algorithmicindent} \textbf{2.} 
      Output $\fail$ with probability $1 - \frac{|\Preimage_V(b)|}{\frac{(1+(ht)^{-2})(hq)^k}{q^h}}$ and terminate.\\
      \hspace*{\algorithmicindent} \textbf{3.} 
      Output a uniformly random $(\xi(1), \xi(2), \ldots, \xi(k))$ from $\Preimage_V(b)$ using \cref{alg:samplingfrompreimage}.\\
\end{algorithm}

\begin{lemma}
\label{lemma:worddecodingguarantee}
Suppose that $k \ge 4h(\log h + \log q + 2\log(ht)).$ Then, \cref{alg:worddecoding} satisfies the following
\begin{enumerate}
    \item Given $b \in (\zqz)^h$, it runs in time $\poly(k,h)$ and either outputs $\fail$ or outputs a $k$-tuple\\ $(\xi(1), \xi(2), \ldots, \xi(k)) \in [hq]^k$ such that $\sum_{i=1}^k v_{\xi(i)} = b^\top$. Furthermore, it outputs $\fail$ with probability at most $2\times (ht)^{-2}.$
    \item If $b \sim \unif((\zqz)^L)$, conditioned on  \cref{alg:worddecoding} not failing, its output $(\xi(1), \xi(2), \ldots, \xi(k))$ is distributed uniformly on $[hq]^k$. 
\end{enumerate}
\end{lemma}
\begin{proof} The running time is $\poly(h,k)$ by \cref{lem:samplingintro,lem:sizeestimationintro}.
The property $\sum_{i=1}^k v_{\xi(i)} = b^\top$ follows by definition of $\Preimage_V(b)$.

By \cref{lem:nearuniformityintro} for $\psi = (ht)^{-2},$
the failure probability is at most $1 - \frac{\frac{(1-(ht)^{-2})(hq)^k}{q^h}}{\frac{(1+(ht)^{-2})(hq)^k}{q^h}}= 1 - \frac{1 - (ht)^{-2}}{1 + (ht)^{-2}} \le 2(ht)^{-2}.$
Note also that $1 - \frac{|\Preimage_V(b)|}{\frac{(1+(ht)^{-2})(hq)^k}{q^h}} \geq 0$, so it is a well-defined probability. 

Finally, we analyze the marginal distribution of the output for input $b\sim \unif((\zqz)^h).$ For any fixed  $(\tilde{\xi}(1), \tilde{\xi}(2), \ldots, \tilde{\xi}(k)),$ we compute:
\begin{align*}
    & \prob[\mathsf{output} = (\tilde{\xi}(1), \tilde{\xi}(2), \ldots, \tilde{\xi}(k))]\\
    & = 
    \prob[\mathsf{output} = (\tilde{\xi}(1), \tilde{\xi}(2), \ldots, \tilde{\xi}(k)), \; \text{Algorithm does not fail at 2.},\; 
    b^\top = \sum_{\psi = 1}^k V_{\tilde{\xi}(\psi),:}
    ]\\
    & = 
    \prob[\mathsf{output} = (\tilde{\xi}(1), \tilde{\xi}(2), \ldots, \tilde{\xi}(k))| \; \text{Algorithm does not fail at 2.},\; 
    b^\top = \sum_{\psi = 1}^k V_{\tilde{\xi}(\psi),:}
    ]\times\\
    & \quad\quad\quad\quad\quad\quad\times
    \prob[\text{Algorithm does not fail at 2.}| 
    b^\top = \sum_{\psi = 1}^k V_{\tilde{\xi}(\psi),:}
    ]\times \\
    & \quad\quad\quad\quad\quad\quad\times
    \prob[b^\top = \sum_{\psi = 1}^k V_{\tilde{\xi}(\psi),:}]\\
    & = \frac{1}{|\Preimage_V(b)|}\times 
    \frac{|\Preimage_V(b)|}{\frac{(1+(ht)^{-2})(hq)^k}{q^h}}\times 
    \frac{1}{q^h} = \frac{1}{(1+(ht)^{-2})(hq)^k}.
\end{align*}
So far, we used the facts that, conditioned on not failing, the algorithm outputs a uniform sample from $\Preimage_V(b).$ The algorithm does not fail with probability  $\frac{|\Preimage_V(b)|}{\frac{(1+(ht)^{-2})(hq)^k}{q^h}}.$ Finally, for any fixed $(\tilde{\xi}(1), \tilde{\xi}(2), \ldots, \tilde{\xi}(k)),$ we have $\prob[b^\top = \sum_{\psi = 1}^k V_{\tilde{\xi}(\psi),:}] = 1/q^h$ over $b\sim \unif((\zqz)^h).$ 

As the quantity $\frac{1}{(1+(ht)^{-2})(hq)^k}$ is uniform over all choices of $(\tilde{\xi}(1), \ldots, \tilde{\xi}(k)),$ the output distribution is uniform over $[hq]^k$. 
\end{proof}

\subsection{Sentence Decoding}
\label{sec:sentencedecodingalgorithm}
We can use the word-decoding primitive \cref{alg:worddecoding} to decode sentences, and prove~\cref{thm:decoding_matrix_and_alg} as follows.
As mentioned before, we first describe the case when $\cD = 1^k$, i.e., the entries on the support are deterministically 1.
We handle general distributions in the next section.
\begin{algorithm}[H]
    \caption{Binary Sentence Decoding}
    \label{alg:sentencedecoding}
     \hspace*{\algorithmicindent} \textbf{Parameters:} $h,k,t.$
     \\ 
     \hspace*{\algorithmicindent} \textbf{ }\\
     \hspace*{\algorithmicindent} \textbf{Input:} $b\in (\zqz)^{ht}.$\\
     \hspace*{\algorithmicindent} \textbf{ }\\
    \hspace*{\algorithmicindent} \textbf{Procedure:}\\
     \hspace*{\algorithmicindent} \textbf{1.} Split $b$ into $(b_1, b_2, \ldots, b_t)$ where $b_i \in (\zqz)^h.$\\
     \hspace*{\algorithmicindent} \textbf{2.} For $j = 1, 2 ,\ldots, t:$\\ 
     \hspace*{\algorithmicindent} \hspace*{\algorithmicindent} \textbf{a.} Run \cref{alg:worddecoding} on input $b_j$ with parameters $k,h.$\\
     \hspace*{\algorithmicindent} \hspace*{\algorithmicindent} \textbf{b.}
     If \cref{alg:worddecoding} outputs $\fail$, output $\fail$ too and terminate.\\
     \hspace*{\algorithmicindent} \hspace*{\algorithmicindent} \textbf{c.} Else, store the output
     $(\xi_j(1), \xi_j(2), \ldots, \xi_j(k)) \in [hq]^k$.\\
    \hspace*{\algorithmicindent} \textbf{3.} Set $\xi(i) = (\xi_1(i), \xi_2(i), \ldots, \xi_t(i)) \in [hq]^t$ for $1\le i \le k.$\\
    \hspace*{\algorithmicindent} \textbf{4.} If there exist $i_1\neq i_2$ such that $\xi(i_1)= \xi(i_2),$ output $\fail$ and terminate.\\
    \hspace*{\algorithmicindent} \textbf{5.} Else, output $a \in (\zqz)^{(hq)^t}$ such that $a_{\xi(i)} = 1$ for $1 \leq i \leq k$ and $a_j = 0$ for all other $j$.\\
\end{algorithm}
\cref{thm:decoding_matrix_and_alg} for the case of $\cD = 1^k$ directly follows from the following lemma.
\begin{lemma}
\label{lemma:sentencedecodingguarantee}
Suppose that $k \ge 4h(\log h + \log q + 2\log(ht))$ and $k = o(\sqrt{n}) = o((hq)^{t/2}).$
Then \cref{alg:sentencedecoding} satisfies the following
\begin{enumerate}
    \item Given $b \in (\zqz)^L$, it runs in time $\poly(n)$ and either outputs $\fail$ or outputs a $k$-sparse $a \in (\zqz)^n$ such that $a^\top G = b^\top$. Furthermore, the algorithm outputs $\fail$ with probability at most $2/L + k^2/n.$
    \item If $b \sim \unif((\zqz)^L)$, 
    conditioned on \cref{alg:sentencedecoding} not failing,
    its output $a$ is distributed according to $\randomsupport(n,k,1^k)$. 
\end{enumerate}
\end{lemma}
\begin{proof} The running time follows from \cref{lemma:worddecodingguarantee} as in Step 2 there are $t$ calls to \cref{alg:worddecoding}, each taking time $\poly(k,h) = \poly(n)$ and all other steps can be implemented to run in this time as well.
Conditioned on not outputting $\fail$, by Step 4, the vector $a$ we output is exactly $k$ sparse.
Further, using the guarantees of~\cref{alg:worddecoding} (cf.~\cref{lemma:worddecodingguarantee}) it holds that
$$
    a^\top G = \sum_{i=1}^k G_{\xi(i),:} = \sum_{i=1}^k (V_{\xi_1(i),:},V_{\xi_2(i),:},\ldots, V_{\xi_t(i),:}) =
    (b_1^\top, b_2^\top, \ldots, b_t^\top) = b^\top \,.
$$

By \cref{lemma:worddecodingguarantee}, each of the $t$ iterations in Step 2 passes independently with probability at least $1 - 2\times (ht)^{-2}$.
By union bound, the algorithm does not terminate before step 3 with probability at least $1 - 2\times h^{-2}t^{-1}\ge 1 - 2/L$ for any $b.$ Conditioned on not terminating by Step 3, the final support indices 
are the columns $\xi(1),\xi(2), \ldots, \xi(k)$ of the matrix
$$\xi = 
\begin{pmatrix}
    \xi_1(1) \; \xi_1(2)\; \ldots \; \xi_1(k)\\
    \xi_2(1) \; \xi_2(2) \; \ldots\; \xi_2(k)\\
    \vdots\\
    \xi_t(1)\; \xi_t(2)\; \ldots  \; \xi_t(k)\\
\end{pmatrix}.
$$
By \cref{cor:marginaldistributionofwalk}, every two columns 
of this matrix are uniform independent over $[hq]^{t}.$ Hence, the probability that two given rows coincide is $(1/hq)^{-t} = 1/n.$ If $k = o(\sqrt{n}),$ by union bound over all pairs, it follows that the probability two rows coincide is at most $k^2/n = o(1).$

Lastly, 
if $b\sim \unif((\zqz)^{ht}),$ then 
 $b_1, b_2, \ldots, b_t \iidsim \unif((\zqz)^{h}).$ Thus, by \cref{lemma:worddecodingguarantee}, the coordinates of $\xi$ are iid numbers distributed uniformly over $[hq].$
Conditioned on $\xi(1),\xi(2), \ldots, \xi(k)$ being pairwise distinct,  $\xi(1),\xi(2), \ldots, \xi(k)$ is a uniformly random subset of size $k$ of $[hq]^t$.
\end{proof}

\begin{remark}
Both \cref{alg:sentencedecoding} and \cref{alg:fullreduction} are trivially parallelizable. In \cref{alg:sentencedecoding}, one can fully parallelize the decoding of individual words. In \cref{alg:fullreduction}, one can fully parallelize the decoding of separate samples. 
\end{remark}

\subsection{The Case of Arbitrary Distributions on the Support}
\label{sec:complete_arbitrary_rows}

We next describe how to prove~\cref{thm:decoding_matrix_and_alg} in the case of arbitrary distributions $\cD$ on the support.
Our goal is to implement a word-shift transformation that takes a pair $(a,b)$ such that $a^\top G = b^\top$, where $a$ is uniform over $k$-sparse  binary vectors and produces $a'$ such that $(a')^\top G = b^\top$ and $a'\in \zqz$ has exactly $k$ non-zero values. These values are, furthermore, distributed according to $\mathcal{D}$ and independent of the location of the support of $a.$
Note that this would directly imply the result, by adding it after the last step in~\cref{alg:sentencedecoding}.

\begin{lemma}
    \label{lem:word_shift}
    Let $\cD$ be an arbitrary distribution on $\zqmult^k$ from which we can sample in time $\Tsample$. 
    There exists an algorithm (\cref{alg:changingsupportdistribution} below) satisfying
    \begin{enumerate}
        \item Given $\xi \in [hq]^k$, \cref{alg:changingsupportdistribution} runs in time $O(\Tsample +\poly(k,h,t,\log q))$ and outputs $\xi' \in [(hq)^t]^k$ and $\rho \in (\zqz \setminus \set{0})^k$ such that
        \[
            \sum_{\psi=1}^k \rho_\psi \cdot G_{\xi'(\psi),:} =\sum_{\psi=1}^k G_{\xi(\psi),:} \,.
        \]
        \item If $\xi \sim \unif([(hq)^t]^k)$, then $\xi' \sim \unif([(hq)^t]^k), \rho \sim \cD$ and the two are independent.
    \end{enumerate}
\end{lemma}
We use the following word-shift operation.
Recall that we can identify $\xi_j(\psi) \in [hq]$ with $(\Delta_j(\psi), \kappa_j(\psi)) \in [h]\times \{0,1,2\ldots, q-1\}$ via the transformation $\xi_j(\psi) = h \cdot \kappa_j(\psi) + \Delta_j(\psi)$.
\begin{definition}[Word-Shift Transformation] 
Let $\xi = (\xi(1), \xi(2), \ldots, \xi(k))\in [(hq)^t]^k$, where\linebreak $\xi(\psi) = (\xi_1(\psi), \xi_2(\psi), \ldots, \xi_t(\psi))\in [hq]^t$.
Given $\rho \in \zqmult^k$, define $T_\rho[\xi] \in [(hq)^t]^k$ as follows:
\begin{enumerate}
    \item For each $\psi \in [k]$ and $j \in [t]$ and $\xi_j(\psi) = (\Delta_j(\psi), \kappa_j(\psi))$, define $\xi_j'(\psi) = (\Delta_j(\psi), \rho_\psi^{-1}\kappa_j(\psi))$, where arithmetic in the second coordinate is $\pmod{q}$. Note that $\rho_\psi^{-1}$ is well-defined as it is a residue in $\zqmult.$
    \item Let $T_\rho[\xi]$ be obtained by applying this transformation to all $\xi_j(\psi)$.
\end{enumerate}
\end{definition}

Note that $T_\rho$ is bijective for every $\rho$ since the operation performed on all $\xi_j(\psi)$ is bijective.
Further, it can be computed in time $\poly(k,h,\log q)$ by doing the arithmetic over each of the $k$ coordinates separately.
The key property of the word-shift operation is the following. 
\begin{observation}
\label{obs:transformationforsupport}
Recall $G$ from \cref{eq:gdefintro}.
For any $\xi \in [(hq)^t]^k, \rho \in (\zqz \setminus \set{0})^k$ and $\xi' = T_\rho[\xi]$, it holds that
\[
    \sum_{\psi=1}^k \rho_\psi \cdot G_{\xi'(\psi),:} =\sum_{\psi=1}^k G_{\xi(\psi),:} \,.
\]
\end{observation}
\begin{proof}
         \begin{equation}
     \begin{split}
          \sum_{\psi=1}^k \rho_\psi \cdot G_{\xi'(\psi),:} 
         & = \sum_{\psi= 1}^k
         \rho_\psi G_{(\xi'_1(\psi), \xi'_2(\psi), \ldots, \xi'_t(\psi)), :}\\
         & = \sum_{\psi= 1}^k
         \rho_\psi (V_{\xi'_1(\psi),:},V_{\xi'_2(\psi),:},, \ldots, V_{\xi'_t(\psi),:})\\
         & = \sum_{\psi= 1}^k
         \rho_\psi \Big(V_{\Delta_1(\psi) + h\kappa_1'(\psi),:},V_{\Delta_2(\psi) + h\kappa_2'(\psi),:}, \ldots, V_{\Delta_t(\psi) + h\kappa_t'(\psi),:}\Big)\\
         & = \sum_{\psi= 1}^k
         \rho_\psi \Big(\kappa_1'(\psi)e^\top_{\Delta_1(\psi)},\kappa_2'(\psi)e^\top_{\Delta_2(\psi)}, \ldots, \kappa_t'(\psi)e^\top_{\Delta_t(\psi)}\Big)\\
         & = 
         \sum_{\psi= 1}^k \Big(\rho_\psi\kappa_1'(\psi)e^\top_{\Delta_1(\psi)},\rho_\psi\kappa_2'(\psi)e^\top_{\Delta_2(\psi)}, \ldots, \rho_\psi\kappa_t'(\psi)e^\top_{\Delta_t(\psi)}\Big)\\
         & \equiv_q 
         \sum_{\psi= 1}^k\Big(\kappa_1(\psi)e^\top_{\Delta_1(\psi)},\kappa_2(\psi)e^\top_{\Delta_2(\psi)}, \ldots, \kappa_t(\psi)e^\top_{\Delta_t(\psi)}\Big)\\
         & = \sum_{\psi= 1}^k
         \Big(V_{\Delta_1(\psi) + h\kappa_1(\psi)},V_{\Delta_2(\psi) + h\kappa_2(\psi)}, \ldots, V_{\Delta_t(\psi) + h\kappa_t(\psi)}\Big)\\
         & = \sum_{\psi= 1}^k
          (V_{\xi_1(\psi),:},V_{\xi_2(\psi),:}, \ldots, V_{\xi_t(\psi),:})\\
          & = \sum_{\psi = 1}^k G_{\xi(\psi),:}.\qedhere
         \end{split}
     \end{equation}
\end{proof}

This immediately implies that the following algorithm satisfies the first property of~\cref{lem:word_shift} (the run time follows since we can sample from $\cD$ in time $\Tsample$).
\begin{algorithm}
    \caption{Fixing Distribution on Support of Rows}
    \label{alg:changingsupportdistribution}
     \hspace*{\algorithmicindent} \textbf{Parameters:} A distribution $\mathcal{D}$ over $(\mathbb{Z}\backslash\{0\})^k$.
     \hspace*{\algorithmicindent} \textbf{ }\\
     \hspace*{\algorithmicindent} \textbf{Input:} $\xi \in [(hq)^t)]^k$ 
     \hspace*{\algorithmicindent} \textbf{ }\\
     \hspace*{\algorithmicindent} \textbf{Procedure:}\\
     \hspace*{\algorithmicindent} \textbf{1.} Sample $\rho \sim \cD$.\\
     \hspace*{\algorithmicindent} 
     \textbf{2. }Output $\xi' = T_\rho[\xi].$
\end{algorithm}


It remains to show that the distributions are transformed correctly.
Assume $\xi = (\xi(1), \ldots, \xi(k)) \sim \unif([(hq)^t]^k)$.
Note that $\rho \sim \cD$ by construction.
Hence, the distribution on the support is correct.
Now, we need to show that $\xi' \sim \unif([(hq)^t]^k)$ as well and that it is independent from $\rho$.
This follows since $\xi$ is independent from $\rho$ and the operation $T_\rho$ is a bijection. 

\paragraph{Modifying~\cref{alg:sentencedecoding} and proof of~\cref{thm:decoding_matrix_and_alg}}

We change~\cref{alg:sentencedecoding} as follows:
Let $\xi \in [(hq)^t]^k$ be the $k$-tuple computed in Step 3 and suppose that the algorithm does not terminate on Step 4.
Run~\cref{alg:changingsupportdistribution} with input $\xi$ and let $\rho$ and $\xi'$ be its output.
In Step 5,
output instead $a \in (\zqz)^{(hq)^t}$ such that $a_{\xi'(\psi)} = \rho_\psi$ for $1\leq \psi \leq k$ and 0 otherwise.

Combining the first parts of~\cref{lemma:sentencedecodingguarantee,lem:word_shift} it follows that this modification runs in time $O(\Tsample +\poly(k,h,t,\log q))$ and its output $a$ is $k$-sparse and satisfies $a^\top G = b^\top$.
From the proof of~\cref{lemma:sentencedecodingguarantee} we know that if the input $b \sim \unif((\zqz)^L)$, then $\xi \sim \unif{[(hq)^t]^k}$.
By~\cref{lem:word_shift} it follows that $\xi' \sim \unif{[(hq)^t]^k}$ as well, $\rho \sim \cD$ and the two are independent.
It follows that the support of $a$ has size exactly $k$ and is uniformly at random and independent of $\rho$.
On its support it takes values $\rho$, thus, $a \sim \randomsupport(n,k,\cD)$. Note that probability of failure is the same as in the binary case as the additional word-shifting operation never fails.

With this, the proof of \cref{thm:decoding_matrix_and_alg} is complete.

\subsection{Variable Support Sizes}
\label{sec:variablesupportsizes}
In certain applications, one might be interested in sparse noisy linear equations where the support size varies between different samples.
This, for example, is the case in the instance of \cite{jain24sparseLWE} since the distribution of entries on the $k$ selected coordinates is $\unif(\zqz)$ and, hence, some entries might be zero.
Similarly, one may consider the ``binomial sparsity'' setting of LPN where each entry is independently $1$ with probability $k/n$.

There is a simple reduction from the variable-size support to the fixed-size support setting by keeping only the most frequent support size, thus losing at most a $1/n$ fraction of the samples. 
Yet, a reduction in opposite direction seems less trivial.

Nevertheless, our reduction can also easily accommodate the setting of supports of varying size. Specifically, suppose that $\mathcal{K}$ is some distribution over $\{0,1,2,\ldots n\}$ and $(\mathcal{D}_i)_{i = 0}^n$ is a family of distributions of values on the support, where $\mathcal{D}_i$ is a distribution over $\zqmult^i.$ Then, the respective problem denoted by 
$$
\NLE(\modulus = q,\Dcoeff=\randomsupport(n, \mathcal{K}, (\mathcal{D}_i)_{i = 0}^n), \Derr)
$$
with $m$ samples
is defined as follows. Independently for each $i \in [m],$ one first draws $k_i\sim \mathcal{K},$ then draws a uniformly random set of $k_i$ distinct indices $j_1, j_2, \ldots, j_{k_i}\in [n].$ Then, draws $\rho\sim \mathcal{D}_{k_i}, \rho \in \zqmult^{k_i}$ and sets $a\in (\zqz)^n$ such that $a_{j_\ell} = \rho_{\ell}$ for each $\ell \in [k_i]$ and $a_\xi = 0$ otherwise. Then, one draws $s\sim \unif((\zqz)^n).$  Finally, one draws an error $e_i \sim \Derr(\langle a_i,s\rangle)$ and computes the sample $(a_i, y_i)$ where $y_i = \langle a_i, s\rangle + e_i.$ The corresponding $\hzero$ model
is defined similarly by drawing the samples $a_i$ as described and the labels $y_1, y_2, \ldots, y_m\iidsim\unif(\zqz)$
The public matrix and labels are $(A,y).$

\paragraph{How to Modify Our Reduction.}

We modify our reduction to accommodate for varying support sizes as follows:
Suppose that the distribution of support sizes is $\mathcal{K}$ where with probability $1 - \eta$ with $\eta = o(1/m)$ the support size is at least $k_\minsf$ and at most $o(\sqrt{n}).$ Then, for each sample $b_i,$ we first draw $k_i\sim \mathcal{S}$ and then use the word-decoding and sentence-decoding procedures for $k_i.$ Note that if the conditions of our main primitives \cref{lem:samplingintro,lem:sizeestimationintro,lem:nearuniformityintro} hold for $k_\minsf,$ hey also hold for any $k \ge k_\minsf.$ With high probability, we never see $k_i<k_\minsf$ or 
$k_i = \Omega(\sqrt{n}).$ We illustrate with two examples.

\paragraph{Example 1: The setting of binomial support.} Suppose that $\mathcal{K} = \mathsf{Binomial}(k/n,n).$ This is the setting in which each entry is independently in the support with probability $k/n.$ If $k = \omega(\log n),$ Then, by simple Chernoff bound, with probability $1 - 2\exp(- k/12),$ it is the case that $k'\sim \mathcal{K}$ satisfies $k'\in [k/2, 2k].$ Hence, one can get a high-probability reduction for $m = o(\exp(k/12))$ samples. 

Note that up to the factor $12$ in the exponent, this behaviour is tight. Specifically, in the case of binomial support with probability $k(1- k/n)^{n-1} = \exp(-\Theta(k)),$ a sample has a single non-zero coordinate. Thus, with $\poly(n)\times\exp(\Theta(k)) = \exp(O(k))$ samples, for each $i \in [n],$ there exist $\poly(n)$ samples $a_j$ such that $a_j$ has $i$ 
as its unique non-zero coordinate. This makes detection and search trivial. 

\paragraph{Example 2: The setting of \cite{jain24sparseLWE}.}
Similarly, in the set-up of 
\cite{jain24sparseLWE}, the authors study 
sparse $a_i$ with uniform ``support'' of size $k,$ in which each value on the ``support'' is distributed as $\unif(\zqz).$ 
Equivalently, when $q$ is a prime, this is the distribution $\randomsupport(n,\mathcal{K},(\mathcal{D}_i)_{i = 0}^k)$ in which $k\sim \mathcal{K}$ is distributed as $\mathsf{Binomial}(1-1/q,k)$
and $\mathcal{D}_i = \unif((\zqz)\backslash\{0\})^i) = \unif(\zqmult^i).$ 

Again, note that with probability $k(1-1/q)(1/q)^k = \exp(-\Theta(k\log q)),$ a sample has a single non-zero coordinate. Thus, the largest number of samples for which one may expect computational hardness is $\exp(O(k\log q)).$ 
At the same time, with probability $1-2\exp(-\Omega(k\log q)),$ $k'\sim \mathcal{K},$
it is the case that $k'\in [k/2, k].$ This shows that our reduction is  tight up to the exact constant in the exponent.

\subsection{Decision, Search, and Refutation in Our Reduction}
\label{sec:reductionandtasks}
We now show that solving one of the decision, search, refutation problems on the sparse instance solves the same problem on the dense instance.
We begin by noting that with high probability our reduction preserves almost all of the input samples.

\begin{lemma} Suppose that the input to \cref{alg:fullreduction} consists of $m$ samples $\{(b_i, y_i)\}_{i = 1}^m$ such that $b_i \in (\zqz)^L, y_i \in \zqz.$ Let the number of samples in the output be $m_1.$ Then, with probability at least $1-\exp(- 2mL^{-2})$, $m_1\ge m(1- 3/L - k^2/n).$
\end{lemma}
\begin{proof} For each sample $b_i,$ let $X_i$ be the indicator that \cref{alg:sentencedecoding} does not fail on input $b_i.$ The variables $X_1, X_2, \ldots, X_m$ are independent Bernoulli random variables with $X_i\sim \Bern(p_i)$ for $p_i\ge 1- 2/L - k^2/n$ by \cref{lemma:sentencedecodingguarantee}. Hence, $\expect[\sum {X_i}]\ge m(1 - 2/L - k^2/n).$ As each $X_i$ is bounded in $[0,1],$ it is $1/4$-subgaussian. Hence,  $\sum {X_i}$ is $m/4$-subgaussian and so it deviates from its mean by $m/L$ with probability at most $\exp((mL^{-1})^2/(m/2)).$     
\end{proof}

In order to reduce to a fixed-size setting, our reduction will simply fail if there are less than $m_1\coloneqq \lfloor m(1 - 3/L - k^2/n)\rfloor $ successfully decoded samples and keep the first $\lfloor m(1 - 3/L - k^2/n)\rfloor $ samples in case there are at least that many samples.

\paragraph{Decision.} We observe that \cref{alg:fullreduction} maps 
$$
\hone\; : \;\NLE(\samples = m, \modulus = q, \Dcoeff = \unif((\zqz)^L),\Derr)
$$
to 
$$
\hone'\; : \;\NLE(\samples = m_1, \modulus = q, \Dcoeff = \randomsupport(n, k,\mathcal{D}),\Derr)
$$
and
$$
\hzero\; : \;  b_1,b_2,\ldots, b_m\iidsim\unif((\zqz)^L), y_1, y_2, \ldots, y_m\iidsim\unif(\zqz)
$$
to
$$
\hzero'\; : \; 
a_1,a_2,\ldots, a_{m_1}\iidsim\randomsupport(n,k,\mathcal{D}), y_1, y_2, \ldots, y_{m_1}\iidsim\unif(\zqz).
$$
Indeed, by \cref{thm:decoding_matrix_and_alg}, a sample $b_i\sim \unif(\mathbb{Z}^L_q)$ is mapped to a sample $a_i\sim\randomsupport(n, k,\mathcal{D}).$ Hence, the distribution of the samples is correct both under $\hzero,\hone.$ We need to also analyze the labels.

In the $\hzero$ case, we have that
$y \coloneqq (y_1, y_2,\ldots ,y_m)\sim \unif(\mathbb{Z}^m_q)$ and is independent of $b_1,b_2, \ldots, b_m.$ Let $y'\in \mathbb{Z}^{m_1}_q$ be the subset of labels corresponding to successfully decoded $a_i.$
As the success of $\dec_V$ is independent of labels, the vector $y'$ is independent of $a_1, a_2, \ldots, a_{m_1}$ and uniform over $\mathbb{Z}^{m_1}_q.$ Hence, so is $y' + Az.$ Thus, $\hzero$ is indeed mapped to $\hzero'.$

In the $\hone$ case, 
$y_i = b_i^\top s+ e_i = a_i^\top Gs + e_i = a_i^\top (Gs) + e_i.$ Hence, the final label vector $y'$ is
$A(Gs + z) +e'.$ Again, note that each entry of $e'$ is independent from $A$ and distributed according to $\Derr.$ Furthermore,
as $z\sim \unif((\zqz)^n)$,
$Gs + z\sim \unif((\zqz)^n)$ and is independent of $A,e.$ Hence, $\hone$ is indeed mapped to $\hone'.$ We record these observations as follows:

\begin{corollary}[Testing Reduction]
\label{cor:testingreduction}
Assume that the conditions in \cref{parameterdependenciess} hold and let $\cD$ be an arbitrary distribution over $\zqmult^k$ from which we can sample in time $\Tsample$.
Suppose that there exists an algorithm which solves in time $\mathcal{T}$ with advantage $\epsilon$ the following sparse testing problem with $m_1\coloneqq \lfloor m(1 - 3/L - k^2/n)\rfloor$ samples:
\begin{align*}
    & \hzero'\; : \; 
    a_1,a_2,\ldots, a_{m_1}\iidsim\randomsupport(n,k,\mathcal{D}), y_1, y_2, \ldots, y_m\iidsim\unif(\zqz),\\
    & \hone'\; : \;\NLE(\samples = m_1, \modulus = q, \Dcoeff = \randomsupport(n, k,\mathcal{D}),\Derr).
\end{align*}
Then, there exists an 
algorithm which solves in time $\mathcal{T} + \poly(n,m,\Tsample)$ with advantage $\epsilon- \exp(- 2mL^{-2})$ the following standard testing problem with $m$ samples:
\begin{align*}
    & \hzero\; : \; 
    b_1,b_2,\ldots, b_m\iidsim\unif((\zqz)^L), y_1, y_2, \ldots, y_m\iidsim\unif(\zqz),\\
    & \hone\; : \;\NLE(\samples = m, \modulus = q, \Dcoeff = \unif((\zqz)^L),\Derr).\\
\end{align*}
\end{corollary}
\paragraph{Search.} As in the case of testing, $\hone$ is mapped to $\hone'$ with the new secret $s' = Gs + z.$ Note, however, that $z$ is chosen by the reduction \cref{alg:fullreduction} and is, hence, known. Therefore, if one finds $s' = Gs + z,$ one can infer $Gs.$ Finally, as $G$ is full rank according to \cref{lem:Ginvertibility}, one can infer $s$ from $Gs.$ We record as follows:

\begin{corollary}[Search Reduction]
\label{cor:searchreduction}
Assume that the conditions in \cref{parameterdependenciess} hold and let $\cD$ be an arbitrary distribution over $\zqmult^k$ from which we can sample in time $\Tsample$.
Suppose that there exists an algorithm which solves in time $\mathcal{T}$ with success probability $p$ the following sparse search problem with $m_1\coloneqq \lfloor m(1- 3/L - k^2/n)\rfloor$ samples:
\begin{align*}
    \NLE(\samples = m_1, \modulus = q, \Dcoeff = \randomsupport(n, k,\mathcal{D}),\Derr).
\end{align*}
Then, there exists an 
algorithm which solves in time $\mathcal{T} + \poly(n,m,\Tsample)$ with probability $p- \exp(- 2mL^{-2})$ the following standard search problem with $m$ samples:
\begin{align*}
    \NLE(\samples = m, \modulus = q, \Dcoeff = \unif((\zqz)^L),\Derr).
\end{align*}
\end{corollary}

\paragraph{Refutation.} Finally, consider the refutation problem \cref{def:refutation}. Observe that for any $b$ that is successfully decoded into $a$ such that $b^\top  = a^\top G,$ any $y\in \zqz$ and any $s\in (\zqz)^L,$ it is the case that 
$$
\mu(b^\top s -y) = 
\mu(a^\top Gs - y) = 
\mu\big(a^\top (Gs+z) - (y + a^\top z)\big) .
$$
In particular, this means that for any fixed secret $s,$ the transformation between dense and sparse linear equations preserves the weight $\mu(\cdot).$ As $\mu(\cdot)$ takes values in $[0,1],$ this means that if $a_1, a_2, \ldots, a_{m_1}$ are the output sparse vectors and $y'_1, y'_2, \ldots, y'_{m_1}$ the corresponding output labels, for any secret $s,$
$$
(m-m_1) + \sum_{i=1}^{m_1} 
\mu\big(a_i^\top (Gs+z) - y'_i\big)\ge 
\sum_{i = 1}^m 
\mu(b_i^\top s -y_i)\ge 
\sum_{i=1}^{m_1} 
\mu\big(a_i^\top (Gs+z) - y'_i\big).
$$
Recalling that $m_1= \lfloor m(1- 3/L - k^2/n)\rfloor,$ 
\begin{align*}
& 4/L+ 2k^2/n+
\frac{1}{m_1}
\sum_{i=1}^{m_1} 
\mu\big(a_i^\top (Gs+z) - y'_i\big)
\ge 
(1+4/L + 2k^2/n)
\frac{1}{m}
\sum_{i = 1}^m 
\mu(b_i^\top s -y_i)\\
&\ge 
\frac{1}{m_1}
\sum_{i=1}^{m_1} 
\mu\big(a_i^\top (Gs+z) - y'_i\big).
\end{align*}
In particular, this means that a $\delta$-\textsf{APPROXIMATELY SATISFIABLE} instance is mapped to a $\delta(1+4/L+ 2k^2/n)$-\textsf{APPROXIMATELY SATISFIABLE} instance whenever \cref{alg:fullreduction} produces at least $m_1$ samples, 
which happens with probability at least $1-\exp(-2mL^{-2}).$ 

Furthermore, a random dense instance $\hzero$ is mapped to a random sparse instance $\hzero'$ with $m_1= \lfloor m(1- 3/L - k^2/n)\rfloor$ samples with probability at least $1-\exp(-2mL^{-2}).$ Again, if the starting dense instance is $\Delta$-\textsf{UNSATISFIABLE}, then the target sparse dense instance is $(\Delta+4/L + 2k^2/n)$-\textsf{UNSATISFIABLE}.

Altogether, this means that under a very small blow-up of parameters, one can use a refutation algorithm for the sparse instance combined with our reduction to refute the dense instance (outputting \textsf{APPROXIMATELY SATISFIABLE} whenever \cref{alg:fullreduction} fails to produce $m_1$ samples).

\begin{corollary}[Refutation Reduction]
\label{cor:refutationreduction}
Assume that the conditions in \cref{parameterdependenciess} hold and let $\cD$ be an arbitrary distribution over $\zqmult^k$ from which we can sample in time $\Tsample$.
Suppose that there exists an algorithm which $(\mu,\delta,\Delta,p)$-refutes in time $\mathcal{T}$ the following sparse random instance with $m_1\coloneqq \lfloor m(1- 3/L - k^2/n)\rfloor$ samples:
\begin{align*}
    a_1,a_2,\ldots, a_{m_1}\iidsim\randomsupport(n,k,\mathcal{D}), y_1, y_2, \ldots, y_m\iidsim\unif(\zqz).
\end{align*}
Then, there exists an 
algorithm which 
$(\mu,\delta(1 - 3/L - 2k^2/n), \Delta + 3/L + 2k^2/n, p- \exp(-2mL^{-2}))$-refutes
in time $\mathcal{T} + \poly(n,m,\Tsample)$ the following standard random instance with $m$ samples:
\begin{align*}
    b_1,b_2,\ldots, b_m\iidsim\unif((\zqz)^L), y_1, y_2, \ldots, y_m\iidsim\unif(\zqz).\\
\end{align*}
\end{corollary}

\section{Certain Parameter Instantiations}
\label{sec:parameters}

Recall that we denote the time to produce a sample from $\mathcal{D}$ by $\Tsample$.
Our reductions \cref{cor:testingreduction,cor:searchreduction,cor:refutationreduction} run in time $\mathcal{T} = \poly(n,m,\Tsample).$ We will assume that $\Tsample = \poly(n),$ so the entire reduction runs in time $\poly(n,m).$
Hence, they imply hardness only if we assume that 
the starting standard (dense) problem  is hard in time  $\poly(n)$.
We show how to choose our parameters based on standard assumptions about how much time is needed for the dense problem.
Recall that they need to satisfy \cref{parameterdependenciess}, restated below.
\restateparameters{parameterdependenciess}

In these parameter relations, $t = \frac{\log n}{\log h + \log q}$ and $L = \frac{h\log n}{\log h + \log q}.$ As the assumed hardness of standard problems increases with the dimension $L,$ we will aim to choose the largest $L$ possible. Clearly, 
$h\longrightarrow \frac{h\log n}{\log h + \log q}$ is increasing, so we will try to choose the largest $h$ possible that satisfies 
\cref{parameterdependenciess}.


\medskip

In \lpn, we have $q = 2$ and in \lwe{}, a typical setting is $q = L^{O(1)} = (ht)^{O(1)}.$ We show that we can choose the following dimension in that case.

\begin{observation}
\label{eq:choiceofL}
If $q \le L^{\kappa}$ for some absolute constant $\kappa,$ we can choose $$L = \Theta\Bigg(\frac{k\log n}{\log k (\log k + \log \log n)}\Bigg).$$    
\end{observation}
\begin{proof}
In that case the inequality $k \ge 4h(\log h + \log q + 2\log(ht))$ is satisfied if $k \ge (5\kappa + 5)(h \log h + h \log t) = \Theta(h\log h + h\log t).$

Choose $h = \lceil\frac{k}{((5\kappa + 5))^2(\log k + \log\log n)}\rceil.$ Hence, $t = \frac{\log n}{\log h} = \Theta(\frac{\log n}{\log k}).$ Thus,
\begin{align*}
& 4h(\log h+ \log q  + 2(\log ht)) \le 
(5\kappa +5)(h\log h  +  h \log t)\\
& \le 
(5\kappa +5) \lceil\frac{k}{(5\kappa +5)^2(\log k + \log\log n)}\rceil\log k\\
& \quad\quad\quad\quad\quad+ 
(5\kappa +5)\lceil\frac{k}{(5\kappa +5)^2(\log k + \log\log n)}\rceil\log \log n\le 
k.
\end{align*}
It follows that $
L = ht = \Theta(\frac{k\log n}{\log k (\log k + \log \log n)}).$
\end{proof}

\begin{corollary}[Exponential Hardness of \textbf{LWE}] 
\label{cor:explwe}
Assume \ref{itm:nearexplwe}: standard testing (respectively, search, refutation) for
$$\LWE(\samples = m,\modulus = q,\Dcoeff = \unif(\mathbb{Z}_q^L),\Derr)$$ is hard to solve with advantage $p$ (respectively, probability of success) in time $2^{o(L)}$ with access to $m$ samples for some $q = L^{O(1)}.$

Let $k,n$ be such that $k = \omega\big((\log \log n)(\log \log \log n)\big), k = o(\sqrt{n}).$ Let $\mathcal{D}$ be any distribution over $\zqmult^k$ from which we can sample in $\poly(n)$ time.

Then, 
testing (respectively, search and refutation) for 
$$
\LWE(\samples = m_1, \modulus = q, \Dcoeff = \randomsupport(n,k,\mathcal{D}),\Derr)
$$
 is hard to solve with advantage (respectively, probability of success) $p + \exp(-m/(k\log n)^2)$
 in time $2^{o(\frac{k\log n}{\log k(\log k + \log \log n)})}$   with access to $m_1 = m(1 - O(1/k + k^2/n))$ samples for $q = (k\log n)^{O(1)}.$
\end{corollary}
\begin{proof} In \cref{eq:choiceofL}, $L = \Theta\Big(\frac{k\log n}{\log k (\log k + \log \log n)}\Big).$ We need
$2^{L} = \omega(\poly(n))$ for the assumed hardness to be stronger than the time of the reduction. This yields
$L = \omega(\log n)$ or $k = \omega((\log \log n)( \log \log\log n)).$
\end{proof}

\begin{corollary}[Near-Exponential Hardness of \textbf{LPN}] 
\label{cor:nearexplpn}
Assume \ref{itm:nearexplpn}: standard testing (respectively, search, refutation) for
$$\LPN(\samples = m,\modulus = 2,\Dcoeff = \unif(\mathbb{Z}_2^L),\Derr)$$ is hard to solve with advantage $p$ (respectively, probability of success) in time $2^{o(L/\log L)}$ with access to $m$ samples. 

Let $k,n$ be such that $k = \omega\big((\log \log n)^2(\log \log \log n)\big), k = o(\sqrt{n}).$

Then, 
testing (respectively, search and refutation) for 
$$
\LPN(\samples = m_1, \modulus = 2, \Dcoeff = \randomsupport(n,k,1^k),\Derr)
$$
 is hard to solve with advantage (respectively, probability of success) $p + \exp(-m/(k\log n)^2)$
 in time $2^{o(\frac{k\log n}{\log k(\log k + \log \log n)^2})}$   with access to $m_1 = m(1 - O(1/k + k^2/n))$ samples.
\end{corollary}
\begin{proof} In \cref{eq:choiceofL}, $L/\log L = \Theta\Big(\frac{k\log n}{\log k (\log k + \log \log n)^2}\Big).$
\end{proof}


\begin{corollary}[Sub-Exponential Hardness of\textbf{ LPN/LWE}] 
\label{cor:subexphardness}
Assume \ref{itm:subexplpn} or \ref{itm:subexplwe}: standard testing (respectively, search, refutation) for
$$\LPN(\samples = m,\modulus = 2,\Dcoeff = \unif(\mathbb{Z}_2^L),\Derr)$$ 
(respectively, \lwe{} for $q = L^{O(1)}$)
is hard to solve with advantage $p$ (respectively, probability of success) in time $2^{L^{1-\alpha}}$ with access to $m$ samples.

Let $k,n$ be such that $k = \omega( (\log n)^{\alpha/(1-\alpha)}(\log \log n)^2), k = o( \sqrt{n}).$ Let $\mathcal{D}$ be any distribution over $\zqmult^k$ from which we can sample in $\poly(n)$ time.

Then, testing (respectively, search and refutation), for 
$$
\LPN(\samples = m_1, \modulus = 2, \Dcoeff = \randomsupport(n,k,\mathcal{D}),\Derr)
$$
(respectively, \lwe{} with $q = (k\log n)^{O(1)}$)
 is hard to solve in time $2^{o(\frac{k^{1-\alpha}(\log n)^{1-\alpha}}{(\log k)^{1-\alpha}(\log k + \log \log n)^{1-\alpha}})}$ with $m_1=  m(1- O(1/k + k^2/n))$ samples and advantage $p + \exp(-m/(k\log n)^2).$   
\end{corollary}
\begin{proof} In \cref{eq:choiceofL},
$L^{1-\alpha} = \Theta\Big(\frac{k^{1-\alpha}(\log n)^{1-\alpha}}{(\log k)^{1-\alpha}(\log k + \log \log n)^{1-\alpha}}\Big).$ We need this quantity to be at least $\poly(n).$ Hence, $\Theta\Big(\frac{k^{1-\alpha}(\log n)^{1-\alpha}}{(\log k)^{1-\alpha}(\log k + \log \log n)^{1-\alpha}}\Big) = \omega(\log n)$ which leads to\linebreak $k = \omega( (\log n)^{\alpha/(1-\alpha)}(\log \log n)^2).$
\end{proof}

\begin{corollary}[Quasi-Polynomial Hardness of\textbf{ LPN/LWE}] 
\label{cor:quasihardness}
Assume \ref{itm:quasilpn} or \ref{itm:quasilwe}: standard testing (respectively, search, refutation) for
$$\LPN(\samples = m,\modulus = 2,\Dcoeff = \unif(\mathbb{Z}_2^L),\Derr)$$ 
(respectively, \lwe{} for $q = L^{O(1)}$)
is hard to solve with advantage $p$ (respectively, probability of success) in time $2^{(\log L)^{1+c}}$ for some constant $c>0$ with access to $m$ samples.

Let $k,n$ be such that $k = \exp(\omega( (\log n)^{1/(1+c)})), k = o( \sqrt{n}).$ Let $\mathcal{D}$ be any distribution over $\zqmult^k$ from which we can sample in $\poly(n)$ time.

Then, testing (respectively, search and refutation), for 
$$
\LPN(\samples = m_1, \modulus = 2, \Dcoeff = \randomsupport(n,k,\mathcal{D}),\Derr)
$$
(respectively, \lwe{} with $q = (k\log n)^{O(1)}$)
 is hard to solve in time $2^{o\big((\log k )^{1+c}\big)}$ with $m_1=  m(1- O(1/k + k^2/n))$ samples and advantage $p + \exp(-m/(k\log n)^2).$   
\end{corollary}
\begin{proof} In \cref{eq:choiceofL},
$(\log L )^{1+c} = \Theta\Big((\log k + \log \log n)^{1+c}\Big).$ We need this quantity to be $\omega(\log n).$ Hence, $\Theta\Big((\log k + \log \log n)^{1+c}\Big) = \omega(\log n)$ which leads to $k = \exp(\omega((\log n)^{1/(1+c)})).$
\end{proof}

\section{Hardness of Noisy Tensor Completion}
\label{sec:tensorcompletion}
\restatetheorem{thm:tensorcompletion}

We note that assuming hardness of the distinguishing problem for \lwe, we obtain hardness for the respective distinguishing problem for tensor completion and, hence, for \cref{problem:tensorcompletion}.

Our algorithm is simple and mimics the relationship between \klpn{} and noisy \kxor{} in the paragraph on lower bounds in \cref{sec:definingtensorcompletion}. It can be summarized succinctly as follows. For a complex number $z,$ denote by $\Re(z)$ its real part.

\begin{algorithm}
    \caption{Reduction from Sparse Binary Noisy Linear Equations to Tensor Completion}
    \label{alg:tensorcompletion}
     \hspace*{\algorithmicindent} \textbf{Input:} $m$ samples $(a_j, y_j = \langle a_j, s\rangle + e_j)_{j = 1}^m$ from $
     \NLE( m,q,\randomsupport(n,k,1^k), \mathcal{E}).
     $\\
     \hspace*{\algorithmicindent} \textbf{ }\\
     \hspace*{\algorithmicindent} \textbf{Procedure: }For each $j\in [m]:$\\
     \hspace*{\algorithmicindent} 
     \hspace*{\algorithmicindent}
     \textbf{1.} Compute the support $(t_j^1, t_j^2, \ldots, t_j^k)$ of $a_j.$\\
     \hspace*{\algorithmicindent} 
     \hspace*{\algorithmicindent}
     \textbf{2.} Compute $\Re(\exp(2\pi i y_j/q )).$\\
     \hspace*{\algorithmicindent} 
     \hspace*{\algorithmicindent} \textbf{ }\\
     \hspace*{\algorithmicindent} \textbf{Output: } $\Big\{(t_j^1, t_j^2, \ldots, t_j^k),\Re\big(\exp(2\pi i y_j/q )\big) \Big\}_{j= 1}^{m}.$
\end{algorithm}

\subsection{Analysis of the Reduction}

\paragraph{The Null Distribution.}
First consider the case when each sample is drawn from \linebreak $\randomsupport(n,k,1^k) \times \unif(\zqz)$.
Clearly, all entries are independent and the indices $(t_j^1, \ldots, t_j^k)$ have the correct distributions.
Further, each tensor entry we output is distributed as $\cos(2 \pi \tfrac{x}{q})$, where $x$ is uniform over $\zqz$.
Note that as $q \rightarrow \infty$, the expectation converges to $\mathbb{E}_{\phi \sim \unif([0,2\pi])} \cos(\phi) = 0$ (since the sum approximates this integral).
It follows that the mean of each entry is $o_q(1)$, since $q$ is growing with $n$, it is also $o(1) = o_n(1)$.
An analogous argument shows that the variance is at least $\mathbb{E}_{\phi \sim \unif([0,2\pi])} \cos^2(\phi) - o(1) = \Omega(1)$.

\paragraph{From Sparse Binary LWE to Complex Tensors.}
The main idea is that each $\langle a_i, s\rangle $ is a sum of $k$ uniformly random terms. This is the same as the tensor, except that we take a sum instead of product. We can fix this easily by taking the multiplicative group $\zqz$ rather than the additive, i.e. represent each $s_t$ by $\exp(2\pi i s_t/q).$ Now, complexify $s$ by  
$$
z \coloneqq \Big(\exp\Big(\frac{2\pi i s_1}{q}\Big),\exp\Big(\frac{2\pi i s_2}{q}\Big),\ldots, \exp\Big(\frac{2\pi i s_n}{q}\Big)\Big).
$$ 
Then, 
$$\exp\Big(\frac{2\pi i y_j}{q}\Big) 
= \exp\Big(\frac{2\pi i}{q}\big(e_j + \sum_{\ell = 1}^k s_{t^\ell_j}\big) )\Big) = 
\prod_{\ell = 1}^k z_{t^\ell_j} \times 
\exp\Big(\frac{2\pi i e_j}{q}\Big)\,.
$$
Now, let $\Delta_{t_j^1, t_j^2, \ldots, t_j^k} \coloneqq \prod_{\ell = 1}^kz_{t_j^\ell}\times \exp(2\pi i e_j/q)  - \prod_{\ell = 1}^kz_{t_j^\ell}.$ Note that 
\begin{equation}
\label{eq:errorcomplexifiedissmall}
\begin{split}
 |\Delta_{A_j}| &= \Big|\prod_{\ell = 1}^kz_{t_j^\ell}\times \exp(2\pi i e_j/q)  - \prod_{\ell = 1}^kz_{t_j^\ell}\Big|\\
& = 
|\exp(2\pi i e_j/q) - 1|\le 
|\cos(2\pi i e_j/q) - 1| + 
|\sin(2\pi ie_j/q)|\le 
4\pi\times \max(|e_j|, |q-e_j|)/q. 
\end{split}
\end{equation}
Hence, under this transformation, $\Big\{\big((t_j^1, t_j^2, \ldots, t_j^k),\exp(2\pi i y_j/q ) \Big\}_{j= 1}^{m}$ exactly corresponds to the entries of the 
noisy complex tensor of order $k$ and rank 1 given by
$
T_{\mathbb{C}} \coloneqq  z^{\otimes k} + \Delta.
$

\paragraph{From a Complex Tensor to a Real Tensor.} Allowing to blow-up the rank to $2^{k-1},$ we obtain a real tensor by simply taking the real part. I.e., if $z = u + iv,\Delta = \Xi + i \Psi,$ we have that 
\begin{equation}
\begin{split}
& \Re(z^{\otimes k} + \Delta)= 
\Re((u + iv)^{\otimes k} + \Xi + i \Psi)\\
& = 
\sum_{h \in \{0,1\}^k \; : \; 1^Th \equiv 0 \pmod{2}} (-1)^{1^Th/2}
(u + h_1 v)\otimes (u + h_2 v )\otimes \cdots\otimes (u + h_k v) + \Xi \eqqcolon T_{\mathbb{R}}.
\end{split}
\end{equation}

Our output will be 
$$\Big\{\big((t_j^1, t_j^2, \ldots, t_j^k),(T_{\mathbb{R}})_{t_j^1, t_j^2, \ldots, t_j^k}) \Big\}_{j= 1}^{m}.$$
What is left is to verify the incoherence and noise rate assumptions.

\paragraph{Assumption 1: Incoherence.}
Each $u_t,v_t$ for $1\le t \le n$ has the same distribution as $\cos(2\pi \tfrac x q)$ (or $\sin(2\pi \tfrac x q)$ respectively) for $x \sim \unif(\zqz)$.
Hence, using the same argument as for the null case, the random variables $u_t^2$ (respectively $v_t^2$) are bounded and have mean $\Omega(1)$. 
Thus, with exponentially high probability, $\|u\| = \Theta(\sqrt{n})$ and $\|v\| = \Theta(\sqrt{n})$.
Hence, the incoherence is $O(1)$ with exponentially high probability. 
\paragraph{Assumption 2: Error Size.}
For each $(t^1_j, t^2_j,\ldots, t^k_j),$ by \cref{eq:errorcomplexifiedissmall}, 
$$
|\Xi_{t^1_j, t^2_j,\ldots, t^k_j}|\le 
|\Delta_{t^1_j, t^2_j,\ldots, t^k_j}|\le 
4\pi\times \max(|e_j|, |q-e_j|)/q.
$$
Hence, a straightforward concentration argument shows that $$\|\Delta\|_1 = O(n^k\times \expect_{e\sim \mathcal{E}}[\max(|e|, |q - e|)]) = O(n^k\alpha)$$ with high probability. At the same time, with high probability, $\|T_{\mathbb{R}}\|_1 = \Theta(n^k).$ Altogether, it follows that $\|\Delta\|_1 / \|T_{\mathbb{R}}\|_1 = O(\alpha).$

\subsection{Extensions}

\paragraph{Complex Tensors.} We extract from our proof a reduction to complex tensor completion for $T_{\mathbb{C}}$. Complex tensor completion is analogously defined to the real tensor completion. 

\begin{theorem}
\label{thm:tensorcompletioncomplex}
There exist an algorithm running in time $O(m\log q)$ which on input
$$
\NLE(\samples = m,\modulus = q,\sampledistribution =\randomsupport(n,k,1^k),\errordistribution = \mathcal{E})
$$
produces $m$ iid sampled entries from an order-$k$ rank-$1$ complex noisy tensor $T$ with noise $\Delta.$ Furthermore, $T, \Delta$ with high probability satisfies the following assumptions: 
\begin{enumerate}
    \item \emph{Incoherence:} The tensor is $O(1)$-incoherent.
    \item \emph{Noise:} The noise $\Delta$ satisfies 
    $\|\Delta\|_1/\|T\|_1 = O(\beta),$ where 
    $\beta\coloneqq \expect_{e\sim \mathcal{E}}[\min(|e|, |q-e|)/q].$
\end{enumerate}
\end{theorem}

\paragraph{Gaussianizing the Components.} One can easily turn the components into jointly Gaussian. Namely, one first samples independent Gaussians  $g,h\iidsim \mathcal{N}(0,I_n)$ in \cref{alg:tensorcompletion}. The only difference is in the output --- it consists of the samples
$$\Big\{(t_j^1, t_j^2, \ldots, t_j^k),\quad\Re\Big(\exp(2\pi i y_j/q )\times\prod_{\ell = 1}^k(g_{t_j^\ell} + i\times h_{t_j^\ell})\Big) \Big\}_{j= 1}^{m}.$$

The rotational invariance of 
the Gaussian distribution shows that $\widetilde{z} = \widetilde{u} + i\widetilde{v},$ where $\widetilde{u}, \widetilde{v}\iidsim \mathcal{N}(0,I_n).$ In particular, this means that:
\begin{itemize}
    \item In the complex case, the single component is drawn from the standard complex Gaussian distribution.
    \item In the real case, each component is of the form $\widetilde{u} + h \widetilde{v}$ for $h \in \{0,1\}$ where $\widetilde{u}, \widetilde{v}\iidsim \mathcal{N}(0,I_n).$ Thus, the components are jointly Gaussian. 
\end{itemize}  

The noise assumption  $\|\Delta\|_1/\|T\|_1 = O(\alpha),$ continues to hold but the 
incoherence increases to $O(\sqrt{\log n})$. This does not seem to have computational or information-theoretic implications.

\paragraph{Symmetrizing the Components.} In \cite[Lemma A.2.]{barak16noisytensor}, the authors show a black-box reduction from a tensor-completion problem into a symmetric tensor completion problem in the case when $k$ is odd with only a $2^{k-1}$-factor blow-up in the dimension. A symmetric tensor is such that $u^1_j = u^2_j = \cdots = u^k_j\; \forall j.$ This reduction also applies to our setting.

\subsection{Implied Hardness}

\paragraph{Hardness from Standard LWE:} We can now immediately compose \cref{thm:tensorcompletion} with \cref{cor:explwe,cor:subexphardness,cor:quasihardness} to get hardness for \cref{problem:tensorcompletion} and obtain the results in \cref{sec:resultsontensor}. We note that in the case of LWE,
$\|\Delta\|_1/\|T\|_1 = O(\beta) = O(\sqrt{L}/q)$ as in \cite{regev05LWE}.

\paragraph{Hardness from Standard LPN:}
In the case of LPN, we simply use the equivalence in \cref{sec:definingtensorcompletion} between noisy $k$-XOR and tensor-completion.
Each entry of $T$ is $1.$ Each entry of $\Delta$ is independently $0$ with probability $1/2 +\gamma$ and $\pm 2$ with probability $1/2-\gamma.$ Hence, each entry has absolute value $1 - 2\gamma$ in expectation. With high probability, the average of $m$ entries is $1 - 2\gamma + \tilde{O}(1/\sqrt{m}).$


\newpage
\printbibliography

\appendix

\newpage

\section{Low-Degree Hardness of Sparse Linear Equations} 
\label{sec:low_degree_hardness}
  \addtocontents{toc}{\protect\setcounter{tocdepth}{1}}%
  \setcounter{tocdepth}{1}
  
  \subsection{Preliminaries on Low-Degree Polynomial Hardness}
\emph{Low-degree polynomial hardness},
emerging from high-dimensional statistics \cite{hopkins2017efficient,hopkins2017power,barak19sos,hopkinsThesis}, is an increasingly popular framework for providing evidence for hardness of statistical problems. 

The central observation is that functions that depend on  very few bits are cheap to evaluate. Hence, so are linear combinations of such functions. We formalize this notion by \emph{coordinate degree}. 
\begin{definition}[e.g.~\cite{ODonellBoolean}] We say that a function $f:\mathcal{X}^n \longrightarrow \mathbb{R}$ has coordinate degree at most $d$ if it can be represented as:
$$
f(x_1, x_2, \ldots, x_n) = 
\sum_{S\subseteq [n]\; : \; |S|\le d}
f_S((x_i)_{i\in S})
$$
\end{definition}
Assuming $\exp(\tilde{O}(d))$ evaluation time of a function in at most $d$ variables, a coordinate-degree $d$ function can be evaluated in time $\exp(\tilde{O}(d)).$

The abelian group structure of $\mathbb{Z}^n_q$ gives a nice representation of coordinate-degree-$d$ functions over $\mathbb{Z}_q^n.$

\begin{lemma}[e.g.~\cite{ODonellBoolean}]
\label{lemma:lowdegreeinzq}
Every coordinate-degree-$d$ function $f:\mathbb{Z}_q^n\longrightarrow \mathbb{R}$ has a unique representation of the form 
$$
f(x) = 
\sum_{v\in \mathbb{Z}^n_q\; : \; |\support(v)|\le d}c_v\exp\big(\frac{2\pi i \langle x,v\rangle}{q}\big)\quad\quad\quad 
\forall x\in \mathbb{Z}_q^n.
$$
\end{lemma}

One can use a low-coordinate-degree function to test between two distributions $\mathcal{P}_0$ and $\mathcal{P}_1$ if the means are sufficiently separated. Immediately implied by Chebyshev's inequality is the fact that if $|
\expect_{X\sim \mathcal{P}_0}f(X) - 
\expect_{X\sim \mathcal{P}_1}f(X)
| = \omega(\sqrt{\var_{X\sim \mathcal{P}_0}f(X) + 
\var_{X\sim \mathcal{P}_1}f(X)}),$ one can test between $H_0:\mathcal{P}_0$ and $H_1:\mathcal{P}_1$ with type I + type II error $o(1)$ by simply evaluating $f(X)$ on input $X$ and comparing to $\frac{1}{2}(\expect_{X\sim \mathcal{P}_0}f(X) + 
\expect_{X\sim \mathcal{P}_1}f(X)).$ When this condition does not hold for \emph{any} low-degree coordinate function, we say that the distinguish problem is low-degree hard. Formally, we also allow conditioning on a high-probability ``well-behaved'' event. 

\begin{definition}[e.g. \cite{dhawan2023detection}]
\label{def:ldphardness}
Suppose that $\mathcal{P}_0$ and $\mathcal{P}_1$ are two distributions over $\mathcal{X}^n.$ Suppose that there exists some event $\mathcal{E}$ such that $\prob_{X\sim \mathcal{P}_1}[X\in \mathcal{E}] = 1- o(1)$ holds and for any coordinate-degree-$d$ function $f:\mathcal{X}^n\longrightarrow \mathbb{R},$
$$|
\expect_{X\sim \mathcal{P}_0}[f(X)] - 
\expect_{X\sim \mathcal{P}_1}[f(X)|X\in \mathcal{E}]
| = o\big(\sqrt{\var_{X\sim \mathcal{P}_0}f(X) + 
\var_{X\sim \mathcal{P}_1}[f(X)|X\in \mathcal{E}]}\big)$$
holds. Then, we say that the hypothesis testing problem  $H_0:\mathcal{P}_0$ vs $H_1:\mathcal{P}_1$ is hard against coordinate-degree-$d$ polynomial tests. 
\end{definition}

Typically, hardness against degree $\omega(\log n)$ tests is viewed as evidence for hardness against polynomial-time algorithms and hardness against degree $n^C$ tests for constant $C$ is viewed as evidence of subexponential hardness. 
 


While low-degree polynomial tests certainly do not capture \emph{all} computationally efficient distinguising algorithms, they are an expressive and rich class. They capture certain common instantiations of spectral algorithms \cite{bandeira2019computational}, SQ algorithms  \cite{brennan2021statisticalquery}, and approximate message passing algorithms \cite{montanari2022equivalence} among others. Low-degree polynomial tests are also conjectured to capture all polynomial-time algorithms for a well-defined class of sufficiently noisy problems in high dimension \cite{hopkinsThesis}.

We do note that there are many problems for which there are better than low-degree polynomial algorithms and, in fact, some of them come exactly from problems on lattices and 
noisy linear equations such as the celebrated Blum-Kalai-Wasserman algorithm \cite{BKW} or more recent work on ``flat'' parallel pancakes 
\cite{zadik22latticebasedsurpassSoS,kane22ngca} based on the celebrated Lenstra-Lenstra-Lovasz algorithm \cite{Lenstra1982LLL}.

\begin{remark}[Distinguishing Advantage] In this overview we focused on advantage $o(1)$ for the distinguishing problem, which is the most common set-up in the literature on high-dimensional statistics. Recent work \cite{moitra2023preciseerrorratescomputationally} has also considered a more fine-grained analysis of the advantage. We do not engage in this endeavour since the low-degree hardness results mostly serve as evidence in the current work. 
\end{remark}

\subsection{Proof of \texorpdfstring{\cref{thm:genericLDP}}{LDPProposition}}
\label{appendix:proofofldp}

With this in hand, we can prove~\cref{thm:genericLDP} restated below.
\restateprop{thm:genericLDP}
 We begin with a simple lemma which we will use to choose a well-behaved high-probability event in \cref{def:ldphardness}.
\begin{lemma}
\label{lem:linearlyindependentsparse}
    Suppose that $m = o\Big(\frac{n^{k/2}}{k^{k/2}D^{k/2-1}e^{3k}}\Big)$ and 
    $a_1, a_2, \ldots, a_m\iidsim\randomsupport(n,k,\mathcal{D})$.
    Then, with high probability, every $D$ vectors among $a_1, a_2, \ldots, a_m$ are linearly independent. 
\end{lemma}
\begin{proof}
    Suppose that there is a set of at most $D$ vectors 
    $a_i$ which are not linearly independent.
    Let
    $S = \{i_1, i_2, \ldots, i_s\}$ be the smallest such set.
    Then, we claim that each $j \in \bigcup_{i\in S} \mathsf{support}(a_i)$ appears at least twice in the multiset $\bigcup_{i\in S} \mathsf{support}(a_i).$
    Otherwise, if some $j$ appears only once and it appears for vector $a_i,$ clearly $a_i$ is linear independent of 
    $(a_\ell)_{\ell \in S\backslash \{i\}}.$
    Hence, 
    $ S\backslash \{i\}$ is a smaller set of vectors which are not linearly independent, contradiction.

    In particular this means that $|\bigcup_{i\in S} \mathsf{support}(a_i)|\le k\times |S|/2 = ks/2.$
We bound the probability of this event:
    \begin{equation}
    \label{eq:badeventinsupportunions}
        \begin{split}
            & \prob\Big[\exists S\subseteq [m]\; : \; |S|\le D\text{ and }\abs{\bigcup_{i\in S} \mathsf{support}(a_i)} \le k\times |S|/2\Big]\\
            & \le \sum_{S\subseteq [m]\; : \; |S|\le D}
            \prob\Big[|\bigcup_{i\in S} \mathsf{support}(a_i)|\le k\times |S|/2\Big]\\
            & \le \sum_{s = 2}^D \binom{m}{s}
            \prob\Big[|\bigcup_{i = 1}^s \mathsf{support}(a_i)|\le k\times s/2\Big].
        \end{split}
    \end{equation}
We finish by bounding the quantity in the last line. Let $Z = \bigcup_{i = 1}^s \mathsf{support}(a_i).$ If $Z$ is contained in a set $Z'$ of size $ks/2,$ there are at most $\binom{n}{ks/2}$ choices for this set. Once the set $Z'$ has been chosen, each of the vectors $a_1, a_2, \ldots, a_s$ can be chosen in $\binom{ks/2}{k}$ ways. Hence, we bound the resulting probability by
$
\frac{\binom{n}{ks/2}\binom{ks/2}{k}^s}{\binom{n}{k}^s}.
$
Going back to \cref{eq:badeventinsupportunions} and using $(\frac{a}{b})^b\le \binom{a}{b}\le (\frac{ea}{b})^b,$
we bound 

\begin{align*}
    & \sum_{s = 2}^D \binom{m}{s}\frac{\binom{n}{ks/2}\binom{ks/2}{k}^s}{\binom{n}{k}^s}\\
    & \le \sum_{s = 2}^D \left(\frac{me}{s}\right)^s\times 
    \left(\frac{ne}{ks/2}\right)^{ks/2}\times
    \left(\frac{eks/2}{k}\right)^{ks}
    \left(\frac{n}{k}\right)^{-sk}\\
    & \le \sum_{s = 2}^D
    \Bigg(
    \frac{me\times (ne)^{k/2}\times (es/2)^{k}\times k^k}{s\times (ks/2)^{k/2}\times n^k}
    \Bigg)^s\\
    & \le \sum_{s = 2}^D
    \Bigg(
    \frac{m k^{k/2} s^{k/2-1}e^{3k}}{n^{k/2}}
    \Bigg)^s\\
    &\leq \sum_{s = 2}^D
    \Bigg(
    \frac{mk^{k/2} D^{k/2-1}e^{3k}}{n^{k/2 }}
    \Bigg)^s.
\end{align*}
When $m = o\Big(\frac{n^{k/2}}{k^{k/2}D^{k/2-1}e^{3k}}\Big),$ the last expression is of order $o(1).$ \end{proof}

\begin{lemma}
\label{lemma:ldphardforlienarlyindep}
Suppose that $a_1, a_2, \ldots, a_m\in \mathbb{Z}_q^n$ are fixed known vectors such that any $D$ of these vectors are linearly independent. 
Let $\mathcal{E}$ be an arbitrary distribution over $\mathbb{Z}_q$.
Then, no function (even depending on $a_1, a_2, \ldots, a_m$) 
of coordinate degree at most $D$ can distinguish between the following hypotheses:
\begin{equation}
    \begin{split}
        & \hzero: (y_1, y_2, \ldots, y_m)\iidsim\unif(\mathbb{Z}_q^m)\\
        & \hone: (y_1, y_2, \ldots, y_m) = (\langle a_1, s\rangle + e_1,\langle a_2, s\rangle + e_2,\ldots \langle a_m, s\rangle + e_m)\,,\\
        & \quad\quad\quad \text{ where }s\sim \unif(\mathbb{Z}_q^n),e_i\iidsim\mathcal{E}.
    \end{split}
\end{equation}
\end{lemma}
\begin{proof} By \cref{lemma:lowdegreeinzq},
 every coordinate-degree $D$ function $g$ of  $(y_i)_{i\in [m]}$ can be represented as 
$$
g(y_1, y_2, \ldots, y_m) = 
\sum_{\substack{S\subseteq [m]\,, \\ |S|\le D}} \;
\sum_{\substack{v \in \mathbb{Z}_q^m\,, \\ \support(v) = S}}
c_v \exp\Big(\frac{2\pi i \langle y,v\rangle }{q}\Big) \,.
$$
To show that coordinate-degree $D$ functions cannot distinguish the two hypothesis, it is enough to show that for each such $g,$
$\expect_{\mathsf{H}_0}[g] = \expect_{\mathsf{H}_1}[g]$ by \cref{def:ldphardness}.
For $S = \emptyset$, both expectations are clearly 1.
We will shot that for any $S \neq \emptyset$ and $v \in \mathbb{Z}_q^n \; : \; \support(v) = S$ it holds that
$$
\expect_{\hzero}\Big[\exp\Big(\frac{2\pi i \langle y,v\rangle }{q}\Big)\Big] = 0\text{ and }
\expect_{\hone}\Big[\exp\Big(\frac{2\pi i \langle y,v\rangle }{q}\Big)\Big] = 0 \,.
$$
The equality follows trivially for $\mathsf{H}_0.$ For $\mathsf{H}_1,$ observe that
$$
\langle y,v\rangle = 
\sum_{i\in S}y_iv_i = 
\sum_{i\in S}(\langle a_i, s\rangle + e_i)v_i = 
\langle s, \sum_{i\in S}v_i a_i\rangle + \sum_{i\in S}v_ie_i.
$$
Since $|S|\le D$ and no $D$ of the vectors $a_1, a_2, \ldots, a_m$ are linearly dependent, $\sum_{i\in S}v_i a_i\neq 0.$ Hence, 
$$
\expect_{s\sim \unif(\mathbb{Z}_q^n)}\Big[
\exp\Big(\frac{2\pi i \big(\langle s, \sum_{i\in S}v_i a_i\rangle + \sum_{i\in S}v_ie_i\big)}{q}\Big)
\Big] = 0 \,,
$$
as desired. 
\end{proof}

\paragraph{Putting It All Together.} By \cref{lem:linearlyindependentsparse}, if $m = o\Big(\frac{n^{k/2}}{k^{k/2}D^{k/2-1}e^{3k}}\Big),$ every $D$ of the $k$-sparse samples $a_1, a_2, \ldots, a_m$ are linearly independent with high-probability.
This is the high-probability well-behaved event we use to invoke \cref{def:ldphardness}.
By \cref{lemma:ldphardforlienarlyindep}, this implies hardness against 
polynomials $g$ which are degree at most $D$ in the labels $y_1, y_2, \ldots, y_m$ but depend arbitrarily on $a_1, a_2, \ldots, a_m$ (recall that $g$ can depend arbitrarily on $a_1, a_2, \ldots, a_m$ in \cref{lemma:ldphardforlienarlyindep}). In particular, this captures all polynomials in $(A,y)$ of degree at most $D.$


\end{document}